\documentclass[12pt]{article}
\usepackage{amsmath,amssymb,amsthm,amsfonts,bm,dsfont,mathrsfs,mathtools,float,algorithm,algorithmic}
\usepackage{avant,array,geometry,fontenc,inputenc,natbib,hyperref,multirow,enumerate,rotating,booktabs,color,latexsym,times,stmaryrd,ulem,xr}

\usepackage{MyArticle_KKO}

\newtheorem{theorem}{Theorem}
\newtheorem{lemma}{Lemma}

\newtheorem{proposition}{Proposition}

\newtheorem{assumption}{Assumption}
\newtheorem{remark}{Remark}

\def\E{\mathbb E}

\def\P{\mathbb P}
\def\R{\mathbb R}
\newcommand{\trans}{^{\mbox{\tiny {\sf T}}}}
\newcommand\independent{\protect\mathpalette{\protect\independenT}{\perp}}
\def\independenT#1#2{\mathrel{\rlap{$#1#2$}\mkern2mu{#1#2}}}

\def\E{\mathbb E}

\def\P{\mathbb P}
\def\R{\mathbb R}

\def\epsilonbf{\mathbf \epsilon}
\newcommand{\Acal}{{\mathcal A}}

\newcommand{\Hcal}{{\mathcal H}}

\newcommand{\Xcal}{{\mathcal X}}

\def\independenT#1#2{\mathrel{\rlap{$#1#2$}\mkern2mu{#1#2}}}

\begin{document}

\begin{center}
{\Large{\bf Kernel Knockoffs Selection for}} \\
\Large{{\bf Nonparametric Additive Models}} \\
\medskip
\medskip
\medskip

%{\large{\sc X and Y}} \\
%\medskip
%{\normalsize {\it University of Z}}
%\medskip
%\end{center}

{\large{\sc Xiaowu Dai, Xiang Lyu, and Lexin Li}} \\
\medskip
{\normalsize {\it University of California, Berkeley}}\\
\medskip
\medskip
{\normalsize To appear in  {\it Journal of the American Statistical Association: Theory and Method}}
\end{center}

\begin{footnotetext}[1]
{\textit{Address for correspondence:} Lexin Li, Department of Biostatistics and Epidemiology, University of California, Berkeley, 2121 Berkeley Way, Berkeley, CA 94720, USA. E-mail: lexinli@berkeley.edu. The first two authors contribute equally to the paper.}
\end{footnotetext}

\begin{abstract}
Thanks to its fine balance between model flexibility and interpretability, the nonparametric additive model has been widely used, and variable selection for this type of model has been frequently studied. However, none of the existing solutions can control the false discovery rate (FDR) unless the sample size tends to infinity. The knockoff framework is a recent proposal that can address this issue, but few knockoff solutions are directly applicable to nonparametric models. In this article, we propose a novel kernel knockoffs selection procedure for the nonparametric additive model. We integrate three key components: the knockoffs, the subsampling for stability, and the random feature mapping for nonparametric function approximation. We show that the proposed method is guaranteed to control the FDR for any sample size, and achieves a power that approaches one as the sample size tends to infinity. We demonstrate the efficacy of our method through intensive simulations and comparisons with the alternative solutions. Our proposal thus makes useful contributions to the methodology of nonparametric variable selection, FDR-based inference, as well as knockoffs.
\end{abstract}
\bigskip


\newpage
\baselineskip=21pt

%%%%%%%%%%%%%%%%%%%%%%%%%%%%%%%%%%%%%%%%%%%%%%%%%%%
\section{Introduction}
\label{sec:introduction}

In the past decades, the nonparametric additive model has been widely used in statistics and machine learning, thanks to its fine balance between model flexibility and model interpretability \citep{stone1985additive, hastie1990generalized, wood2017generalized}. For a univariate response variable $Y \in \R$ and $p$ predictor variables $\bX=(X_1,\ldots,X_p)\trans \in \cX^p \subseteq \R^p$, the model postulates that, 
\begin{equation} \label{eq: additivemodel}
Y = \mu+\sum_{j=1}^pf_j(X_j)+\epsilon,
\end{equation}
where $\mu$ is the intercept, $f_j : \cX \mapsto \R$, with $\E_{X_j}[f_j(X_j)]=0$, $j=1,\ldots,p$, are the component functions that are modeled nonparametrically, and $\epsilon\sim\cN(0,\sigma^2)$ is the random error, with unknown $\sigma$. Furthermore, we assume throughout this article that the component functions $f_j$'s reside in a reproducing kernel Hilbert space \citep[RKHS,][]{aronszajn1950theory, wahba1990spline}. 

Variable selection for the nonparametric additive model dates back to \citet{LZ06}, and has seen substantial developments ever since \citep[][among others]{meier2009high, RLLW09, HHW10, koltchinskii2010sparsity,wood2017generalized}. In particular, \citet{LZ06} proposed a component selection and smoothing operator (COSSO) penalty that extends the Lasso penalty to the nonparametric additive model, and penalized the sum of the reproducing kernel Hilbert space norms of the component functions. Meanwhile, \citet{meier2009high} and \citet{RLLW09} both employed basis expansion, and penalized the sparsity and smoothness seminorms. \citet{HHW10} employed the group Lasso penalty to obtain an initial estimator and to reduce the dimension of the problem, then employed the adaptive group Lasso to select nonzero components. This family of methods guarantee the asymptotic optimality of the function estimation and the selection consistency as the sample size tends to infinity. However, none has achieved the control of false discovery rate (FDR) unless the sample size tends to infinity. There has been another family of solutions that target simultaneous testing of multiple hypotheses and concentrate on controlling some forms of false discovery \citep[see, e.g.,][among others]{BH95, Efron2001, Storey2007, SunCai2007, SunCai2009}; see also \citet{CaiSun2017} for a review. Nevertheless, none of the existing solutions in this family directly addresses the problem of variable selection for the nonparametric additive model while controlling the false discovery at the same time. 

More recently, \citet{BC15} proposed a powerful framework called knockoffs that effectively controls the FDR for variable selection in the linear model under the finite-sample setting, in the sense that the sample size does not have to go to infinity. The key idea is to construct a set of so-called ``knockoff variables" that are not associated with the response conditioning on the original variables, while the structure of the knockoff variables mimics that of the original ones. It then computes an importance score for each variable, and selects those that have considerably higher scores than their knockoff counterparts. There have then been numerous generalizations of this work; see \citet{BCS20} and many references therein. Related to our target of high-dimensional nonparametric additive model, \citet{BC19} considered the high-dimensional fixed design, and focused on the linear model only. \citet{DB16} developed an ``expansion first" strategy that performs feature expansion first then constructs the knockoffs based on the expanded features, and proposed to employ a group Lasso penalty for subsequent variable selection. They actually still studied the linear regression setting; however, their proposal can, in principle, be extended to the nonparametric additive model. \cite{candes2018panning} proposed a model-X knockoffs extension that works for random designs of predictors and allows the conditional distribution of the response given the predictors to be arbitrary and unknown, though they mostly focused on the step of how to generate the knockoff variables. \cite{FDLL20} further built on the model-X knockoffs  framework, and developed a knockoffs-based variable selection procedure that is applicable to the nonparametric additive model. It employs data splitting, uses half of the data to estimate the predictor precision matrix and screen the predictors, and uses the other half to perform knockoffs based on some empirical norm of the estimated component functions. These pioneering works have opened the door for knockoffs-based selection for the nonparametric additive model. However, some may be difficult to extend beyond the linear model, and others suffer a limited power or expensive computation. In addition, there is generally a lack of theoretical power analysis for the existing knockoffs-based methods, except for \citet{FDLL20} and \citet{weinstein2020power}, who made important first steps, but only studied the power behavior for the linear model.

In this article, we propose a novel kernel knockoffs selection procedure for the nonparametric additive model \eqref{eq: additivemodel}. We build on and integrate three key components: the knockoffs, the subsampling for stability, and the random feature mapping for nonparametric function approximation in RKHS. Specifically,  we employ the random feature mapping \citep{RR07} to approximate the component function $f_j$, and construct a projection operator between the RKHS and the original predictor space. Such a projection allows us to define an analog of the effect size of the individual predictor in the setting of nonparametric additive model. We then construct the importance score based on the projected component function $f_j$, instead of the original predictor $X_j$ or its knockoff. Moreover, we note that the random features may introduce additional stochastic errors, which can disturb the order of the variables entering the model and lead to both false positives and false negatives. We thus further employ the subsampling strategy to improve the selection stability \citep{MB10}. That is, we subsample the data and apply the random feature mapping multiple times, and compute the importance score as the difference of selection frequencies over subsampling replications between each predictor and its knockoff counterpart. We show that the proposed method is guaranteed to control the FDR below the nominal level under any sample size, and achieves a power that approaches one as the sample size tends to infinity. 

Our proposal makes useful contributions to the methodology and theory of nonparametric variable selection, FDR-based inference, as well as knockoffs. 

First, whereas the methods such as \citet{LZ06, RLLW09} have obtained the variable selection consistency asymptotically, there has been no existing method that controls the FDR in the setting of nonparametric additive model for the finite-sample setting. A low FDR in such a setting assures that most of the discoveries are indeed true under any given sample size. By contrast, the asymptotic control is valid only when the sample size goes to infinity, which can be problematic for the applications with limited sample sizes. In those cases, the asymptotic control may provide little guidance on quantifying the threshold of variable selection \citep{BC15}, and it is possible that the selected variables may still include many unrelated ones \citep{su2017false}. On the other hand, the classical reproducing kernel methods usually involve non-separable variables and their knockoffs, which renders the FDR control infeasible. To address this challenge, we employ the random feature mapping to ensure the exchangeability of the null variables and their knockoffs, and in turn achieve the finite-sample FDR control. The random feature mapping nevertheless introduces an extra layer of randomness. We further resort to subsampling, construct an importance score by averaging over multiple subsampling replications, and show these techniques can handle the extra randomness; see Section A.2 of the Appendix for more technical details. In short, the finite-sample FDR control for nonparametric variable selection has been a long-standing and open question, and our proposal is among the first solutions for this type of question.

Second, we employ the subsampling strategy, but it is different from the existing subsampling based methods for FDR control. Specifically, \citet{B08} proposed a selection method based on bootstrap replications, which may suffer from a limited power when the sample size does not go to infinity.  \citet{MB10} proposed a stability-based procedure to select the variables with the selection frequencies exceeding a threshold level over the entire solution path, which guarantees the control of the expected number of false positives, but may fail to control the FDR. \citet{LHPW13} proposed to use a mixture model for the distribution of selection frequencies, whereas \citet{ahmed2011false, he2016component} suggested to estimate this distribution via permutations. However, such a distribution estimation requires either strict parametric assumptions, or  expensive computations. By contrast, our method does not require estimation of the distribution of selection frequencies, but uses subsampling to tackle the extra randomness introduced by random feature mapping and to achieve the FDR control. 

Last but not least, our method expands the scope of the currently fast growing area of knockoffs. Compared to \citet{BC19} who focused on the high-dimensional linear model only, we generalize the knockoffs to the high-dimensional nonparametric additive model. Such an extension is far from incremental, as it has to deal with nonlinear dependency between the response and predictors, as well as the non-separability of variables of the usual reproducing kernel methods. Compared to \citet{DB16} who did ``expansion first", we adopt a ``knockoffs first" strategy, which leads to an easier construction of the knockoff variables, ensures a good statistical power, and is computationally more economical. Compared to \cite{candes2018panning} who developed the model-X framework, but can not handle the nonparametric additive model directly, and did not provide any formal theoretical justification for the model-X knockoffs beyond the generalized linear model setting, we propose a new and effective importance score based on random feature mapping and resampling, and we explicitly study the power behavior in the nonparametric setting. Finally, compared to \cite{FDLL20} who extended the model-X framework, and proposed a novel importance score based on the basis functions and some empirical norm of the estimated component functions, but only studied the power for the linear model, we again develop a new importance score that is built on the selection probability of the variables and their knockoffs. As a result, we show that our solution is more powerful and also more robust to the data distribution. Moreover, we establish the power guarantee for the nonparametric additive model, which is more challenging than the linear model case as it involves nonlinear associations. Toward that end, we employ some functional data analysis techniques and concentration inequalities for functional  empirical processes to study the spectral properties; see Section A.4 of the Appendix for more technical details. We also briefly comment that, our theoretical tools are applicable to the FDR control for more general nonparametric models, e.g., the functional analysis of variance type models that involve higher-order interactions \citep{wahba1995smoothing, LZ06}. Moreover, in Section \ref{sec:numerical}, we further compare with some of these key alternative knockoff solutions numerically, and demonstrate the advantages of our proposed method empirically as well.

The rest of the article is organized as follows. Section \ref{sec:background} formulates the problem. Section \ref{sec:KKO} develops the kernel knockoffs procedure. Section \ref{sec:theory} establishes the theoretical guarantees on the FDR and power. Section \ref{sec:numerical} presents the simulations, and also an analysis of brain imaging data. The Supplementary Appendix collects all proofs and some additional numerical results.

%%%%%%%%%%%%%%%%%%%%%%%%%%%%%%%%%%%%%%%%%%%%%%%%%%%
\section{Problem Setup}
\label{sec:background}

%%%%%%%%%%%%%%%%%%%%%%%%%%%%%%%%%%%%%%%%%%%%%%%%%%%
\subsection{Kernel learning}

Throughout this article, we consider regression functions that reside in an infinite-dimensional reproducing kernel Hilbert space. We begin with a Mercer kernel $K:\mathcal X\times \mathcal X\to\mathbb R$, 
\begin{equation*}
K(X,X')=\sum_{\nu=1}^\infty\widetilde{\lambda}_\nu\widetilde{\psi}_\nu(X)\widetilde{\psi}_\nu(X'),
\end{equation*}
where $\{\widetilde{\psi}_\nu\}_{\nu=1}^\infty$ are eigenfunctions, $\{\widetilde{\lambda}_\nu\}_{\nu=1}^\infty$ are eigenvalues of the integral operator defined by the kernel function, and $\widetilde{\lambda}_\nu\widetilde{\psi}_\nu(X) = \int_\cX K(X,X')\widetilde{\psi}_\nu(X')dX'$ \citep{mercer1909functions}. 
The domain $\mathcal X\subseteq \mathbb R$ can be either a compact or an unbounded space. We consider the RKHS $\cH_1$ generated by this kernel, which is defined as the closure of linear combinations of the basis functions $\{\widetilde\psi_\nu\}_{\nu=1}^\infty$ as follows, where  $\overline{\{\cdot\}}$ denotes the closure of a function space, 
\begin{equation*}
\cH_1 = \overline{\left\{f:\cX\to\R | f(X) = \widetilde{\bPsi}(X)\trans \widetilde{\bc},\text{ and }\|f\|_K<\infty\text{ with }\|f\|_K^2=\sum_{\nu=1}^\infty\frac{\widetilde{c}_\nu^2}{\widetilde{\lambda}_\nu}\right\}}.
\end{equation*}
Here $\widetilde{\bPsi}(X)$ is an infinite-dimensional vector with the $\nu$th element equal to $\sqrt{\widetilde{\lambda}_\nu}\widetilde{\psi}_\nu(X)$, and $\widetilde{\bc}$ is an infinite-dimensional coefficient vector with the $\nu$th element $\widetilde{c}_\nu$, $\nu = 1, 2, \ldots$.

Next, define the kernel $K_p:\cX^p\times\cX^p\to\R$, 
\begin{equation*}
K_p\Big( (X_1,\ldots,X_p)\trans,(X'_1,\ldots,X_p')\trans \Big) = K(X_1,X_1') + \ldots + K(X_p,X_p'). 
\end{equation*}
The RKHS $\cH_p$ generated by $K_p$ is of the form \citep{aronszajn1950theory}, 
\begin{align*}
\cH_p = \cH_1 \oplus \ldots \oplus \cH_1 = \Big\{f : \mathcal X^p \to \mathbb R \; | \; f(\bX) = f(X_1,\ldots,X_p) = f_1(X_1) + \ldots + f_p(X_p),\\
f_j\in\mathcal H_1, \text{ and } \mathbb E[f_j(X_j)]=0, \; j=1,\ldots,p\Big\}.
\end{align*}

Suppose the observed training data $\{ (\bx_i, y_i) \}_{i=1}^{n}$ consist of $n$ i.i.d.\ copies of $(\bX, Y)$ following the nonparametric additive model \eqref{eq: additivemodel}, with $\bx_i \in \R^p, y_i \in \R$. The representer theorem \citep{wahba1990spline} shows that the solution to the kernel learning problem when restricting $f\in\mathcal H_p$, 
\begin{equation*}
\underset{f\in\mathcal H_p}{\min}\left[n^{-1} \sum_{i=1}^n\mathcal L(f(\bx_i),y_i)+\lambda\|f\|_{K_p}^2\right], 
\end{equation*}
for some loss function $\mathcal L$, the kernel $K_p$, and the penalty parameter $\lambda$, is of the form, 
\begin{equation*} \label{eqn:krrsolution} 
\widetilde{f}(\bX) = \sum_{i=1}^n\alpha_iK_p(\bX,\bx_i),
\end{equation*}
where $\balpha = (\alpha_1, \ldots, \alpha_n)\trans \in \R^n$ are the corresponding coefficients. This in effect turns an infinity-dimensional optimization problem to an optimization problem over $n$ parameters. This minimizer can be further written as, for any $\bX\in\cX^p$, 
\begin{equation} \label{eqn:representer}
\widetilde{f}(\bX) = \widetilde{\bPsi}_p(\bX)\trans \widetilde{\bc}_p,
\end{equation} 
where $\widetilde{\bPsi}_p(\bX) = \left[ \widetilde{\bPsi}(X_1)\trans,\ldots,\widetilde{\bPsi}(X_p)\trans \right]\trans$ assembles $\widetilde{\bPsi}(X_j)$'s and is an infinite-dimensional vector, and $\widetilde{\bc}_p = \Big[ \widetilde{\bPsi}_p(\bx_1),\ldots,\widetilde{\bPsi}_p(\bx_{n}) \Big] \balpha$ is the infinite-dimensional coefficient vector.

%%%%%%%%%%%%%%%%%%%%%%%%%%%%%%%%%%%%%%%%%%%%%%%%%%%
\subsection{Variable selection for nonparametric additive models}
\label{sec:vs-namodel}

Next, we formally frame variable selection in the context of nonparametric additive models. We say a variable $X_j$ is null if and only if $Y$ is independent of $X_j$ conditional on all other variables $\bX_{-j} = \{X_1,\ldots,X_p\}\backslash\{X_j\}$, i.e., $Y \independent X_j | \bX_{-j}$, and say $X_j$ is non-null otherwise \citep{LiCN2005}. Let $\cS \subseteq\{1,\ldots,p\}$ denote the indices of all the non-null variables, and $\cS^\perp \subseteq\{1,\ldots,p\}$ the indices of all the null variables, or equivalently, the complement set of $\cS$. Let $|\cdot|$ denote the cardinality, and $\widehat{\mathcal S}$ the indices of variables selected by some selection procedure. Our goal is to discover as many non-null variables as possible while controlling the FDR, which is defined as,
\begin{equation*}
\text{FDR} = \mathbb E\left[\frac{|\widehat{\mathcal S}\cap\mathcal S^\perp|}{|\widehat{\mathcal S}|\vee 1}\right].
\end{equation*} 

We next establish the identifiability of the problem under the following condition. 

\begin{assumption}[Irrepresentable Condition in RKHS]
\label{eqn:perfect-relationship}
For any $j\in\{1,\ldots,p\}$, and any functions $g_k\in\mathcal H_1, k \neq j$, $f_j(X_j) \neq \sum_{k=1; k \neq j}^p g_k(X_k)$. 
\end{assumption}

\noindent
This condition simply says that the component function $f_j(X_j)$ in model \eqref{eq: additivemodel} can not be strictly written as a linear combination of some functions of other variables $X_k, k \neq j$. This is a fairly mild condition, and its parametric counterpart that $X_j \neq \sum_{k=1;k\neq j}^p\beta_kX_k$ for any $\beta_k\in\mathbb R$ has been commonly imposed in the linear model scenario \citep{candes2018panning}.

Under this condition, we establish the equivalence between variable selection and selection of the component functions $f_j$ in model \eqref{eq: additivemodel}. In other words, testing the hypothesis that $X_j$ is null is the same as testing whether $f_j = 0$.

\begin{proposition}
\label{prop:nullequivalence}
Suppose the nonparametric additive model \eqref{eq: additivemodel} and Assumption \ref{eqn:perfect-relationship} hold. Then $j \in \mathcal S^\perp$ if and only if $f_j = 0$, for  $j=1,\ldots,p$.
\end{proposition}

\noindent
Proposition \ref{prop:nullequivalence} makes the variable selection in a nonparametric additive model comparable to that in a linear model, and is to serve as the foundation for the new kernel knockoffs procedure in Section \ref{sec:KKO}, and the finite-sample FDR control in Section \ref{sec:theory}, both of which are built upon the selection of the component functions $f_j$'s. We next develop the selection procedure for model \eqref{eq: additivemodel} that is capable of controlling the FDR below any given nominal level $q \in (0,1)$ under any sample size, while achieving a good power at the same time.

%%%%%%%%%%%%%%%%%%%%%%%%%%%%%%%%%%%%%%%%%%%%%%%%%%%
\section{Kernel Knockoffs Procedure}
\label{sec:KKO}

%%%%%%%%%%%%%%%%%%%%%%%%%%%%%%%%%%%%%%%%%%%%%%%%%%%
\subsection{Algorithm}
\label{sec:algorithm}

Our kernel knockoffs selection procedure consists of six main steps. Step 1 is to generate the knockoff variables. Step 2 is to subsample without replacement half of the sample observations. Step 3 is to construct the random features for both the original and knockoff variables. Step 4 is to solve the coefficient vector through a group Lasso penalized regression based on the subsamples, which in effect leads to the selection of a set of important variables. In addition, Steps 2 to 4 are carried out repeatedly over a number of subsampling replications. Step 5 is to compute the importance score for each original variable, which is defined as the empirical selection frequency based on multiple subsampling replications. Finally, Step 6 is to apply a knockoff filter to the importance scores to produce the final set of selected variables under the given FDR level, as well as the final estimate of the component functions. We summarize our procedure in Algorithm \ref{alg: kernel_KO} first, then discuss each step in detail.

%%%%%%%%%%%%%%%%%%%%%%%%%%%%%%%%%%%%%%%%%%%%%%%%%%%
\subsection{Knockoff variable construction}
\label{sec:knockoffskernel}

A random vector $\widetilde{\bX}\in\R^p$ is said to be a knockoff copy of $\bX\in\R^p$ \citep{candes2018panning} if
\begin{equation} \label{eqn:exchange}
(\bX,\widetilde{\bX}) \overset{d}{=}(\bX,\widetilde{\bX})_{\text{swap}(\mathcal A)}, \; \text{ for any } \; \mathcal A\subseteq\{1,\ldots,p\}, \quad \textrm{ and } \quad Y\independent \widetilde{\bX} \ | \ \bX,
\end{equation}
where the symbol $\overset{d}{=}$ denotes the equality in distribution, and $\text{swap}(j)$ is the operator swapping $X_j$ with $\widetilde{X}_j$. For instance, if $p=3$ and $\mathcal A=\{1,3\}$, then $(X_1,X_2,X_3,\widetilde{X}_1,\widetilde{X}_2,\widetilde{X}_3)_{\text{swap}(\mathcal A)}$ becomes $(\widetilde{X}_1,X_2,\widetilde{X}_3,X_1,\widetilde{X}_2,X_3)$. In the variable selection literature, there are alternative methods that add pseudo-variables to help control the false positives in selection, e.g., by generating independent features, or permuting entries of the existing features \citep{miller2002subset, wu2007controlling}. Different from those methods, the knockoff framework has a unique property of exchangeability as given by \eqref{eqn:exchange}.

\begin{algorithm}[t!]
\caption{Kernel knockoffs selection procedure for nonparametric additive models} 
\begin{algorithmic}[1]
\STATE \textbf{Input}: Training data $\{(x_i,y_i)\}_{i=1}^n$, the number of random features $r$, the number of subsampling replications $L$, and the nominal FDR level $q\in[0,1]$.
\STATE \textbf{Step 1}: Construct the knockoff variables $\{\widetilde{\bx}_i\}_{i=1}^n$ to augment the original variables $\{\bx_i\}_{i=1}^n$ using the second-order knockoffs or the deep knockoffs machine.
\FOR{$\ell=1$ to $L$} 
\STATE \textbf{Step 2}: Draw without replacement to obtain a subsample $I_\ell \subset\{1,\ldots,n\}$ of size $\lfloor n/2\rfloor$.
\STATE \textbf{Step 3}: Sample $2p$ of i.i.d. $r$-dimensional random features $\{w_\nu,b_\nu\}_{\nu=1}^r$ by (\ref{eqn:fouriermc}), and construct the augmented random feature vector $\bPsi_{2p}(\bX)$ by (\ref{eqn:psi2p}).  
\STATE \textbf{Step 4}: Solve the coefficient vector $\widehat{\bc}_{2p}(I_\ell)$ by \eqref{eq: kernel_reg}, and record the selected variables.
\ENDFOR
\STATE \textbf{Step 5}: Compute the importance score by \eqref{eqn:importance}, i.e., the empirical selection frequency, $\{\widehat{\Pi}_{j}\}_{j\in [2p]}$ based on the $L$ estimates of $\{\widehat{\bc}_{2p}(I_\ell)\}_{\ell \in [L]}$. 
\STATE \textbf{Step 6}: Apply the knockoff filter by \eqref{eqn:filterT} at the nominal FDR level $q$.
\STATE \textbf{Output}: the set of selected variables $\widehat{\mathcal S}$, and the function estimate $\widehat{f}^{\text{RF}}(\bX)$.
\end{algorithmic} 
\label{alg: kernel_KO}
\end{algorithm}

There have been numerous ways proposed to construct the knockoff variables. We adopt two particular constructions, depending on the data. 

The first is the second-order knockoffs construction \citep{candes2018panning}, which generates the knockoffs by matching only the first two moments of the two distributions. In this case, $\widetilde{\bX}$ is a second-order knockoff copy of $\bX$ if
\begin{equation*}
\E[\bX] = \E[\widetilde{\bX}], \quad \text{ and } \quad \text{cov}[(\bX,\widetilde{\bX})] =\begin{bmatrix}
\bSigma       & \bSigma-\text{diag}(\bs) \\
\bSigma -\text{diag}(\bs)      & \bSigma
\end{bmatrix},
\end{equation*}
where $\bSigma$ is the covariance matrix of $\bX$, and $\bs$ is a $p$-dimensional vector such that $\text{cov}[(\bX,\widetilde{\bX})]$ is positive semi-definite. To ensure a good statistical power, $\bs$ should be chosen as large as possible, so that the original and knockoff variables are differentiable \citep{candes2018panning}. This strategy is implemented in practice by approximating the distribution of $\bX$ as the multivariate normal, and is employed in numerous knockoffs-based applications \citep{BCS20}.

The second is the deep knockoffs machine \citep{romano2019deep}, which generates the knockoff variables using deep generative models. The key idea is to iteratively refine a knockoff sampling mechanism until a criterion measuring the validity of the produced knockoffs is optimized. This strategy is shown to be able to match higher-order moments, and also achieve a better approximation of exchangeability. 

In our construction of knockoff variables, we employ the second-order knockoffs when there is clear evidence that the predictor variables approximately follow a multivariate normal distribution, and employ the deep knockoffs machine otherwise. Given the training samples $\left\{ (\bx_i, y_i) \right\}_{i=1}^{n}$, we first augment with the knockoff samples $\left\{ \widetilde{\bx}_i \right\}_{i=1}^{n}$, and form the data $\left\{ (\bx_i, \widetilde{\bx}_i, y_i) \right\}_{i=1}^{n}$, where $\bx_i = (x_{i,1},\ldots,x_{i,p})\trans\in\cX^p$, and $\widetilde{\bx}_i = (\widetilde{x}_{i,1},\ldots,\widetilde{x}_{i,p})\trans\in\R^p$.

%%%%%%%%%%%%%%%%%%%%%%%%%%%%%%%%%%%%%%%%%%%%%%%%%%%
\subsection{Random feature mapping}
\label{sec:rfe}

We next construct the random features for both original and knockoff variables. The key idea is to employ the random feature mapping \citep{RR07, buazuavan2012fourier} to approximate the kernel function, which enables us to construct a projection operator between the RKHS and the original predictor space. Specifically, if the kernel functions that generate $\mathcal{H}_1$ are shift-invariant, i.e., $K(X,X') = K(X-X')$, and integrate to one, i.e., $\int_\cX K(X-X') d(X-X') = 1$, then the Bochner's theorem \citep{bochner1934theorem} states that such kernel functions satisfy the Fourier expansion: 
\begin{equation*}
\begin{aligned}
K(X-X') &=\int_\R p(w)\exp\left\{\sqrt{-1}w(X-X')\right\}dw,
\end{aligned}
\end{equation*}
where $p(w)$ is a probability density defined by
\begin{equation*} \label{eqn:probdensityf}
p(w) = \int_\cX K(X)e^{-2\pi\sqrt{-1} wX}dX.
\end{equation*}
We note that many kernel functions are shift-invariant and integrate to one. Examples include the Laplacian kernel, $K(X,X')=c_1e^{-|X-X'|/b_1}$, the Gaussian kernel, $K(X,X')=c_2e^{-b_2^2|X-X'|^2/2}$, and the Cauchy kernel, $K(X,X')=c_3(1+b_3^2|X-X'|^2)^{-1}$, where $c_1,c_2,c_3$ are the normalization constants, and $b_1, b_2, b_3$ are the scaling parameters. It is then shown that \citep{RR07,buazuavan2012fourier} the minimizer in \eqref{eqn:representer} can be approximated by,
\begin{equation} \label{eqn:RF}
\widehat{f}^{\text{RF}}(\bX) = \bPsi_p(\bX)\trans \bc_p,
\end{equation}
where $\bPsi_p(\bX) = \left[ \bPsi(X_1)\trans,\ldots,\bPsi(X_p)\trans \right]\trans \in \R^{pr}$, and $\bPsi(X_j) = \left[ \psi_1(X_j),\ldots,\psi_r(X_j) \right]\trans \in \R^r$ is a vector of $r$ Fourier bases with the frequencies drawn from the density $p(w) $, i.e., 
\begin{equation}
\label{eqn:fouriermc}
\begin{aligned}
& \omega_{j,\nu} \overset{\text{i.i.d.}}{\sim} p(\omega), \quad\quad  b_{j,\nu} \overset{\text{i.i.d.}}{\sim} \textrm{Uniform}[0,2\pi], \\
&\psi_\nu(X_j) = \sqrt{\frac{2}{r}}\cos(X_j \omega_{j,\nu} + b_{j,\nu}), \quad\quad j=1,\ldots,p, \; \nu=1,\ldots,r.
\end{aligned}
\end{equation}

The use of random feature mapping achieves potentially substantially dimension reduction. More specifically, the estimator in (\ref{eqn:RF}) only requires to learn the $pr$-dimensional coefficient $\bc_p$, compared to the estimator in (\ref{eqn:representer}) that involves an infinite-dimensional vector $\widetilde{\bc}_p$. \citet{rudi2017generalization} showed that the random feature mapping obtains an optimal bias-variance tradeoff if $r$ scales at a certain rate and $r/n\to0$ when $n$ grows. They further proved that the estimator in (\ref{eqn:RF}) can achieve the minimax optimal estimation error. Beyond the estimation optimality, we note that the random feature mapping also efficiently reduces the computational complexity. That is, the computation complexity of the estimator in (\ref{eqn:RF}) is only $O(nr^2)$, compared to the computation complexity of the kernel estimator in \eqref{eqn:representer} that is $O(n^3)$. The saving of the computation is substantial if $r/n\to 0$ as $n$ grows. 

In our setting of kernel knockoffs selection, we construct the random features for both original and knockoff variables and obtain the augmented random feature vector as, 
\begin{equation} \label{eqn:psi2p}
\bPsi_{2p}(\bX) = \left(\bPsi(X_1)\trans,\ldots, \bPsi(X_p)\trans,\bPsi(\widetilde{X}_1)\trans,\ldots,\bPsi(\widetilde{X}_p)\trans \right)\trans \in \R^{2pr}, 
\end{equation}
where $\bPsi(X_j) = \left[ \psi_1(X_j),\ldots,\psi_r(X_j) \right]\trans \in \R^r$, and $\bPsi(\widetilde{X}_j) = \left[ \psi_1(\widetilde{X}_j),\ldots,\psi_r(\widetilde{X}_j) \right]\trans \in \R^r$, $j=1,\ldots,p$, are two sets of $r$-dimensional random features that are independently sampled from \eqref{eqn:fouriermc}. Then the minimizer in \eqref{eqn:representer} can be approximated by,
\begin{equation} \label{eqn:rfrepresentation}
\widehat{f}(\bX) = \Psi_{2p}(\bX)\trans \bc_{2p}.
\end{equation}

Meanwhile, we note that the randomness of the features generated from \eqref{eqn:fouriermc} may alter the ranking of variable significances. As such, we couple the random feature mapping with knockoffs and subsampling to achieve the desired FDR control and power.

%%%%%%%%%%%%%%%%%%%%%%%%%%%%%%%%%%%%%%%%%%%%%%%%%%%
\subsection{Resampling, importance score, and knockoff filtering}
\label{sec:resample-score-kko}

We adopt the subsampling scheme similarly as that in \citet{MB10, dumbgen2013stochastic}. Specifically, we subsample a subset of the training samples without replacement with size $n_s$, and let $I$ denote the corresponding subsample indices out of $\{1,\ldots,n\}$. We set $n_s = \lfloor n/2 \rfloor$, where $\lfloor n/2 \rfloor$ is the largest integer no greater than $n/2$. We then estimate the coefficient vector $\bc_{2p} = (c_1\trans, \ldots, c_{2p}\trans)\trans\in \R^{2pr}$ in (\ref{eqn:rfrepresentation}), in which each $c_j\in\R^r$ for $j=1,\ldots,2p$, via a group Lasso penalized regression based on the subsample $I$ of the observations, 
\begin{equation} \label{eq: kernel_reg} 
\begin{aligned}
\underset{\underset{j=1,\ldots,2p}{c_{j} \in \mathbb R^{r}}}{\min}~\frac{1}{|I|}\sum_{i\in I} \left[y_i - \bar{y}(I) - \sum_{j=1}^{p}\bPsi(x_{i,j})\trans c_{j} - \sum_{j=p+1}^{2p}\bPsi(\widetilde{x}_{i,j-p})\trans c_{j} \right]^2 + \tau \sum_{j=1}^{2p} \|c_{j}\|_2,
\end{aligned}
\end{equation}
where $\bar{y}(I)=\sum_{i\in I}y_i/|I|$ is the empirical mean, and $\tau\geq 0$ is the penalty parameter. Let $\widehat{\bc}_{2p}(I) = \left( \widehat{c}_1\trans(I),\ldots,\widehat{c}_{2p}\trans(I) \right)\trans$ denote the minimizer of \eqref{eq: kernel_reg}. We remark that the group Lasso penalty in \eqref{eq: kernel_reg} encourages the entire vector $c_j\in\R^{r}$ to be shrunk to zero, for $j = 1, \ldots, 2p$. Consequently, estimating $\bc_{2p}$ via \eqref{eq: kernel_reg} in effect leads to the selection of important variables among all $2p$ candidate variables $(X_1, \ldots, X_p, \widetilde{X}_1, \ldots, \widetilde{X}_p)$. We also remark that, our use of the group Lasso penalty in (\ref{eq: kernel_reg}) is different from \citet{HHW10}. Specifically, \citet{HHW10} used the B-spline basis for nonparametric function approximation, and used the group Lasso twice, first for obtaining an initial estimator and reducing the dimension of the problem, then for selecting the nonzero components. By contrast, we use the random feature mapping for nonparametric function approximation, and apply group Lasso for selection and finite-sample FDR control. These differences have different theoretical implications; for instance, the B-spline basis is usually orthonormal, whereas the random feature mapping is generally not orthogonal.  Our group Lasso penalty is also different from the COSSO penalty used in \citet{LZ06}, which takes the form $\sum_{j=1}^{p}\|\bPsi(x_{i,j})\trans c_j\|_{K}+\sum_{j=p+1}^{2p}\|\bPsi(\widetilde{x}_{i,j-p})\trans c_j\|_{K}$. Since the random feature mapping generally cannot form an orthogonal basis, there is no closed-form representation of the RKHS norms $\|\bPsi(x_{i,j})\trans c_j\|_{K}$ and $\|\bPsi(\widetilde{x}_{i,j-p})\trans c_j\|_{K}$ in our setting. As a result, the COSSO penalty is difficult to implement, and instead we adopt the group Lasso penalty in \eqref{eq: kernel_reg} that also yields the desired theoretical properties. 

Given the penalized estimate $\widehat{\bc}_{2p}(I)$, we obtain an estimate of the selected variable indices $\widehat{\mathcal S}(I) \subseteq \{1,\ldots,2p\}$. That is, for each $j \in \{1,\ldots,p\}$, $j\in\widehat{\mathcal S}(I)$ if $\widehat{c}_j(I)\neq \mathbf{0}$ and the original variable $X_j$ is selected, and $(j+p)\in\widehat{\mathcal S}(I)$ if $\widehat{c}_{j+p}(I)\neq \mathbf{0}$ and the knockoff variable $\widetilde{X}_j$ is selected. Then the probability of being in the selected set $\widehat{\mathcal S}(I)$ is
\begin{equation} \label{eqn:calofpi}
\widehat{\Pi}_j = \P\{j\in\widehat{\mathcal S}(I)\}, \;\; \text{ for } \; j=1,\ldots,2p,
\end{equation}
where $\P$ is with respect to both subsampling $I$ and the random features. We note that $\widehat{\Pi}_j$ can be estimated accurately using the empirical selection frequencies \citep{MB10}. Specifically, we repeat the above subsampling and coefficient estimation procedure $L$ times, each time for a subsample $I_\ell, \ell = 1,\ldots,L$. We then obtain the selected variable indices $\widehat{\cS}(I_\ell)$ for $I_\ell$, and compute \eqref{eqn:calofpi} using the empirical selection frequency as the percentage of times the $j$th variable, $j=1,\ldots,2p$, is included in $\{ \widehat{\cS}(I_\ell) \}_{\ell=1}^{L}$. 

Next, we define the importance score for the original variable $X_j$, $j=1,\ldots,p$, as, 
\begin{equation} \label{eqn:importance}
\Delta_j = \widehat{\Pi}_j- \widehat{\Pi}_{j+p}. 
\end{equation}
\noindent
We comment that $\Delta_j$ in \eqref{eqn:importance} is calculated for only one run of the knockoffs procedure, i.e., we generate the knockoffs only once. This is different from the derandomized knockoffs method recently proposed by \citet{RWC20}, which aggregates the selection results across multiple runs of knockoffs to reduce the randomness of the knockoff generation. 

Finally, given the target nominal FDR level $q$, we apply a knockoff filter \citep{BC15} to the importance scores to produce the final set of selected variables, 
\begin{equation} \label{eqn:filterT}
T = \min \left\{t\in\{|\Delta_j|:|\Delta_j|>0\}: \frac{\# \{j: \Delta_j \le -t \}}{\# \{ j : \Delta_j \ge t\}} \le q\right\} \quad (\text{knockoffs}).
\end{equation}
Set $T=\infty$ if the above set is empty. Another commonly used but slightly more conservative knockoff filter \citep{BC15,candes2018panning} is, 
\begin{equation} \label{eqn:knockoff+thres}
T_+ = \min \left\{t\in\{|\Delta_j|:|\Delta_j|>0\}: \frac{\# \{j: \Delta_j \le -t \}+1}{\# \{ j : \Delta_j \ge t\}} \le q\right\}\quad (\text{knockoffs+}).
\end{equation}
In our simulations, we have experimented with both filers, which produce very similar results, so we only present the results based on $T$.

Given the threshold value $T$, the final set of selected variables is, 
\begin{equation} \label{eqn:selectedRF}
\widehat{\mathcal S} = \Big\{ j\in\{1,\ldots,p\}:\Delta_j \geq T \Big\}.
\end{equation}
We then reestimate $\bc_p$ in (\ref{eqn:RF}) using all the sample observations as, 
\begin{equation*}
\widehat{\bc}_p^{\text{RF}} = \left( (\widehat{c}_1^{\text{RF}})\trans, \ldots, (\widehat{c}_{p}^{\text{RF}})\trans \right)\trans= \underset{{c_{j} \in \mathbb R^{r},j\in\widehat{\cS}}}{\arg\min}~\frac{1}{n}\sum_{i=1}^n \left[ y_i- \frac{1}{n}\sum_{i=1}^ny_i-\sum_{j\in  \widehat{\cS}}\bPsi(x_{i,j})\trans c_j \right]^2.
\end{equation*}
We obtain the final knockoffs-based kernel regression estimator as, 
\begin{equation} \label{eqn:newRFestimator}
\widehat{f}^{\text{RF}}(\bX) =\bPsi_p(\bX)\trans \widehat{\bc}_p^{\text{RF}}.
\end{equation}

%%%%%%%%%%%%%%%%%%%%%%%%%%%%%%%%%%%%%%%%%%%%%%%%%%%
\subsection{Parameter tuning}
\label{sec:tuning}

We next discuss the parameter tuning. We further carry out a sensitivity analysis in Section B.2, and a parallelization experiment in Section B.5 of the Appendix.
  
For the number of random features $r$, we start with an initial set $\Xi$ of candidate values for $r$. For each working rank $r \in \Xi$, we calculate the selection frequencies $\big\{ \widehat{\Pi}_{j,r} \big\}_{j\in[2p],r\in\Xi}$, and the standard deviation $ \widehat{\sigma}_r={\rm sd} \left( \big\{ \widehat{\Pi}_{j,r} \big\}_{j\in[2p]} \right)$, with a relatively small number for the subsampling replications. We then choose the value of $r \in \Xi$ that maximizes the following criterion that balances the selection standard deviation and model complexity,  
\begin{equation*} \label{eq: tune_object}
\widehat{r} = \argmax_{r \in \Xi}  2 p \widehat{\sigma}_r - \ln (r).
\end{equation*}
We have observed through our numerical simulations that, when we start from a small value of $r$, the selection frequencies of both original variables and their knockoffs counterparts are close to  zero. As $r$ increases, it starts to separate the truly important variables from the null variables and knockoffs, where the selection frequencies of those truly  important variables grow positively, and correspondingly, the standard deviation $\widehat{\sigma}_r$ increases. Meanwhile, the log penalty term helps balance the model complexity. 
  
For the regularization parameter $\tau$ in (\ref{eq: kernel_reg}), we choose it by minimizing the BIC criterion, 
\begin{equation*} \label{eqn:rssbic}
\widehat{\tau} = \argmin_{\tau\geq 0} \log[{\rm RSS}(\tau)] + r \frac{\log n }{n } | \widehat{\cS}(\tau) |, 
\end{equation*}
where $\rm{RSS}(\tau)$ is the cross-validation residual sum of squares, and $| \widehat{\cS}(\tau) |$ is the cardinality.  

For the number of subsampling replications $L$, our numerical experiments have found that $L=100$ results in a competitive performance in FDR control and power. For the subsampling sample size $n_s$, we have found that, when $n_s$ is no smaller than $\lfloor n/2 \rfloor$, the method performs well. We also comment that, the computation of our method can be easily parallelized, since it requires no information sharing across different subsamples.

%%%%%%%%%%%%%%%%%%%%%%%%%%%%%%%%%%%%%%%%%%%%%%%%%%%
\section{Theoretical Guarantees}
\label{sec:theory}

%%%%%%%%%%%%%%%%%%%%%%%%%%%%%%%%%%%%%%%%%%%%%%%%%%%
\subsection{FDR Control}
\label{sec:fdr-control}

We show that our proposed procedure controls the FDR under any given nominal level and any given sample size. Due to the intrinsic difficulty of the nonparametric additive model, we employ the random feature mapping to approximate the component function $f_j$, and construct a projection operator between the RKHS and the original predictor space. By construction, the random features of the knockoff variables have a similar structure mimicking the random features of the original variables, even though the knockoffs are not associated with the response conditioning on the original ones. Since the random features may introduce additional stochastic errors, which can disturb the order of the variables entering the model, we further employ the subsampling strategy to improve the selection stability. Finally, we compute the importance score for each variable based on the projected component function $f_j$, and select the variables that have considerably higher scores than their knockoff counterparts. Intuitively, such a procedure enjoys the finite-sample FDR control, similarly as the existing knockoff solutions \citep{BC15, candes2018panning}.

We first show that the importance score $\Delta_j$ in (\ref{eqn:importance}) has a symmetric distribution for a null variable $X_j \in \cS^\perp$, and is equally likely to be positive or negative. The symmetric property of the null variables is crucial for the knockoffs procedure, which then chooses a data-dependent threshold while having the FDR under control \citep{BC15}.
\begin{theorem} 
\label{prop:coinflip}
Suppose Assumption \ref{eqn:perfect-relationship} holds. Let $(s_1, \ldots, s_p)$ be a set of independent random variables, such that $s_j = \pm1$ with probability $1/2$ if $j \in \mathcal S^\perp$, and $s_j=1$ if $j \in \mathcal S$. Then, 
\begin{equation*}
(\Delta_1,\ldots, \Delta_p) \overset{d}{=} (\Delta_1 \cdot s_1, \ldots, \Delta_p \cdot s_p).
\end{equation*}
\end{theorem}

Next, we show that our selection procedure successfully controls the false discovery under any sample size. The result holds regardless of the distribution or the number of predictors, and does not require any knowledge of the noise level. The false discovery here is measured by both the FDR, and the modified FDR, which is defined as, 
\begin{equation*}
\text{mFDR} = \E\left[\frac{|\widehat{\mathcal S}\cap\mathcal S^\perp|}{|\widehat{\mathcal S}|+1/q}\right].
\end{equation*}
The definition of mFDR follows the knockoffs literature \citep{BC15}, with 1 replaced by $1/q$ in the denominator compared to FDR. Meanwhile, it is close to FDR in the setting when there are a large number of variables selected, i.e., when $|\widehat{\mathcal S}|$ is large, and it is less conservative than FDR, in that mFDR is always under control if FDR is.

\begin{theorem}
\label{prop:mFDRcontrol}
For any $q\in[0,1]$ and any sample size $n$, the selected set $\widehat{\mathcal S}$ in (\ref{eqn:selectedRF}) based on the knockoff filter $T$ in \eqref{eqn:filterT} satisfies that $\text{mFDR} \leq q$. Meanwhile, the selected set $\widehat{\mathcal S}$ based on the knockoff filter $T_+$ in \eqref{eqn:knockoff+thres} satisfies that $\text{FDR} \leq q$. 
\end{theorem}
\noindent
We remark that, Theorem \ref{prop:mFDRcontrol} achieves the valid FDR control with no restriction on the dimension $p$ relative to the sample size $n$. As such, the proposed method works for both settings of $p < n$ and $p>n$. In particular, the FDR control under $p > n$ is achieved by building upon the model-X knockoffs \citep{candes2018panning}, which treats the variables as random and utilizes the stochasticity of the random variables. This is different from the original knockoff solution \citep{BC15}, which treats the variables as fixed and relies on specific stochastic properties of the linear model, and thus excludes the setting of $p>n$ or nonlinear models.

%%%%%%%%%%%%%%%%%%%%%%%%%%%%%%%%%%%%%%%%%%%%%%%%%%%
\subsection{Power Analysis}
\label{sec:poweranalysis}

\noindent
Next, we show that our proposed kernel knockoffs selection procedure achieves a power that approaches one as the sample size tends to infinity. We first note that, the theoretical power analysis for the knockoff methods is largely missing in the current literature, with a few exceptions such as \citet{FDLL20} and \citet{weinstein2020power}. \citet{FDLL20} studied the power for linear regressions under the model-X knockoff framework. \citet{weinstein2020power} studied the power of knockoffs with thresholded Lasso for linear models. By contrast, we study the power for nonparametric models. We also remark that, as is common for all knockoffs selection methods, the power of our knockoffs-based method is usually no greater than that of the group Lasso-based selection. This is because the proposed knockoffs procedure is built on top of the group Lasso selection in (\ref{eq: kernel_reg}). In a sense, the knockoffs procedure further selects variables from the set of variables that are identified by group Lasso for the augmented predictors. Therefore, the key of our power analysis is to investigate how much power loss that the knockoffs procedure would induce. We introduce some regularity conditions. 

\begin{assumption}
\label{assump:complexity}
The number of nonzero component functions,  i.e., $|\cS|$, is bounded. 
\end{assumption}

\begin{assumption}
\label{network:dependency}
Suppose there exists a constant $C_{\min}>0$, such that the minimal eigenvalue of matrix $\E[n^{-1} \bSigma_{\cS}\trans \bSigma_{\cS}]$ satisfies that,  
\begin{equation*}
\Lambda_{\min}\left(\E\left[\frac{1}{n}\bSigma_{\cS}\trans \bSigma_{\cS}\right]\right)\ge \frac{1}{2}C_{\min},
\end{equation*}
where  the expectation is taken over the random features and $\bSigma_\cS \in\R^{n\times 2r|\cS|}$ is the design matrix with the $i$th row equal to $\left[ \bPsi(x_{i,j_1})\trans, \ldots,   \bPsi(x_{i,j_{|\cS|}})\trans, \bPsi(\widetilde{x}_{i,j_1})\trans, \ldots, \bPsi(\widetilde{x}_{i,j_{|\cS|}})\trans \right]$, $i=1,\ldots,n$, $\cS=\{j_1,\ldots,j_{|\cS|}\}$.
\end{assumption}

\begin{assumption}
\label{network:minregeffect}
Suppose $p<e^n$. Let $\eta_R\equiv c_\eta\left\{ n^{-\beta/(2\beta+1)}+[(\log p)/n]^{1/2} \right\}$ for some constant $c_\eta>0$. Suppose $\min_{j\in \cS}\|f_{j}(X_j)\|_{L_2(X_j)}\geq \kappa_n\eta_R$, for some slowly diverging sequence $\kappa_n\to\infty$, as $n\to\infty$, where the RKHS $\cH_1$ is embedded to a $\beta$th order Sobolev space with $\beta>1$.
\end{assumption}

\begin{assumption}
\label{network:incoherence}
Suppose there exists a constant $0 \leq \xi_{\bSigma} < 1$ such that,
\begin{equation*} 
\max_{j\not\in \cS}\left\|\{\bSigma_{j}(I)\}\trans \bSigma_{\cS}(I)[\bSigma_{\cS}(I)\trans \bSigma_{\cS}(I)]^{-1}\right\|_{2}\leq \xi_{\bSigma}, \quad \textrm{ and } \quad 
 \frac{\xi_{\bSigma}\sqrt{|\cS|}+1}{\tau}\eta_{R}+\xi_{\bSigma}\sqrt{|\cS|}<1. 
\end{equation*}
\end{assumption}

\noindent 
All these conditions are reasonable and are commonly imposed in the literature. Specifically, Assumption \ref{assump:complexity} concerns the overall complexity in that it upper bounds the total number of nonzero component functions. Similar conditions have been commonly adopted in sparse additive models over RKHS \citep[e.g.,][]{koltchinskii2010sparsity, raskutti2012minimax, Yuan2016, dai2021kernel}. Moreover, we carry out a numerical experiment in Section B.4 of the Appendix, and show empirically that our method still works reasonably well when the number of nonzero components $|\cS|$ increases along with the sample size. We speculate that it is possible to allow $|\cS|$ to diverge, but leave the full theoretical investigation as future research.  Assumption \ref{network:dependency} ensures the identifiability among the $|\cS|$ submatrices of $\bSigma_{\cS}$. The same condition has been used in \citet{Zhao2006, ravikumar2010high}.  Assumption \ref{network:minregeffect} imposes some regularity on the minimum regulatory effect. Similar conditions have been used in Lasso regressions \citep{ravikumar2010high,  raskutti2012minimax, FDLL20}. In addition, similar to the treatment of high-dimensional linear and additive models \citep{raskutti2011minimax, Yuan2016}, we assume that $p<e^n$ to ensure nontrivial probabilistic bounds. In other words, we allow the dimension $p$ to diverge at the exponential order of the sample size $n$. Assumption \ref{network:incoherence} reflects the intuition that the large number of irrelevant variables cannot exert an overly strong effect on the relevant variables. Besides, the second inequality characterizes the relationship between $\xi_{\bSigma}$, the sparse tuning parameter $\tau$, and the sparsity level $|\cS|$. This condition is again standard for Lasso regressions \citep{Zhao2006, ravikumar2010high}.
 
Next, we characterize the statistical power of the proposed kernel knockoffs procedure. For the true set $\cS$ and the selected set $\widehat{\cS}$, the power is defined as
\vspace{-0.01in}
\begin{equation*}
\text{Power}(\widehat{\mathcal S}) = \E\left[\frac{|\widehat{\mathcal S}\cap \mathcal S|}{|\mathcal S|\vee 1}\right].
\end{equation*}

\begin{theorem}
\label{thm:optimalrecovery}
Suppose Assumptions  \ref{eqn:perfect-relationship}--\ref{network:incoherence} hold, and the number of random features $r \geq c_rn^{2\beta/(2\beta+1)}$ for some $c_r>0$. Then, the selected set $\widehat{\cS}$ in (\ref{eqn:selectedRF}) satisfies that, 
$\text{Power}(\widehat{\cS}) \to 1$, as $n \to \infty$.
\end{theorem}
\noindent
We again remark that, Theorem \ref{thm:optimalrecovery} holds for both settings of $p < n$ and $p > n$, or more specifically, $n < p < e^n$, which is implied by Assumption \ref{network:minregeffect}. The power property under $p > n$ is achieved by integrating Rademacher processes and the concentration inequalities for empirical processes \citep{G02, Yuan2016} with the deviation conditions for nonparametric regressions \citep{loh2012, dai2021kernel}.

Together, Theorems \ref{prop:mFDRcontrol} and \ref{thm:optimalrecovery} show that our proposed selection method is able to achieve both the finite-sample FDR control and the asymptotic power that approaches one.

%%%%%%%%%%%%%%%%%%%%%%%%%%%%%%%%%%%%%%%%%%%%%%%%%%%
\section{Numerical Studies}
\label{sec:numerical}
 
We carry out intensive simulations to examine the empirical performance of our proposed method under the varying signal strength, the predictor distribution, the nonparametric component function, the sample size and the number of predictors. We compare with several alternative solutions. We also illustrate our method with an analysis of brain imaging data for Alzheimer’s disease. We report additional simulation results in Section B of the Appendix.

%%%%%%%%%%%%%%%%%%%%%%%%%%%%%%%%%%%%%%%%%%%%%%%%%%%
\subsection{Alternative methods for comparison}

We abbreviate our proposed kernel knockoffs selection method as KKO. We solve the group Lasso penalized problem in \eqref{eq: kernel_reg} using the R package \texttt{grpreg}. We employ the Laplacian kernel with $r\in\Xi=\{2,3,4\}$, and tune the hyperparameters following Section \ref{sec:tuning}. We set the target FDR level at $q=0.2$ following \citet{FDLL20}. We also briefly comment that, in addition to the reproducing kernel approach, the spline basis expansion is another commonly used approach in the nonparametric additive modeling. But it involves a totally different set of methodological tools and theoretical analysis, and we leave it as future research.

We compare our method with three main competitors. The first competitor is the nonparametric selection method for sparse additive models (SPAM) of \citet{RLLW09}, which combines B-spline basis expansion with grouped Lasso. We set the number of B-spline expansions at $\lceil n^{1/5}\rceil$, i.e., the largest integer no greater than $n^{1/5}$, and tune the  sparsity penalty by generalized cross-validation. We implement the method using the R package \texttt{SAM}. We did not compare with the COSSO method of \citet{LZ06}, due to that the code is not available, and SPAM usually achieves a similar and sometimes more competitive performance than COSSO. The second competitor is the linear knockoffs (LKO) selection method, and we implement it using the R package \texttt{knockoffs}. The third competitor is the graphical nonlinear knockoffs (RANK) selection method of \citet{FDLL20}. We follow the same parameter setup as in \citet{FDLL20}, and implement their method based on the R package \texttt{gamsel}.

We also compare our method with that of \cite{DB16}. More specifically, \cite{DB16} adopted an ``expansion first" strategy, which first performs feature expansion $\widetilde{\bPsi}(X_j)$, $j=1,\ldots,p$, then constructs the knockoffs based on the expanded features. By contrast, we adopt the ``knockoffs first" approach, which constructs the knockoffs directly for the variables $\{X_j\}_{j=1}^p$, then performs the random feature expansion on both original variables and their knockoffs. There are two advantages of doing the knockoffs first. First, it ensures a better knockoffs construction and eventually a better statistical power. That is, to construct a good knockoff variable using either the second-order knockoffs or the deep knockoffs, it requires a reasonably slow eigenvalue decay of the covariance $\bSigma$, so that the original variables and their knockoffs are differentiable \citep{candes2018panning}. We consider a simulation example replicated 200 times, where the predictors follow a multivariate normal distribution with the zero mean and the identity covariance matrix, $n=500, p=5$, and we employ the Laplacian kernel with $r=3$. Figure \ref{fig: eigendecay}(a) shows the eigenvalue decay of the sample covariance matrix (blue line), and the random kernel expansion (red line). It is seen that the former decays more slowly than the latter, and therefore it is better to construct the knockoffs based on the original variables. Second, the ``expansion first" leads to larger correlations between the original variables and their knockoffs, compared to the ``knockoffs first", as shown in Figure \ref{fig: eigendecay}(b), which would in turn make the subsequent group Lasso selection harder. Finally, the ``expansion first'' is computationally more expensive. This is because the ``expansion first'' approach requires generating the knockoffs for $pr$-dimensional variables, whereas our ``knockoffs first" approach only requires generating the knockoffs for $p$-dimensional variables. 

\def\myfigheight{1.65in}

\begin{figure}[t!]
\centering
\includegraphics[width=3.0in,height=\myfigheight]{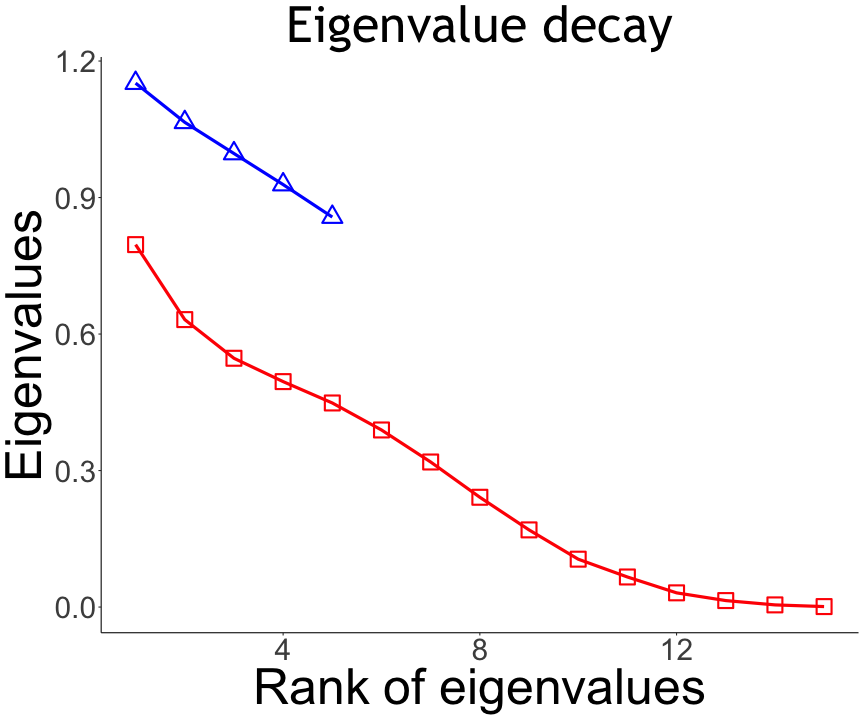}
\includegraphics[width=3.0in,height=\myfigheight]{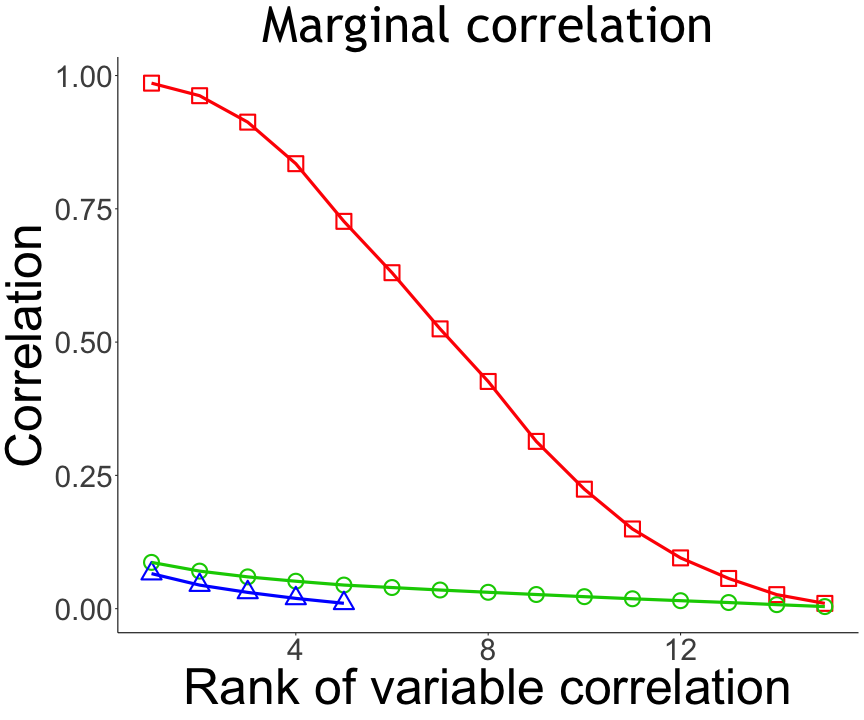}
\caption{Comparison between the ``expansion first" strategy and the ``knockoffs first" strategy. Left panel: the eigenvalue decay of the sample covariance matrix (blue line) and the random kernel expansion (red line). Right panel: the marginal sample correlations between the original variables and the knockoff variables (blue line), between the random kernel expansion and the ``expansion first'' knockoff variables (red line), and between the random kernel expansion and the ``knockoffs first'' knockoff variables (green line).}
\label{fig: eigendecay}
\vskip 0em
\end{figure}

%%%%%%%%%%%%%%%%%%%%%%%%%%%%%%%%%%%%%%%%%%%%%%%%%%%
\subsection{Varying signal strength, predictor design, and component functions}
\label{sec:varysamplesim}

We first study the performance with the varying signal strength and the predictor distribution. 
 
\begin{figure}[t!]
\centering
\begin{tabular}{cc}
\multicolumn{2}{c}{{\small Multivariate normal predictor distribution}} \\
\includegraphics[width=3.0in,height=\myfigheight]{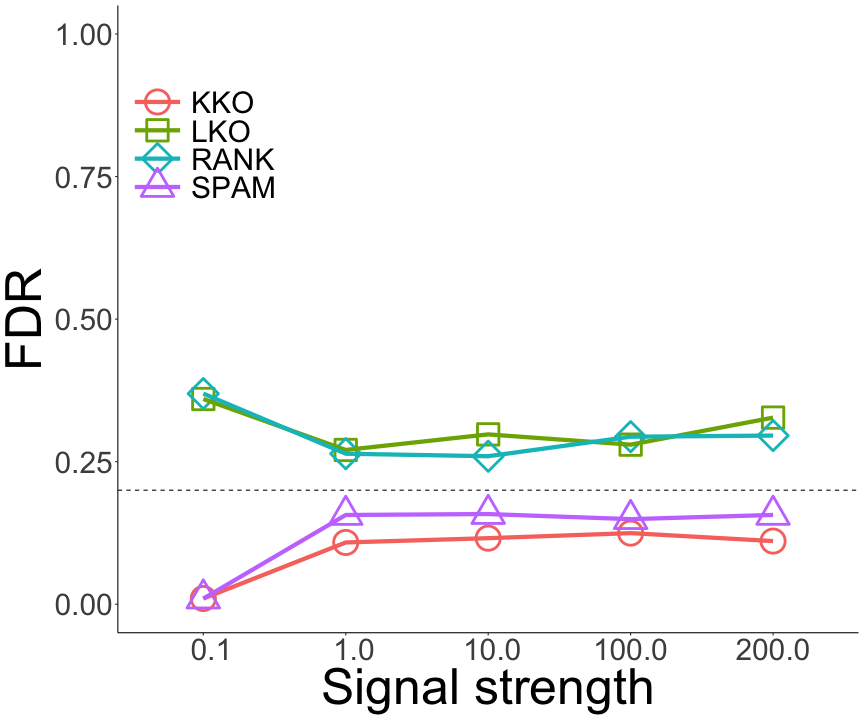} &
\includegraphics[width=3.0in,height=\myfigheight]{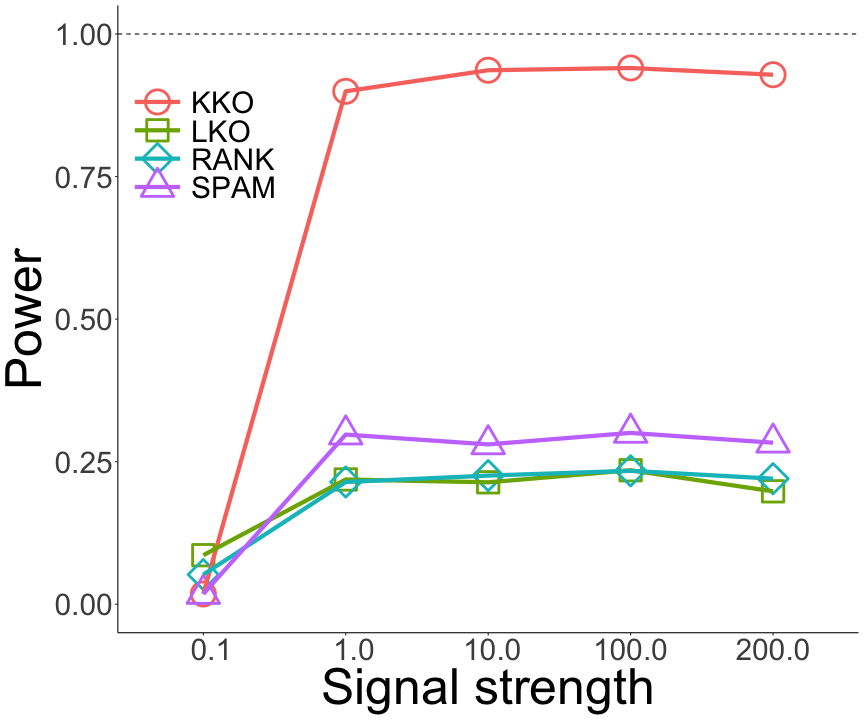} \vspace*{8pt} \\
\multicolumn{2}{c}{{\small Mixutre of multivariate normal predictor distribution}} \\
\includegraphics[width=3.0in,height=\myfigheight]{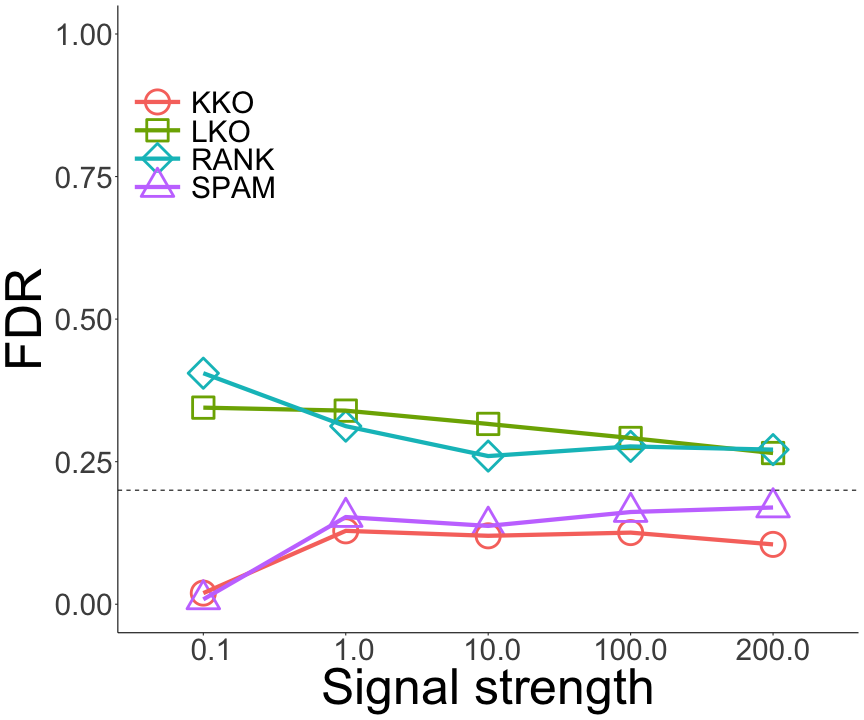} & 
\includegraphics[width=3.0in,height=\myfigheight]{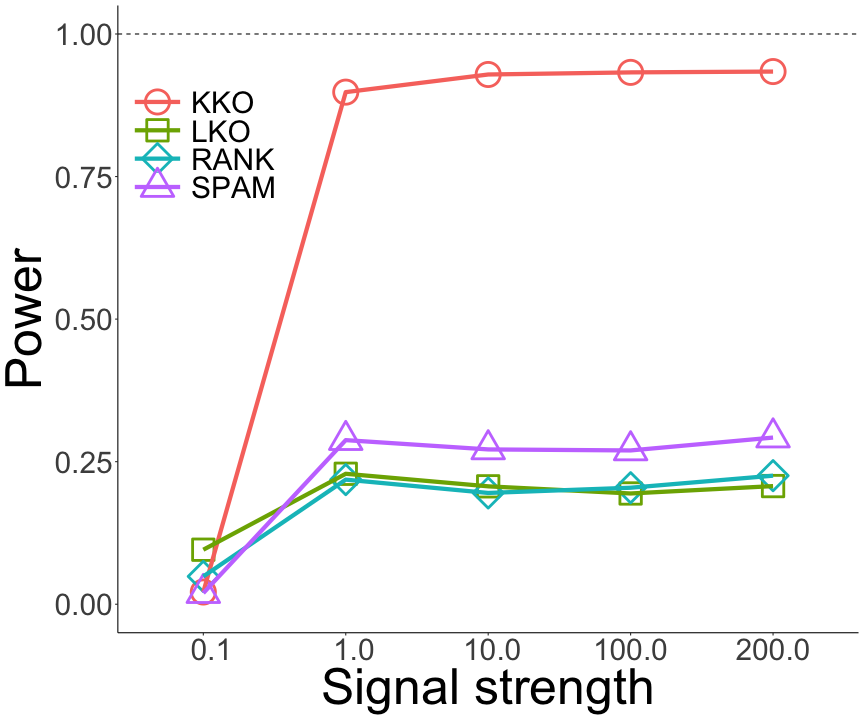} \vspace*{8pt} \\
\multicolumn{2}{c}{{\small Uniform predictor distribution}} \\
\includegraphics[width=3.0in,height=\myfigheight]{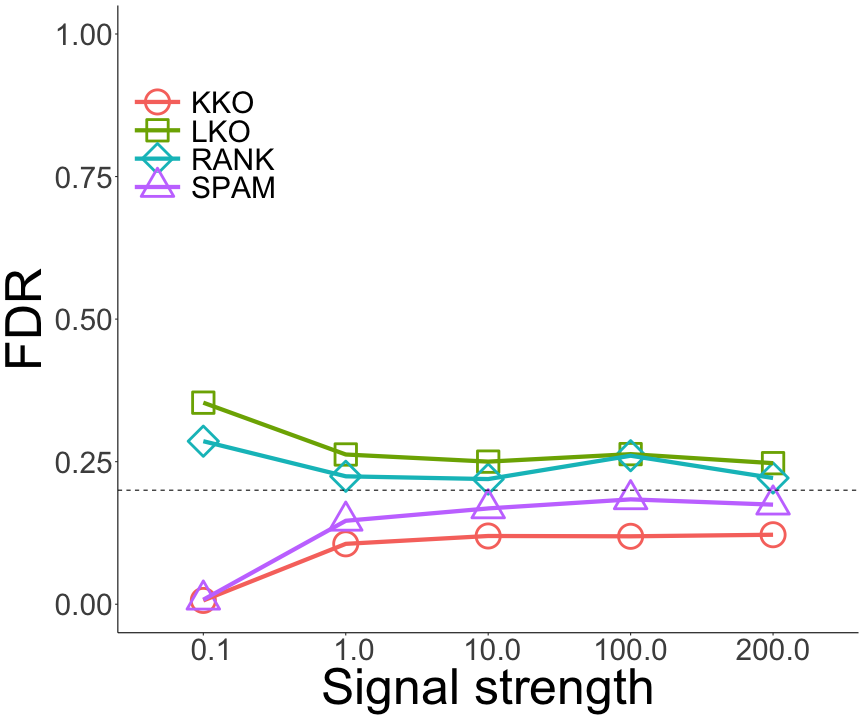} & 
\includegraphics[width=3.0in,height=\myfigheight]{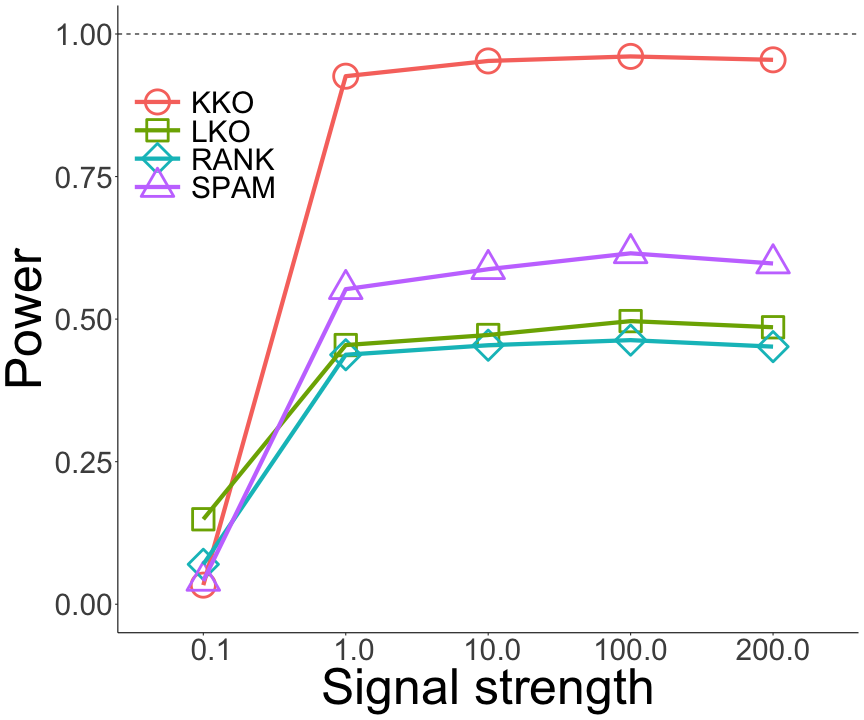} 
\end{tabular}
\caption{Empirical performance and comparison in terms of FDR and power with the varying signal strength and predictor distribution. Four methods are compared: the nonparametric selection method for sparse additive models (SPAM) of \citet{RLLW09}, the linear knockoffs (LKO) of \citet{BC15}, the graphical nonlinear knockoffs (RANK) of \citet{FDLL20}, and our proposed kernel knockoffs (KKO).}
\label{fig: robust_distribution_SNR}
\end{figure} 

We simulate the response, $Y = \sum_{j \in \cS}  \theta_j f_j (X_j ) + \epsilon$, where $\cS$ is the set of relevant predictors with $|\cS| = 10$, and $\epsilon$ is a standard normal error. We sample $\theta_j$ independently from a uniform distribution $(-\theta, \theta)$ for some positive constant $\theta$. The magnitude of $\theta$ reflects the strength of the signal, and we vary $\theta = \{0.1,1,10,100,200\}$. We simulate the predictors independently from three different distributions, a multivariate normal distribution with mean zero and covariance $\Sigma_{ij}=0.3^{|i-j|}$, a mixture normal distribution, with an equal probability from three multivariate normal distributions, all with mean zero, and different covariances where $\Sigma_{1,ij}=0.1^{|i-j|}$, $\Sigma_{2,ij}=0.3^{|i-j|}$, and $\Sigma_{3,ij}=0.5^{|i-j|}$, and a uniform $[-2, 2]$ distribution. We employ the second-order knockoffs when the predictor distribution is normal, and the deep knockoffs otherwise, to generate the knockoff variables. We first consider a trigonometric polynomial component function, and fix the number of predictors at $p=50$, and the sample size at $n=900$. We later consider other forms of component functions, and different $(p, n)$. 
\begin{equation} \label{eq: sinpoly}
f_j (x)= u_{j,1} \sin (c_{j,1} x) + u_{j,2}\cos (c_{j,2} x) + u_{j,3} \sin^2 (c_{j,3} x) +u_{j,4} \cos^2 (c_{j,4} x),
\end{equation}
where $u_{j,k}$ follows a uniform $(1,2)$ distribution, and $c_{j,k}$ follows a uniform $(1,10)$, for $k=1,2,3,4$. Figure \ref{fig: robust_distribution_SNR} reports the FDR and power over 200 data replications for the four methods with the varying signal strength $\theta$ and three different predictor distributions. It is seen that our method successfully controls the FDR blow the expected level, and at the same time achieves the best power. Besides, the performance is robust with respect to different predictor distributions. By comparison, the alternative methods are much more sensitive in terms of the FDR control, and the powers are consistently lower. Moreover, the linear knockoffs method often fails to control the FDR. 

\begin{figure}[t!]
\centering
\begin{tabular}{cc}
\multicolumn{2}{c}{{\small Sin-ratio component function}} \\
\includegraphics[width=3.0in,height=\myfigheight]{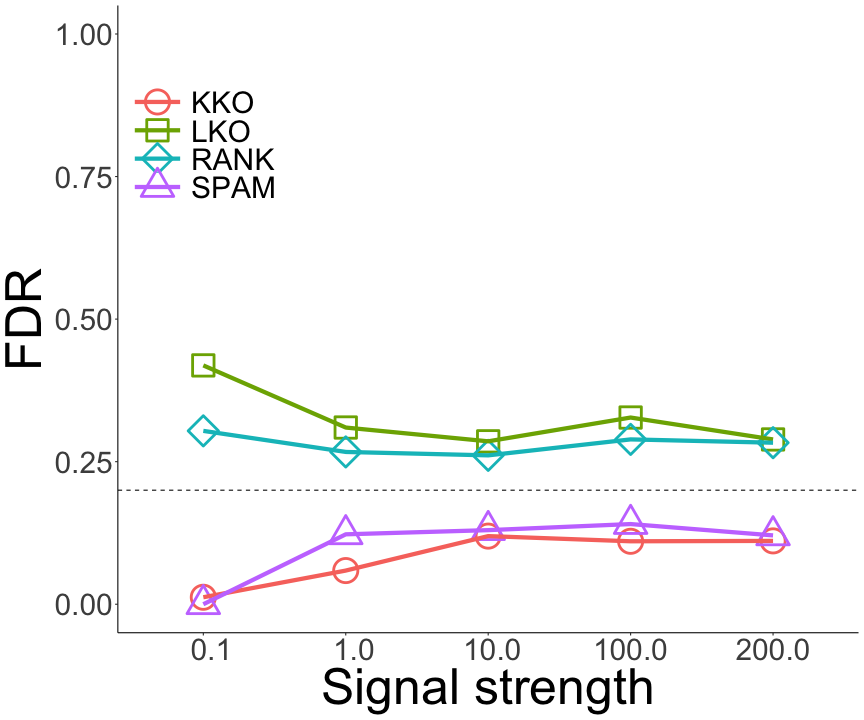} &
\includegraphics[width=3.0in,height=\myfigheight]{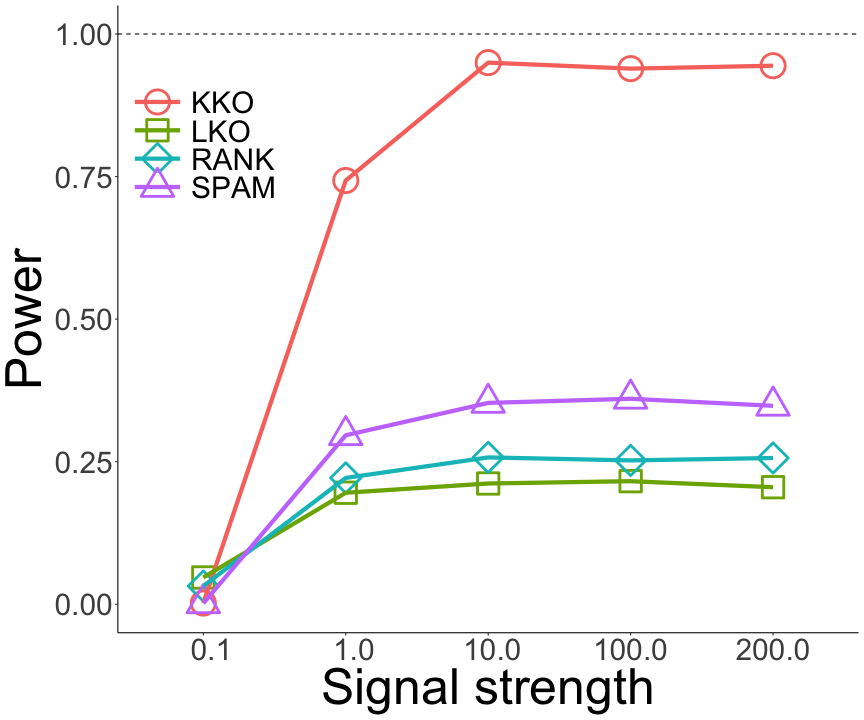} \vspace*{8pt} \\
\multicolumn{2}{c}{{\small Mixture of trigonometric polynomial and sin-raito component function}} \\
\includegraphics[width=3.0in,height=\myfigheight]{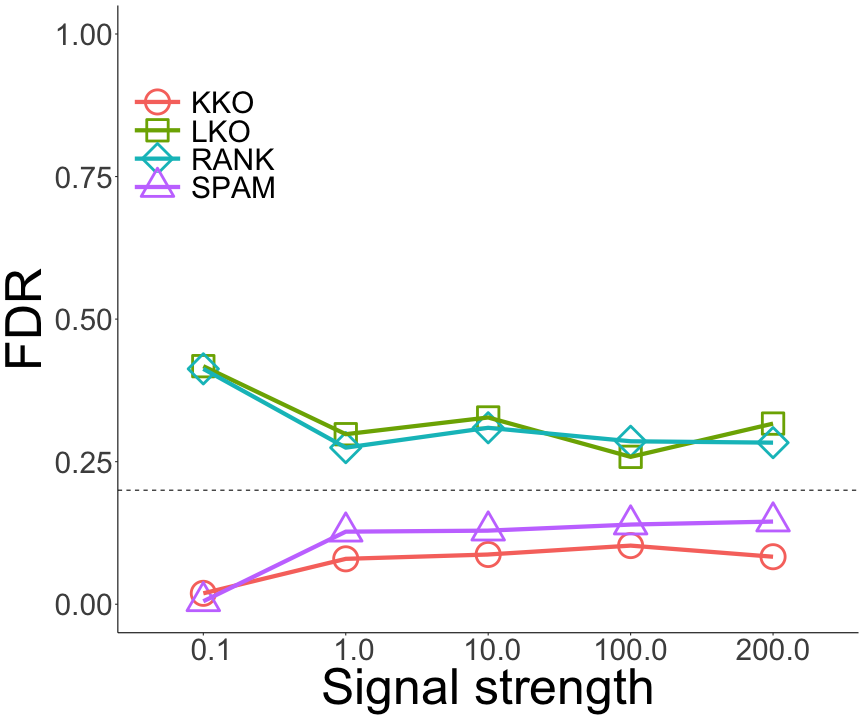} &
\includegraphics[width=3.0in,height=\myfigheight]{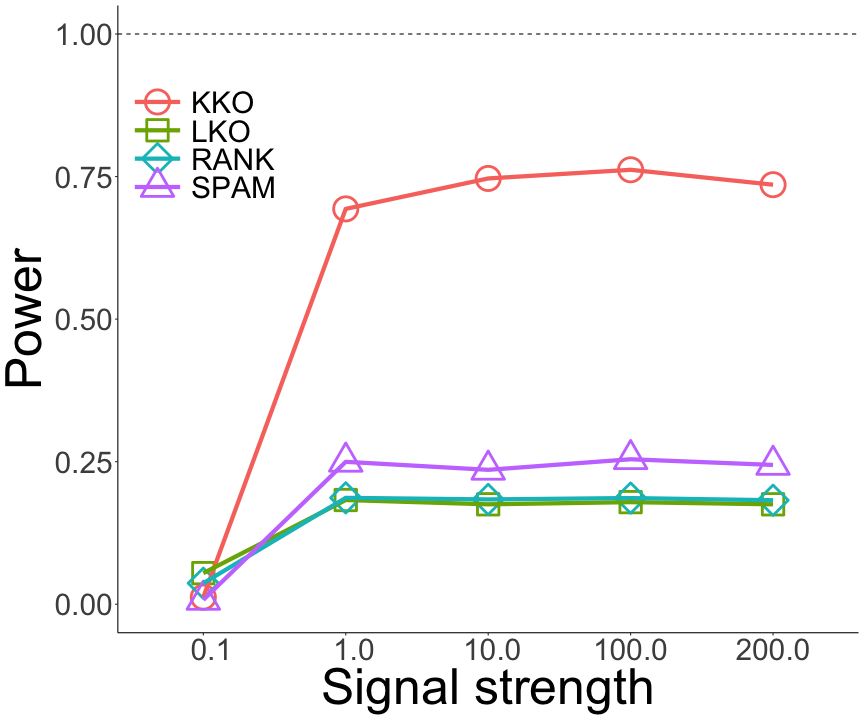} 
\end{tabular}
\caption{Empirical performance and comparison in terms of FDR and power with the varying signal strength and component function. The same four methods as in Figure \ref{fig: robust_distribution_SNR} are compared.}
\label{fig: robust_model}
\end{figure} 

Next, we consider more forms of component functions. The first is a sin-ratio function,   
\begin{equation} \label{eq: sinratio}
f_j (x) = \frac{\sin (c_{j,1} x)}{ 2- \sin (c_{j,2} x) },
\end{equation}
where $c_{j,k}$ follows a uniform $(1,10)$ distribution for $k=1,2$, and $|\cS| = 10$. We note that it is generally more difficult to estimate the sin-ratio function \eqref{eq: sinratio} compared to the trigonometric polynomial function \eqref{eq: sinpoly}. The second is a mixed additive model, where we sample the component function with an equal probability from \eqref{eq: sinpoly} or \eqref{eq: sinratio}. We fix the predictor distribution as the multivariate normal, $p=50$ and $n=900$. We continue to vary the signal strength $\theta = \{0.1,1,10,100,200\}$. Figure \ref{fig: robust_model} reports the FDR and power based on 200 data replications. It is seen again that our method achieves the best power while controlling the FDR under the nominal level for the new component functions.

%%%%%%%%%%%%%%%%%%%%%%%%%%%%%%%%%%%%%%%%%%%%%%%%%%%
\subsection{Varying sample size and dimension}
\label{sec:sim-np}

Next, we investigate the empirical performance with the varying $n$ and $p$. 

For the varying $n$, we consider the trigonometric polynomial function \eqref{eq: sinpoly} with $p = 50, |\cS| = 10, \theta = 100$, and the multivariate normal predictor distribution with mean zero and covariance $\Sigma_{ij}=0.3^{|i-j|}$. We vary the sample size $n = \{400, 900, 1500, 2500\}$. Figure \ref{fig: sam_effect} reports the FDR and power based on 200 data replications. It is seen that our method successfully control the FDR at all sample sizes, while its power quickly increases as $n$ increases, and dominates the powers of all the competitive methods considerably. Besides, both LKO and RANK have an inflated FDR especially when the sample size is small.

\begin{figure}[b!]
\centering
\includegraphics[width=3.0in,height=\myfigheight]{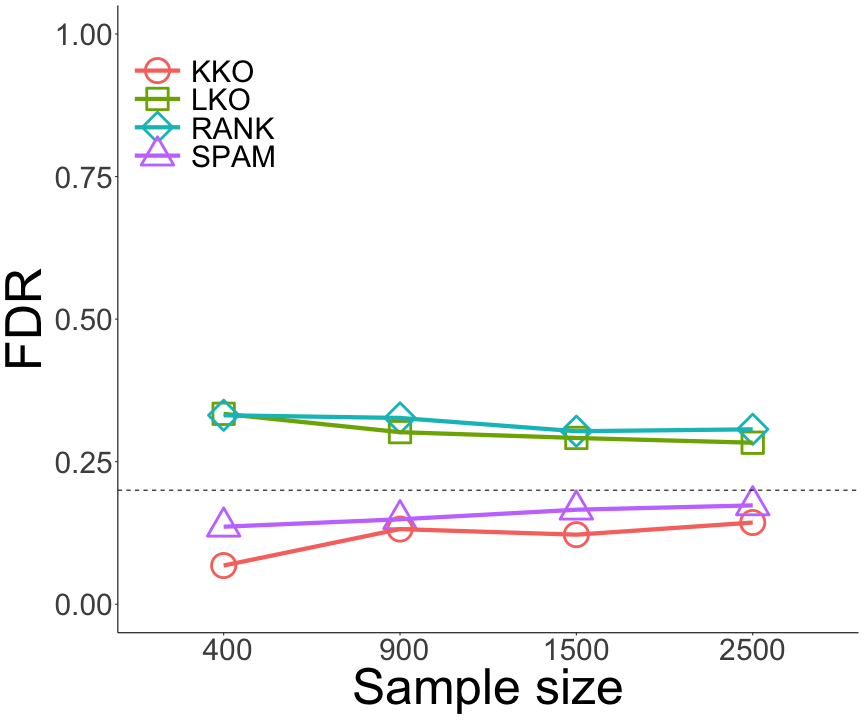}
\includegraphics[width=3.0in,height=\myfigheight]{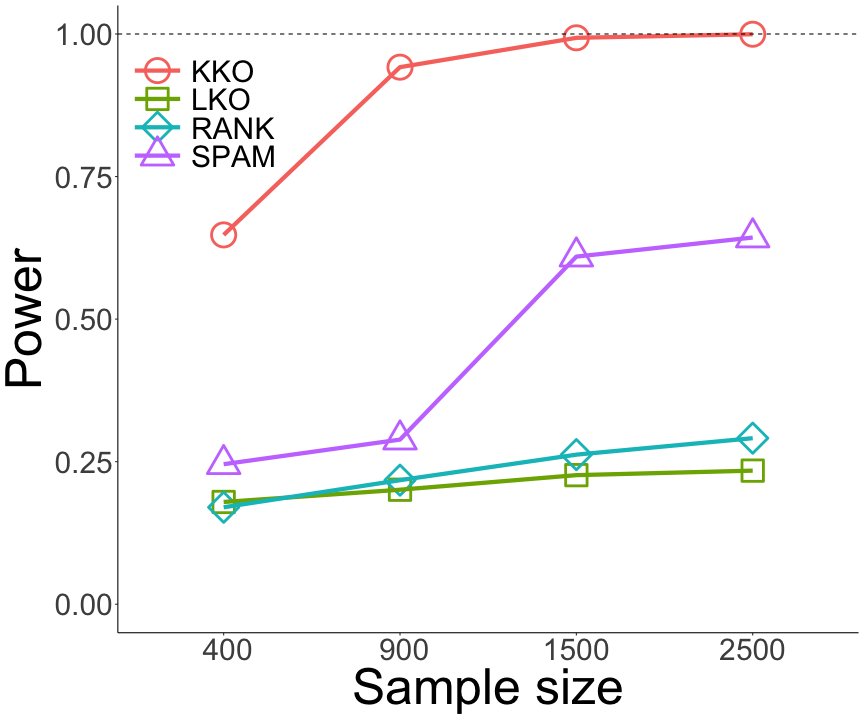}
\caption{Empirical performance and comparison in terms of FDR and power with the varying sample size $n$. The same four methods as in Figure \ref{fig: robust_distribution_SNR} are compared.}
\label{fig: sam_effect}
\end{figure} 

For the varying $p$, we consider the trigonometric polynomial function \eqref{eq: sinpoly} with $n = 900, |\cS| = 10, \theta = 100$, and the multivariate normal predictor distribution  with mean zero and covariance $\Sigma_{ij}=0.3^{|i-j|}$. We vary the number of predictors $p = \{30, 50, 150, 500, 1500, 2000\}$. As such, we have considered both cases that $p < n$ and $p > n$. Recall that the proposed method can handle both the low-dimensional and high-dimensional regimes, and the theoretical guarantees in Section \ref{sec:theory} are established for both $p < n$ and $p > n$. Moreover, for the $p > n$ case, we construct the knockoff variables for all the original variables; i.e., we construct the $p$-dimensional knockoffs $\widetilde{\bX}$ for the $p$-dimensional $\bX$. This follows a similar strategy as in \citet{FDLL20}, but is different from \citet{BC19}, where a precedent step of feature screening is carried out first, and the knockoff variables are constructed only for those selected variables. Our approach avoids to require the sure screening property; see also the discussion in \citet{FDLL20}.  Figure \ref{fig: robust_sparsity_ratio} reports the FDR and power based on 200 data replications. It is seen that our method again achieves the best performance in terms of FDR and power in both regimes. By contrast, LKO and RANK have a low power and inflated FDR. Although SPAM sometimes obtains a power similar to ours, its FDR is much inflated and is far above the nominal level.

\begin{figure}[t!]
\centering
\includegraphics[width=3.0in,height=\myfigheight]{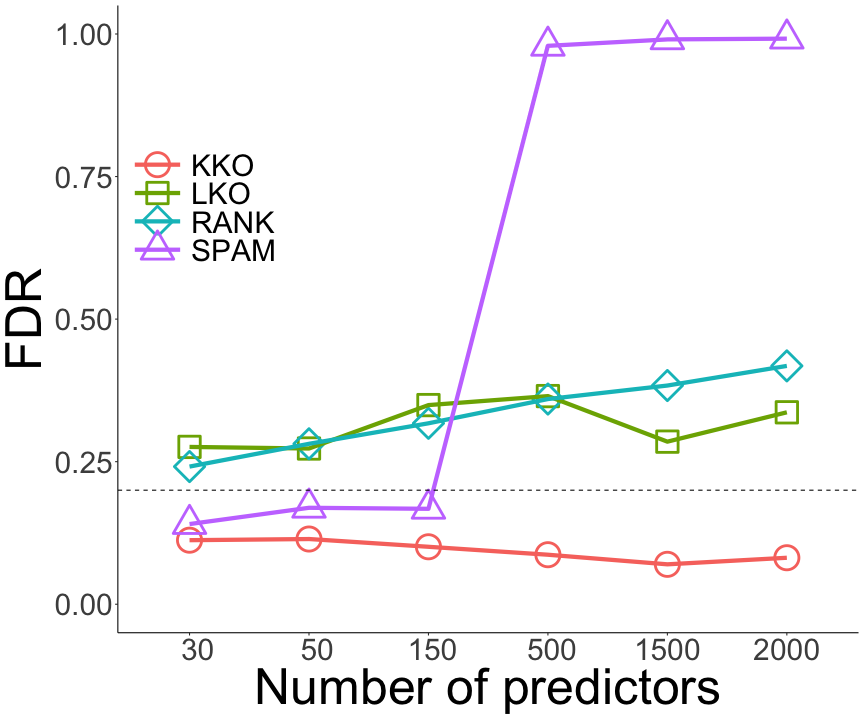}
\includegraphics[width=3.0in,height=\myfigheight]{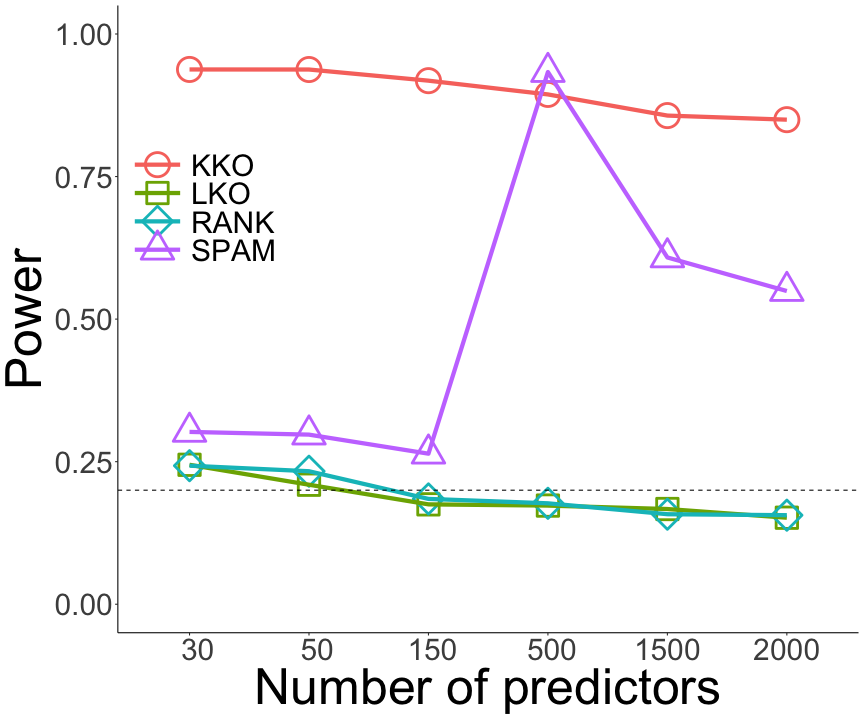}
\caption{Empirical performance and comparison in terms of FDR and power with the varying number of predictors $p$ in both regimes where $p < n$ and $p > n$ with $n=900$. The same four methods as in Figure \ref{fig: robust_distribution_SNR} are compared.}
\label{fig: robust_sparsity_ratio}
\end{figure}

%%%%%%%%%%%%%%%%%%%%%%%%%%%%%%%%%%%%%%%%%%%%%%%%%%%
\subsection{Brain imaging data analysis}
\label{sec:realdata}

We illustrate the proposed method with a brain imaging data analysis to study the Alzheimer's disease (AD). AD is an irreversible neurodegenerative disorder, and is characterized by progressive impairment of cognitive and memory functions, loss of bodily functions, and ultimately death. It is the leading form of dementia in elderly subjects, and is the sixth leading cause of death in the United States. Over 5.5 million Americans were affected by AD in 2018, and without any effective treatment or prevention, this number is projected to triple by 2050 \citep{AD2020}. Brain atrophy as reflected by brain grey matter cortical thickness is a well-known biomarker for AD. We study a dataset with $n = 697$ subjects, each of whom 
received an anatomical magnetic resonance imaging (MRI) scan that measures cortical thickness. The data is publicly available at \url{http://adni.loni.usc.edu/}. The MRI image has been preprocessed by the standard pipeline, and is summarized in the form of a vector of cortical thickness measurements for a set of parcellated brain regions-of-interest. There are $p=68$ regions in total. Brain parcellation is particularly useful to facilitate the interpretation, and has been frequently employed in brain imaging analysis \citep{Fornito2013, Kang2016}. In addition to the MRI image, for each subject, the data also records a composite cognitive score, which combines numerous tests that assess episodic memory, timed executive function, and global cognition \citep{Donohue2014}. Our goal is to study the association between the composite cognitive score and the vector of brain cortical thickness, and identify individual brain regions with strong associations. A linear model is inadequate to capture such an association, and we turn to the nonparametric additive model instead. We apply the proposed kernel knockoffs selection procedure. As the distribution of the predictors is not necessarily normal, we employ the deep knockoffs machine to generate the knockoff variables. We continue to set the target FDR level at $q=0.2$. 

\begin{table}[t!]
\centering
\begin{tabular}{ccccc} \toprule
l-middletemporal & l-superiorparietal & r-fusiform & r-inferiorparietal & l-entorhinal \\
l-fusiform & l-inferiorparietal & l-precuneus & l-superiortemporal & r-entorhinal \\ 
\bottomrule
\end{tabular}
\caption{Brain regions identified by the kernel knockoffs selection procedure. ``l-" stands for the left hemisphere, and ``r-" stands for the right hemisphere.}
\label{tab: ad-data}
\end{table}
 
Table \ref{tab: ad-data} reports the ten brain regions selected by our method. These findings agree with and support the current literature on AD research. Particularly, the middle temporal gyrus is located on the temporal lobe, and is associated with processes of recognition of known faces and accessing word meaning while reading. Middle temporal lobe atrophy is common in AD as well as its prodromal stage, mild cognitive impairment \citep{Visser2002}. The superior parietal lobe is involved with attention, visuospatial perception, and spatial orientation. Damage to the parietal lobe is common in AD, and leads to problems with performing gestures and skilled movements \citep{Pini2016}. The fusiform is linked with various neural pathways related to recognition. The inferior parietal lobe is involved in perception of emotions. The superior temporal gyrus is involved in auditory processing, and has also been implicated as a critical structure in social cognition. Numerous studies have found involvement of these brain regions in the development of AD \citep{Convit2000, Du2007, Pini2016}. The precuneus is associated with episodic memory, visuospatial processing, reflections upon self, and aspects of consciousness, and is found to be an AD-signature region \citep{Bakkour2013}. Finally, the entorhinal cortex functions as a hub in a widespread network for memory, navigation and the perception of time. It is found implicated in the early stages of AD, and is one of the most heavily damaged cortices in AD \citep{Hoesen1991}.

\section*{Acknowledgments}
The authors thank the Editor, the Associate Editor, and three referees for their constructive comments and insightful suggestions.  Li's research was partially supported by the NSF grant CIF-2102227, and the NIH grants R01AG061303, R01AG062542, and R01AG034570.

%%%%%%%%%%%%%%%%%%%%%%%%%%%%%%%%%%%%%%%%%%%%%%%%%%%
\baselineskip=19pt
\bibliographystyle{ims}
\bibliography{ref-kko}

\begin{thebibliography}{62}
\expandafter\ifx\csname natexlab\endcsname\relax\def\natexlab#1{#1}\fi
\expandafter\ifx\csname url\endcsname\relax
  \def\url#1{\texttt{#1}}\fi
\expandafter\ifx\csname urlprefix\endcsname\relax\def\urlprefix{URL }\fi

\bibitem[{Ahmed et~al.(2011)Ahmed, Hartikainen, J{\"a}rvelin and
  Richardson}]{ahmed2011false}
\textsc{Ahmed, I.}, \textsc{Hartikainen, A.-L.}, \textsc{J{\"a}rvelin, M.-R.}
  and \textsc{Richardson, S.} (2011).
\newblock False discovery rate estimation for stability selection: application
  to genome-wide association studies.
\newblock \textit{Statistical Applications in Genetics and Molecular Biology}
  \textbf{10}.

\bibitem[{{Alzheimer's Association}(2020)}]{AD2020}
\textsc{{Alzheimer's Association}} (2020).
\newblock 2020 {A}lzheimer's disease facts and figures.
\newblock \textit{Alzheimer's \& Dementia} \textbf{16} 391--460.

\bibitem[{Aronszajn(1950)}]{aronszajn1950theory}
\textsc{Aronszajn, N.} (1950).
\newblock Theory of reproducing kernels.
\newblock \textit{Transactions of the American Mathematical Society}
  \textbf{68} 337--404.

\bibitem[{Bach(2017)}]{bach2017equivalence}
\textsc{Bach, F.} (2017).
\newblock On the equivalence between kernel quadrature rules and random feature
  expansions.
\newblock \textit{The Journal of Machine Learning Research} \textbf{18}
  714--751.

\bibitem[{Bach(2008)}]{B08}
\textsc{Bach, F.~R.} (2008).
\newblock Bolasso: model consistent lasso estimation through the bootstrap.
\newblock In \textit{Proceedings of the 25th International Conference on
  Machine Learning}.

\bibitem[{Bakkour et~al.(2013)Bakkour, Morris, Wolk and
  Dickerson}]{Bakkour2013}
\textsc{Bakkour, A.}, \textsc{Morris, J.~C.}, \textsc{Wolk, D.~A.} and
  \textsc{Dickerson, B.~C.} (2013).
\newblock The effects of aging and alzheimer's disease on cerebral cortical
  anatomy: Specificity and differential relationships with cognition.
\newblock \textit{NeuroImage} \textbf{76} 332--344.

\bibitem[{Barber and Cand{\`e}s(2015)}]{BC15}
\textsc{Barber, R.~F.} and \textsc{Cand{\`e}s, E.~J.} (2015).
\newblock Controlling the false discovery rate via knockoffs.
\newblock \textit{The Annals of Statistics} \textbf{43} 2055--2085.

\bibitem[{Barber and Cand{\`e}s(2019)}]{BC19}
\textsc{Barber, R.~F.} and \textsc{Cand{\`e}s, E.~J.} (2019).
\newblock A knockoff filter for high-dimensional selective inference.
\newblock \textit{The Annals of Statistics} \textbf{47} 2504--2537.

\bibitem[{Barber et~al.(2020)Barber, Cand{\`e}s and Samworth}]{BCS20}
\textsc{Barber, R.~F.}, \textsc{Cand{\`e}s, E.~J.} and \textsc{Samworth, R.~J.}
  (2020).
\newblock Robust inference with knockoffs.
\newblock \textit{The Annals of Statistics} \textbf{48} 1409--1431.

\bibitem[{B{\u{a}}z{\u{a}}van et~al.(2012)B{\u{a}}z{\u{a}}van, Li and
  Sminchisescu}]{buazuavan2012fourier}
\textsc{B{\u{a}}z{\u{a}}van, E.~G.}, \textsc{Li, F.} and \textsc{Sminchisescu,
  C.} (2012).
\newblock \textit{Fourier Kernel Learning}.
\newblock New York: Springer.

\bibitem[{Benjamini and Hochberg(1995)}]{BH95}
\textsc{Benjamini, Y.} and \textsc{Hochberg, Y.} (1995).
\newblock Controlling the false discovery rate: a practical and powerful
  approach to multiple testing.
\newblock \textit{Journal of the Royal statistical Society: Series B
  (Methodological)} \textbf{57} 289--300.

\bibitem[{Bochner(1934)}]{bochner1934theorem}
\textsc{Bochner, S.} (1934).
\newblock A theorem on fourier-stieltjes integrals.
\newblock \textit{Bulletin of the American Mathematical Society} \textbf{40}
  271--276.

\bibitem[{Cai and Sun(2017)}]{CaiSun2017}
\textsc{Cai, T.~T.} and \textsc{Sun, W.} (2017).
\newblock Large-scale global and simultaneous inference: Estimation and testing
  in very high dimensions.
\newblock \textit{Annual Review of Economics} \textbf{9} 411--439.

\bibitem[{Candès et~al.(2018)Candès, Fan, Janson and Lv}]{candes2018panning}
\textsc{Candès, E.}, \textsc{Fan, Y.}, \textsc{Janson, L.} and \textsc{Lv, J.}
  (2018).
\newblock Panning for gold: `model-x' knockoffs for high dimensional controlled
  variable selection.
\newblock \textit{Journal of the Royal Statistical Society Series B}
  \textbf{80} 551--577.

\bibitem[{Convit et~al.(2000)Convit, de~Asis et~al.}]{Convit2000}
\textsc{Convit, A.}, \textsc{de~Asis, J.} \textsc{et~al.} (2000).
\newblock Atrophy of the medial occipitotemporal, inferior, and middle temporal
  gyri in non-demented elderly predict decline to alzheimer's disease.
\newblock \textit{Neurobiology of Aging} \textbf{21} 19--26.

\bibitem[{Dai and Barber(2016)}]{DB16}
\textsc{Dai, R.} and \textsc{Barber, R.} (2016).
\newblock The knockoff filter for fdr control in group-sparse and multitask
  regression.
\newblock In \textit{International Conference on Machine Learning}. PMLR.

\bibitem[{Dai and Li(2021)}]{dai2021kernel}
\textsc{Dai, X.} and \textsc{Li, L.} (2021).
\newblock Kernel ordinary differential equations.
\newblock \textit{Journal of the American Statistical Association}  1--35.

\bibitem[{Donohue et~al.(2014)Donohue, Sperling and Others}]{Donohue2014}
\textsc{Donohue, M.~C.}, \textsc{Sperling, R.~A.} and \textsc{Others} (2014).
\newblock {The Preclinical Alzheimer Cognitive Composite: Measuring
  Amyloid-Related Decline}.
\newblock \textit{Journal of the American Medical Association: Neurology}
  \textbf{71} 961--970.

\bibitem[{Du et~al.(2007)Du, Schuff, Kramer, Rosen, Gorno-Tempini, Rankin,
  Miller and Weiner}]{Du2007}
\textsc{Du, A.-T.}, \textsc{Schuff, N.}, \textsc{Kramer, J.~H.}, \textsc{Rosen,
  H.~J.}, \textsc{Gorno-Tempini, M.~L.}, \textsc{Rankin, K.}, \textsc{Miller,
  B.~L.} and \textsc{Weiner, M.~W.} (2007).
\newblock {Different regional patterns of cortical thinning in Alzheimer's
  disease and frontotemporal dementia}.
\newblock \textit{Brain} \textbf{130} 1159--1166.

\bibitem[{D{\"u}mbgen et~al.(2013)D{\"u}mbgen, Samworth, Schuhmacher
  et~al.}]{dumbgen2013stochastic}
\textsc{D{\"u}mbgen, L.}, \textsc{Samworth, R.~J.}, \textsc{Schuhmacher, D.}
  \textsc{et~al.} (2013).
\newblock Stochastic search for semiparametric linear regression models.
\newblock In \textit{From Probability to Statistics and Back: High-Dimensional
  Models and Processes--A Festschrift in Honor of Jon A. Wellner}. Institute of
  Mathematical Statistics, 78--90.

\bibitem[{Efron et~al.(2001)Efron, Tibshirani, Storey and Tusher}]{Efron2001}
\textsc{Efron, B.}, \textsc{Tibshirani, R.}, \textsc{Storey, J.~D.} and
  \textsc{Tusher, V.} (2001).
\newblock Empirical bayes analysis of a microarray experiment.
\newblock \textit{Journal of the American Statistical Association} \textbf{96}
  1151--1160.

\bibitem[{Fan et~al.(2020)Fan, Demirkaya, Li and Lv}]{FDLL20}
\textsc{Fan, Y.}, \textsc{Demirkaya, E.}, \textsc{Li, G.} and \textsc{Lv, J.}
  (2020).
\newblock Rank: large-scale inference with graphical nonlinear knockoffs.
\newblock \textit{Journal of the American Statistical Association} \textbf{115}
  362--379.

\bibitem[{Fornito et~al.(2013)Fornito, Zalesky and Breakspear}]{Fornito2013}
\textsc{Fornito, A.}, \textsc{Zalesky, A.} and \textsc{Breakspear, M.} (2013).
\newblock Graph analysis of the human connectome: Promise, progress, and
  pitfalls.
\newblock \textit{NeuroImage} \textbf{80} 426--444.

\bibitem[{Hastie and Tibshirani(1990)}]{hastie1990generalized}
\textsc{Hastie, T.~J.} and \textsc{Tibshirani, R.~J.} (1990).
\newblock \textit{Generalized Additive Models}, vol.~43.
\newblock London: Chapman and Hall.

\bibitem[{He et~al.(2016)He, Li, Zhu, Liu, Lee, Amos, Hyslop, Jin, Lin and
  Wei}]{he2016component}
\textsc{He, K.}, \textsc{Li, Y.}, \textsc{Zhu, J.}, \textsc{Liu, H.},
  \textsc{Lee, J.~E.}, \textsc{Amos, C.~I.}, \textsc{Hyslop, T.}, \textsc{Jin,
  J.}, \textsc{Lin, H.} and \textsc{Wei, Q.} (2016).
\newblock Component-wise gradient boosting and false discovery control in
  survival analysis with high-dimensional covariates.
\newblock \textit{Bioinformatics} \textbf{32} 50--57.

\bibitem[{Huang et~al.(2010)Huang, Horowitz and Wei}]{HHW10}
\textsc{Huang, J.}, \textsc{Horowitz, J.~L.} and \textsc{Wei, F.} (2010).
\newblock Variable selection in nonparametric additive models.
\newblock \textit{Annals of Statistics} \textbf{38} 2282.

\bibitem[{Kang et~al.(2016)Kang, Bowman, Mayberg and Liu}]{Kang2016}
\textsc{Kang, J.}, \textsc{Bowman, F.~D.}, \textsc{Mayberg, H.} and
  \textsc{Liu, H.} (2016).
\newblock A depression network of functionally connected regions discovered via
  multi-attribute canonical correlation graphs.
\newblock \textit{NeuroImage} \textbf{141} 431--441.

\bibitem[{Koltchinskii and Yuan(2010)}]{koltchinskii2010sparsity}
\textsc{Koltchinskii, V.} and \textsc{Yuan, M.} (2010).
\newblock Sparsity in multiple kernel learning.
\newblock \textit{The Annals of Statistics} \textbf{38} 3660--3695.

\bibitem[{Li et~al.(2005)Li, Cook and Nachtsheim}]{LiCN2005}
\textsc{Li, L.}, \textsc{Cook, R.~D.} and \textsc{Nachtsheim, C.~J.} (2005).
\newblock Model-free variable selection.
\newblock \textit{Journal of the Royal Statistical Society. Series B
  (Statistical Methodology)} \textbf{67} 285--299.

\bibitem[{Li et~al.(2013)Li, Hsu, Peng and Wang}]{LHPW13}
\textsc{Li, S.}, \textsc{Hsu, L.}, \textsc{Peng, J.} and \textsc{Wang, P.}
  (2013).
\newblock Bootstrap inference for network construction with an application to a
  breast cancer microarray study.
\newblock \textit{The Annals of Applied Statistics} \textbf{7} 391.

\bibitem[{Lin and Zhang(2006)}]{LZ06}
\textsc{Lin, Y.} and \textsc{Zhang, H.~H.} (2006).
\newblock Component selection and smoothing in multivariate nonparametric
  regression.
\newblock \textit{The Annals of Statistics} \textbf{34} 2272--2297.

\bibitem[{Loh and Wainwright(2012)}]{loh2012}
\textsc{Loh, P.-L.} and \textsc{Wainwright, M.~J.} (2012).
\newblock High-dimensional regression with noisy and missing data: Provable
  guarantees with nonconvexity.
\newblock \textit{Annals of Statistics} \textbf{40} 1637--1664.

\bibitem[{Meier et~al.(2009)Meier, Van~de Geer and
  B{\"u}hlmann}]{meier2009high}
\textsc{Meier, L.}, \textsc{Van~de Geer, S.} and \textsc{B{\"u}hlmann, P.}
  (2009).
\newblock High-dimensional additive modeling.
\newblock \textit{The Annals of Statistics} \textbf{37} 3779--3821.

\bibitem[{Meinshausen and B{\"u}hlmann(2010)}]{MB10}
\textsc{Meinshausen, N.} and \textsc{B{\"u}hlmann, P.} (2010).
\newblock Stability selection.
\newblock \textit{Journal of the Royal Statistical Society: Series B
  (Statistical Methodology)} \textbf{72} 417--473.

\bibitem[{Mercer(1909)}]{mercer1909functions}
\textsc{Mercer, J.} (1909).
\newblock Functions of positive and negative type and their connection with the
  theory of integral equations.
\newblock \textit{Philosophical Transactions of the Royal Society, London A}
  415--446.

\bibitem[{Miller(2002)}]{miller2002subset}
\textsc{Miller, A.} (2002).
\newblock \textit{Subset Selection in Regression}.
\newblock CRC Press.

\bibitem[{Pini et~al.(2016)Pini, Pievani, Bocchetta, Altomare, Bosco, Cavedo,
  Galluzzi, Marizzoni and Frisoni}]{Pini2016}
\textsc{Pini, L.}, \textsc{Pievani, M.}, \textsc{Bocchetta, M.},
  \textsc{Altomare, D.}, \textsc{Bosco, P.}, \textsc{Cavedo, E.},
  \textsc{Galluzzi, S.}, \textsc{Marizzoni, M.} and \textsc{Frisoni, G.~B.}
  (2016).
\newblock Brain atrophy in alzheimer's disease and aging.
\newblock \textit{Aging Research Reviews} \textbf{30} 25--48.
\newblock Brain Imaging and Aging.

\bibitem[{Rahimi and Recht(2007)}]{RR07}
\textsc{Rahimi, A.} and \textsc{Recht, B.} (2007).
\newblock Random features for large-scale kernel machines.
\newblock \textit{Advances in Neural Information Processing systems}
  \textbf{20} 1177--1184.

\bibitem[{Raskutti et~al.(2012)Raskutti, J~Wainwright and
  Yu}]{raskutti2012minimax}
\textsc{Raskutti, G.}, \textsc{J~Wainwright, M.} and \textsc{Yu, B.} (2012).
\newblock Minimax-optimal rates for sparse additive models over kernel classes
  via convex programming.
\newblock \textit{Journal of Machine Learning Research} \textbf{13}.

\bibitem[{Raskutti et~al.(2011)Raskutti, Wainwright and
  Yu}]{raskutti2011minimax}
\textsc{Raskutti, G.}, \textsc{Wainwright, M.~J.} and \textsc{Yu, B.} (2011).
\newblock Minimax rates of estimation for high-dimensional linear regression
  over $\ell_q $-balls.
\newblock \textit{IEEE Transactions on Information Theory} \textbf{57}
  6976--6994.

\bibitem[{Ravikumar et~al.(2009)Ravikumar, Lafferty, Liu and
  Wasserman}]{RLLW09}
\textsc{Ravikumar, P.}, \textsc{Lafferty, J.}, \textsc{Liu, H.} and
  \textsc{Wasserman, L.} (2009).
\newblock Sparse additive models.
\newblock \textit{Journal of the Royal Statistical Society: Series B
  (Statistical Methodology)} \textbf{71} 1009--1030.

\bibitem[{Ravikumar et~al.(2010)Ravikumar, Wainwright and
  Lafferty}]{ravikumar2010high}
\textsc{Ravikumar, P.}, \textsc{Wainwright, M.~J.} and \textsc{Lafferty, J.~D.}
  (2010).
\newblock High-dimensional ising model selection using $l_1$-regularized
  logistic regression.
\newblock \textit{The Annals of Statistics} \textbf{38} 1287--1319.

\bibitem[{Ren et~al.(2020)Ren, Wei and Cand{\`e}s}]{RWC20}
\textsc{Ren, Z.}, \textsc{Wei, Y.} and \textsc{Cand{\`e}s, E.} (2020).
\newblock Derandomizing knockoffs.
\newblock \textit{arXiv preprint arXiv:2012.02717} .

\bibitem[{Romano et~al.(2019)Romano, Sesia and Cand{\`e}s}]{romano2019deep}
\textsc{Romano, Y.}, \textsc{Sesia, M.} and \textsc{Cand{\`e}s, E.} (2019).
\newblock Deep knockoffs.
\newblock \textit{Journal of the American Statistical Association}  1--12.

\bibitem[{Rudi and Rosasco(2017)}]{rudi2017generalization}
\textsc{Rudi, A.} and \textsc{Rosasco, L.} (2017).
\newblock Generalization properties of learning with random features.
\newblock In \textit{NIPS}.

\bibitem[{Sch{\"o}lkopf and Smola(2002)}]{scholkopf2002learning}
\textsc{Sch{\"o}lkopf, B.} and \textsc{Smola, A.~J.} (2002).
\newblock \textit{Learning with Kernels: Support Vector Machines,
  Regularization, Optimization, and Beyond}.
\newblock MIT press.

\bibitem[{Stone(1985)}]{stone1985additive}
\textsc{Stone, C.~J.} (1985).
\newblock Additive regression and other nonparametric models.
\newblock \textit{The Annals of Statistics}  689--705.

\bibitem[{Storey(2007)}]{Storey2007}
\textsc{Storey, J.~D.} (2007).
\newblock The optimal discovery procedure: a new approach to simultaneous
  significance testing.
\newblock \textit{Journal of the Royal Statistical Society: Series B
  (Statistical Methodology)} \textbf{69} 347--368.

\bibitem[{Su et~al.(2017)Su, Bogdan and Candes}]{su2017false}
\textsc{Su, W.}, \textsc{Bogdan, M.} and \textsc{Candes, E.} (2017).
\newblock False discoveries occur early on the lasso path.
\newblock \textit{The Annals of statistics}  2133--2150.

\bibitem[{Sun and Cai(2007)}]{SunCai2007}
\textsc{Sun, W.} and \textsc{Cai, T.~T.} (2007).
\newblock Oracle and adaptive compound decision rules for false discovery rate
  control.
\newblock \textit{Journal of the American Statistical Association} \textbf{102}
  901--912.

\bibitem[{Sun and Cai(2009)}]{SunCai2009}
\textsc{Sun, W.} and \textsc{Cai, T.~T.} (2009).
\newblock Large-scale multiple testing under dependence.
\newblock \textit{Journal of the Royal Statistical Society: Series B
  (Statistical Methodology)} \textbf{71} 393--424.

\bibitem[{van~de Geer(2000)}]{Vandegeer2000}
\textsc{van~de Geer, S.} (2000).
\newblock \textit{Empirical Processes in $M$-Estimation}.
\newblock Cambridge University Press.

\bibitem[{van~de Geer(2002)}]{G02}
\textsc{van~de Geer, S.~A.} (2002).
\newblock On hoeffding’s inequality for dependent random variables.
\newblock In \textit{Empirical process techniques for dependent data}.
  Springer, 161--169.

\bibitem[{van Hoesen et~al.(1991)van Hoesen, Hyman and Damasio}]{Hoesen1991}
\textsc{van Hoesen, G.~W.}, \textsc{Hyman, B.~T.} and \textsc{Damasio, A.~R.}
  (1991).
\newblock Entorhinal cortex pathology in alzheimer's disease.
\newblock \textit{Hippocampus} \textbf{1} 1--8.

\bibitem[{Visser et~al.(2002)Visser, Verhey, Hofman, Scheltens and
  Jolles}]{Visser2002}
\textsc{Visser, P.~J.}, \textsc{Verhey, F. R.~J.}, \textsc{Hofman, P. A.~M.},
  \textsc{Scheltens, P.} and \textsc{Jolles, J.} (2002).
\newblock Medial temporal lobe atrophy predicts alzheimer{\textquoteright}s
  disease in patients with minor cognitive impairment.
\newblock \textit{Journal of Neurology, Neurosurgery \& Psychiatry} \textbf{72}
  491--497.

\bibitem[{Wahba(1990)}]{wahba1990spline}
\textsc{Wahba, G.} (1990).
\newblock \textit{Spline Models for Observational Data}.
\newblock SIAM.

\bibitem[{Wahba et~al.(1995)Wahba, Wang, Gu, Klein and
  Klein}]{wahba1995smoothing}
\textsc{Wahba, G.}, \textsc{Wang, Y.}, \textsc{Gu, C.}, \textsc{Klein, R.} and
  \textsc{Klein, B.} (1995).
\newblock Smoothing spline anova for exponential families, with application to
  the wisconsin epidemiological study of diabetic retinopathy.
\newblock \textit{The Annals of Statistics} \textbf{23} 1865--1895.

\bibitem[{Weinstein et~al.(2020)Weinstein, Su, Bogdan, Barber and
  Cand{\`e}s}]{weinstein2020power}
\textsc{Weinstein, A.}, \textsc{Su, W.~J.}, \textsc{Bogdan, M.},
  \textsc{Barber, R.~F.} and \textsc{Cand{\`e}s, E.~J.} (2020).
\newblock A power analysis for knockoffs with the lasso coefficient-difference
  statistic.
\newblock \textit{arXiv preprint arXiv:2007.15346} .

\bibitem[{Wood(2017)}]{wood2017generalized}
\textsc{Wood, S.~N.} (2017).
\newblock \textit{Generalized Additive Models: An Introduction With R}.
\newblock CRC press.

\bibitem[{Wu et~al.(2007)Wu, Boos and Stefanski}]{wu2007controlling}
\textsc{Wu, Y.}, \textsc{Boos, D.~D.} and \textsc{Stefanski, L.~A.} (2007).
\newblock Controlling variable selection by the addition of pseudovariables.
\newblock \textit{Journal of the American Statistical Association} \textbf{102}
  235--243.

\bibitem[{Yuan and Zhou(2016)}]{Yuan2016}
\textsc{Yuan, M.} and \textsc{Zhou, D.-X.} (2016).
\newblock Minimax optimal rates of estimation in high dimensional additive
  models.
\newblock \textit{Annals of Statistics} \textbf{44} 2564--2593.

\bibitem[{Zhao and Yu(2006)}]{Zhao2006}
\textsc{Zhao, P.} and \textsc{Yu, B.} (2006).
\newblock On model selection consistency of lasso.
\newblock \textit{Journal of Machine Learning Research} \textbf{7} 2541--2563.

\end{thebibliography}


\begin{thebibliography}{14}
\expandafter\ifx\csname natexlab\endcsname\relax\def\natexlab#1{#1}\fi
\expandafter\ifx\csname url\endcsname\relax
  \def\url#1{\texttt{#1}}\fi
\expandafter\ifx\csname urlprefix\endcsname\relax\def\urlprefix{URL }\fi

\bibitem[{Bach(2017)}]{bach2017equivalence}
\textsc{Bach, F.} (2017).
\newblock On the equivalence between kernel quadrature rules and random feature
  expansions.
\newblock \textit{The Journal of Machine Learning Research} \textbf{18}
  714--751.

\bibitem[{Barber and Cand{\`e}s(2015)}]{BC15}
\textsc{Barber, R.~F.} and \textsc{Cand{\`e}s, E.~J.} (2015).
\newblock Controlling the false discovery rate via knockoffs.
\newblock \textit{The Annals of Statistics} \textbf{43} 2055--2085.

\bibitem[{B{\u{a}}z{\u{a}}van et~al.(2012)B{\u{a}}z{\u{a}}van, Li and
  Sminchisescu}]{buazuavan2012fourier}
\textsc{B{\u{a}}z{\u{a}}van, E.~G.}, \textsc{Li, F.} and \textsc{Sminchisescu,
  C.} (2012).
\newblock \textit{Fourier Kernel Learning}.
\newblock New York: Springer.

\bibitem[{Candès et~al.(2018)Candès, Fan, Janson and Lv}]{candes2018panning}
\textsc{Candès, E.}, \textsc{Fan, Y.}, \textsc{Janson, L.} and \textsc{Lv, J.}
  (2018).
\newblock Panning for gold: `model-x' knockoffs for high dimensional controlled
  variable selection.
\newblock \textit{Journal of the Royal Statistical Society Series B}
  \textbf{80} 551--577.

\bibitem[{Dai and Li(2021)}]{dai2021kernel}
\textsc{Dai, X.} and \textsc{Li, L.} (2021).
\newblock Kernel ordinary differential equations.
\newblock \textit{Journal of the American Statistical Association}  1--35.

\bibitem[{Koltchinskii and Yuan(2010)}]{koltchinskii2010sparsity}
\textsc{Koltchinskii, V.} and \textsc{Yuan, M.} (2010).
\newblock Sparsity in multiple kernel learning.
\newblock \textit{The Annals of Statistics} \textbf{38} 3660--3695.

\bibitem[{Loh and Wainwright(2012)}]{loh2012}
\textsc{Loh, P.-L.} and \textsc{Wainwright, M.~J.} (2012).
\newblock High-dimensional regression with noisy and missing data: Provable
  guarantees with nonconvexity.
\newblock \textit{Annals of Statistics} \textbf{40} 1637--1664.

\bibitem[{Rahimi and Recht(2007)}]{RR07}
\textsc{Rahimi, A.} and \textsc{Recht, B.} (2007).
\newblock Random features for large-scale kernel machines.
\newblock \textit{Advances in Neural Information Processing systems}
  \textbf{20} 1177--1184.

\bibitem[{Ravikumar et~al.(2010)Ravikumar, Wainwright and
  Lafferty}]{ravikumar2010high}
\textsc{Ravikumar, P.}, \textsc{Wainwright, M.~J.} and \textsc{Lafferty, J.~D.}
  (2010).
\newblock High-dimensional ising model selection using $l_1$-regularized
  logistic regression.
\newblock \textit{The Annals of Statistics} \textbf{38} 1287--1319.

\bibitem[{Rudi and Rosasco(2017)}]{rudi2017generalization}
\textsc{Rudi, A.} and \textsc{Rosasco, L.} (2017).
\newblock Generalization properties of learning with random features.
\newblock In \textit{NIPS}.

\bibitem[{Sch{\"o}lkopf and Smola(2002)}]{scholkopf2002learning}
\textsc{Sch{\"o}lkopf, B.} and \textsc{Smola, A.~J.} (2002).
\newblock \textit{Learning with Kernels: Support Vector Machines,
  Regularization, Optimization, and Beyond}.
\newblock MIT press.

\bibitem[{van~de Geer(2000)}]{Vandegeer2000}
\textsc{van~de Geer, S.} (2000).
\newblock \textit{Empirical Processes in $M$-Estimation}.
\newblock Cambridge University Press.

\bibitem[{Wahba(1990)}]{wahba1990spline}
\textsc{Wahba, G.} (1990).
\newblock \textit{Spline Models for Observational Data}.
\newblock SIAM.

\bibitem[{Yuan and Zhou(2016)}]{Yuan2016}
\textsc{Yuan, M.} and \textsc{Zhou, D.-X.} (2016).
\newblock Minimax optimal rates of estimation in high dimensional additive
  models.
\newblock \textit{Annals of Statistics} \textbf{44} 2564--2593.

\end{thebibliography}

\newpage
\appendix
%%%%%%%%%%%%%%%%%%%%%%%%%%%%%%%%%%%%%%%%%%%%%%%%%%%
\section{Proofs}
\label{sec:proofs}

%%%%%%%%%%%%%%%%%%%%%%%%%%%%%%%%%%%%%%%%%%%%%%%%%%%
\subsection{Proof of Proposition \ref{prop:nullequivalence}}
\begin{proof}
Recall that, in the nonparametric additive model (\ref{eq: additivemodel}), the conditional distribution of  $Y$ given $\bX\in\Xcal^p$ follows a normal distribution, i.e., 
\begin{equation*}
p_{Y|\bX}(y|x_1,\ldots,x_p) = \frac{1}{\sigma\sqrt{2\pi}}\exp\left\{ -\frac{1}{2\sigma^2}\left[ y-\sum_{j=1}^pf_j(x_j) \right]^2 \right\}.
\end{equation*}
Let $\bx_{-j}=\{x_1,\ldots,x_p\}\backslash\{x_{j}\}$.
Note that 
\begin{equation*}
\begin{aligned}
    p_{Y,X_j|\bX_{-j}}(y,x_j|\bx_{-j})  =   p_{Y|X_j,\bX_{-j}}(y|x_j,\bx_{-j}) p_{X_j|\bX_{-j}}(x_j|\bx_{-j}).
\end{aligned}
\end{equation*}
Now if $f_j=  0$ for $j\in\{1,\ldots,p\}$, then
\begin{equation*}
\begin{aligned}
    p_{Y|X_j,\bX_{-j}}(y|x_j,\bx_{-j})  & =   p_{Y|\bX_{-j}}(y|\bx_{-j}).
    \end{aligned}
\end{equation*}
Henceforth,
\begin{equation}
\label{eqn:factorizing}
\begin{aligned}
p_{Y,X_j|\bX_{-j}}(y,x_j|\bx_{-j})  =   p_{Y|\bX_{-j}}(y|\bx_{-j}) p_{X_j|\bX_{-j}}(x_j|\bx_{-j}).
\end{aligned}
\end{equation}
That is, the conditional distribution $p_{Y,X_j|\bX_{-j}}(y,x_j|\bx_{-j})$ factorizes. Consequently, 
$Y \independent X_j | \bX_{-j}$ and thus, $j\in\cS^\perp$.

On the other hand, suppose that $j\in\cS^\perp$. By definition, $Y \independent X_j | \bX_{-j}$. 
Then the likelihood $p_{Y,X_j|\bX_{-j}}(y,x_j|\bx_{-j})$ need to factorize and satisfy (\ref{eqn:factorizing}). 
Note that $p_{Y,X_j|\bX_{-j}}(y,x_j|\bx_{-j})$ has an interaction term,
\begin{equation}
\label{eqn:interaction}
\frac{1}{\sigma\sqrt{2\pi}}\exp\left(\frac{1}{\sigma^2}f_j(x_j)\left(y-\sum_{k=1;k\neq j}^pf_k(x_k)\right)\right).
\end{equation} 
Since $p_{Y,X_j|\bX_{-j}}(y,x_j|\bx_{-j})$  satisfies (\ref{eqn:factorizing}), (\ref{eqn:interaction}) needs to be a constant that is independent of $\bx$ and $y$. Under Assumption \ref{eqn:perfect-relationship}, this is true only if $f_j=0$. This completes the proof of Proposition \ref{prop:nullequivalence}. 
\end{proof}

%%%%%%%%%%%%%%%%%%%%%%%%%%%%%%%%%%%%%%%%%%%%%%%%%%%
\subsection{Proof of Theorem \ref{prop:coinflip}}
\label{sec:profofthmcoin}
\begin{proof}
Recall that the selected set $\widehat{\mathcal S}(I)$ is obtained by (\ref{eq: kernel_reg}). Then  $\widehat{\mathcal S}(I)$ is a function of 
\begin{equation}
\label{eqn:defofSigmaI}
\bSigma(I)\equiv \begin{bmatrix}
\bPsi(x_{i_1,1})\trans & \cdots & \bPsi(x_{i_1,p})\trans & \bPsi(\widetilde{x}_{i_1,1})\trans & \cdots & \bPsi(\widetilde{x}_{i_1,p})\trans \\
\vdots &  & \vdots & \vdots &  & \vdots\\
\bPsi(x_{i_{|I|},1})\trans & \cdots & \bPsi(x_{i_{|I|},p})\trans  & \bPsi(\widetilde{x}_{i_{|I|},1})\trans & \cdots & \bPsi(\widetilde{x}_{i_{|I|},p})\trans 
\end{bmatrix}\in\R^{|I|\times 2pr},
\end{equation}
where $I=(i_1,\ldots,i_{|I|})$, and the feature vectors $\bPsi(x_{i,j}) = \left[ \psi_1(x_{i,j}),\ldots,\psi_r(x_{i,j}) \right]\trans \in \R^r$, $\bPsi(\widetilde{x}_{i,j})$ $= \left[ \psi_1(\widetilde{x}_{i,j}),\ldots,\psi_r(\widetilde{x}_{i,j}) \right]\trans \in \R^r$ with $i=i_1,\ldots,i_{|I|}$, and $j=1,\ldots,p$. Moreover, $\widehat{\mathcal S}(I)$  is a function of the response vector $\by(I)=(y_1,\ldots,y_{|I|})\trans\in\R^{|I|}$. Since the importance score $\Delta_j = \widehat{\Pi}_j- \widehat{\Pi}_{j+p}$ with $\widehat{\Pi}_j = \P\{j\in\widehat{\mathcal S}(I)\}$, $\Delta_j$ is a function of $\bSigma(I)$ and $\by(I)$. That is,
\begin{equation*}
\Delta_j =   \Delta_j\left(\bSigma(I),\by(I)\right).
\end{equation*}
Next, we complete the proof of Theorem \ref{prop:coinflip} in three steps.

\paragraph{Step 1.} We show the \textit{flip sign property} that, for any $\Acal_1\subseteq\{1,\ldots,p\}$, the importance score satisfies that, 
\begin{equation}
\label{eqn:flipsign}
     \Delta_j\left(\{\bSigma(I)\}_{\text{swap}(\Acal_1)},\by(I)\right) = 
\begin{cases}
\Delta_j\left(\bSigma(I),\by(I)\right), &j\not\in \Acal_1,\\
-\Delta_j\left(\bSigma(I),\by(I)\right), &j\in \Acal_1,
\end{cases}
\end{equation}
where $\text{swap}(j)$ is an operator that swaps the submatrix 
\begin{equation*}
\begin{bmatrix}
\bPsi(x_{i_1,j})\trans  \\
\vdots \\
\bPsi(x_{i_{|I|},j})\trans 
\end{bmatrix}\in\R^{|I|\times r}\quad \text{ with }\quad 
\begin{bmatrix}
\bPsi(\widetilde{x}_{i_1,j})\trans  \\
\vdots \\
\bPsi(\widetilde{x}_{i_{|I|},j})\trans 
\end{bmatrix}\in\R^{|I|\times r}
\end{equation*}

Note that after the operation $\text{swap}(\Acal_1)$, the selection probability $\widehat{\Pi}_j^{\text{new}}$ obtained by (\ref{eqn:calofpi}) satisfies
\begin{equation*}
\widehat{\Pi}^{\text{new}}_j = 
\begin{cases}
\widehat{\Pi}_j\quad&j\not\in \Acal_1,\\
\widehat{\Pi}_{j+p} \text{ or }\widehat{\Pi}_{j-p}\quad&j\in \Acal_1,
\end{cases}
\end{equation*}
Since $\Delta_j = \widehat{\Pi}_j- \widehat{\Pi}_{j+p}$ is an asymmetric function, we have (\ref{eqn:flipsign}) holds.

\paragraph{Step 2.} We show that the \textit{exchangeability} holds, in that
\begin{equation}
\label{eqn:exchangeability}
\left(\{\bSigma(I)\}_{\text{swap}(\Acal_2)},\by(I)\right) \overset{d}{=} \left(\bSigma(I),\by(I)\right),
\end{equation}
where $\Acal_2\subseteq\cS^\perp$ is any set of null variables. Without loss of generality, assume that $\Acal_2=\{1,2,\ldots, |\cS^{\perp}|\}$. Since the rows of $(\bSigma(I),\by(I))\in\R^{|I|\times(2pr+1)}$ are independent, it suffices to show that
\begin{equation*}
   \left(\{\bPsi_{2p}(\bX)\trans\}_{\text{swap}(\Acal_2)},Y\right) \overset{d}{=} \left(\bPsi_{2p}(\bX)\trans,Y\right),
\end{equation*}
where $\bPsi_{2p}(\bX)\trans$ is a row of the matrix $\bSigma(I)$, i.e., 
\begin{equation*}
\bPsi_{2p}(\bX)\trans \equiv \left[\bPsi(X_{1})\trans,\ldots, \bPsi(X_{p})\trans,\bPsi(\widetilde{X}_{1})\trans,\ldots,\bPsi(\widetilde{X}_{p})\trans \right] \in \R^{1\times 2pr}.
\end{equation*}
Furthermore, since $\{\bPsi\trans_{2p}(\bX)\}_{\text{swap}(\Acal_2)}\overset{d}{=}\bPsi\trans_{2p}(\bX)$, we only need to establish that
\begin{equation}
\label{eqn:indepYpsi2p}
Y|\{\bPsi\trans_{2p}(\bX)\}_{\text{swap}(\Acal_2)} \overset{d}{=} Y|\bPsi\trans_{2p}(\bX).
\end{equation}

Let $\bPsi_p(\widetilde{\bX}) = \left[\bPsi(\tilde{X}_1)\trans,\ldots,\bPsi(\tilde{X}_p)\trans \right]\trans \in \R^{ pr}$ and $\bPsi_p(\bX) = \left[\bPsi(X_1)\trans,\ldots,\bPsi(X_p)\trans \right]\trans \in \R^{ pr}$. Here for each $j=1,\ldots,p$, $\bPsi(X_j) = \left[ \psi_1(X_j),\ldots,\psi_r(X_j) \right]\trans \in \R^r$, where $\psi_k(X_j)\in\Hcal_1$, with $k=1,\ldots,r$. Since $Y \independent \widetilde{\bX}|\bX$, by Proposition \ref{prop:nullequivalence}, we have that, 
\begin{equation}
\label{eqn:indepxtildex}
Y\independent \bPsi_p(\widetilde{\bX})| \bPsi_p(\bX),
\end{equation} 
The conditional distribution  of $Y$ satisfies, 
\begin{equation*}
\begin{aligned}
p_{Y|\{\bPsi_{2p}(\bX)\}_{\text{swap}(\Acal_2)\trans}}\left(y|\bPsi_{2p}(\bx)\trans\right) & =   p_{Y|\bPsi_{2p}(\bX)\trans}\left(y|\{\bPsi_{2p}(\bx)\trans\}_{\text{swap}(\Acal_2)}\right)\\
& = p_{Y|\bPsi_{p}(\bX)\trans}\left(y|\bPsi_{p}(\bx')\trans\right),
\end{aligned}
\end{equation*} 
where $\bx'=(x_1',\ldots,x_p')\in\R^{p}$ is a vector defined as $x_j'=\widetilde{x}_{j}$ if $j\in \Acal_2$ and $x_{j}'=x_{j}$ otherwise. The second equality above comes from (\ref{eqn:indepxtildex}), and the assumption that $\Acal_2$ is a set of null variables. Since $X_1 \in \Acal_2$, $Y \indep X_1$ given the rest of variables $X_2$ to $X_p$. Then, by Proposition \ref{prop:nullequivalence}, we have,
\begin{equation*}
Y \independent \bPsi(X_1)| [\bPsi(X_2)\trans,\ldots,\bPsi(X_p)\trans ].
\end{equation*} 
Then
\begin{equation*}
\begin{aligned}
p_{Y|\bPsi\trans_{p}(\bX)}\left(y|[\bPsi(\tilde{x}_1)\trans, \bPsi(x'_2)\trans,\ldots,\bPsi(x'_p)\trans]\right) 
= \; & p_{Y|[\bPsi(X_2)\trans,\ldots,\bPsi(X_p)\trans]}\left(y|[\bPsi(x'_2)\trans,\ldots,\bPsi(x'_p)\trans]\right) \\
= \; & p_{Y|\bPsi\trans_{p}(\bX)}\left(y|[\bPsi(x_1)\trans, \bPsi(x'_2)\trans,\ldots,\bPsi(x'_p)\trans]\right).
\end{aligned}
\end{equation*}
This shows that
\begin{equation*}
    Y|\{\bPsi\trans_{2p}(\bX)\}_{\text{swap}(\Acal_2)} \overset{d}{=} Y|\{\bPsi\trans_{2p}(\bX)\}_{\text{swap}(\Acal_2\backslash\{1\})}.
\end{equation*}
By repeating the argument with $\text{swap}(\Acal_2\backslash\{1,2\})$, $\text{swap}(\Acal_2\backslash\{1,2,3\})$, and so on, until $\Acal_2$ is empty, completes the proof of (\ref{eqn:indepYpsi2p}).

\begin{remark}
For the kernel learning based on the original variables in Section \ref{sec:background}, without using the random feature mapping, the representer theorem \citep{wahba1990spline} suggests that the estimator satisfies
\begin{equation*}
\widetilde{f}(\bX) = \sum_{i=1}^n\alpha_iK_p(\bX,\bx_i) + \sum_{i=1}^n\alpha_{i+n}K_p(\bX,\widetilde{\bx}_i).
\end{equation*}
Then the selected set $\widehat{\mathcal S}(I)$ by (\ref{eq: kernel_reg}) is a function of 
\begin{equation*}
\check{\bSigma}(I)\equiv \begin{bmatrix}
K_p(\bx_{i_1},\bx_{i_1}) & \cdots & K_p(\bx_{i_1},\bx_{i_{|I|}}) & K_p(\bx_{i_1},\widetilde{\bx}_{i_1}) & \cdots & K_p(\bx_{i_1},\widetilde{\bx}_{i_{|I|}})\\
\vdots &  & \vdots & \vdots &  & \vdots\\
K_p(\bx_{i_{|I|}},\bx_{i_1}) & \cdots & K_p(\bx_{i_{|I|}},\bx_{i_{|I|}}) & K_p(\bx_{i_{|I|}},\widetilde{\bx}_{i_1}) & \cdots & K_p(\bx_{i_{|I|}},\widetilde{\bx}_{i_{|I|}})
\end{bmatrix}\in\R^{|I|\times 2|I|},
\end{equation*}
where $I=(i_1,\ldots,i_{|I|})$.
For a general kernel function $K_p$, both the response $Y$ and the kernel function $K_p(\bX,\widetilde{\bX})$ are dependent on $\bX$. Therefore, 
\begin{equation*}
Y \not\perp \! \!\! \perp K_p(\bX,\widetilde{\bX})| \bX,
\end{equation*} 
which violates the conditional independence in (\ref{eqn:indepxtildex}). Consequently, the exchangeability result (\ref{eqn:exchangeability}) may not hold without using the random feature mapping. 

We tackle this problem explicitly by adopting the feature expansion for the kernel $K_p$. The feature expansion  ensures the estimator to follow a finite linear combination of basis functions as in (\ref{eqn:rfrepresentation}). As we show in the proof above, the conditional independence (\ref{eqn:indepxtildex}) holds.
\end{remark}

\paragraph{Step 3.} We next combine the results in the previous two steps to complete the proof. Let a sequence of independent random variables  $\bs=(s_1,\ldots,s_p)$ satisfy $s_j=\pm1$ with probability $1/2$ if $j\in\cS^\perp$, and $s_j=1$ otherwise. Let $\Acal_3=\{j:s_j=-1\}$. Then $\Acal_3$ only consists of some null variables. Define the swapped importance score as,
\begin{equation*}
\Delta_{\text{swap}(\Acal_3)} \equiv  \Delta\left(\{\bSigma(I)\}_{\text{swap}(\Acal_3)},\by(I)\right)\in\R^p.
\end{equation*}
By the flip sign property in Step 1,
\begin{equation*}
\Delta_{\text{swap}(\Acal_3)}=\bs\odot \Delta\left(\bSigma(I),\by(I)\right),
\end{equation*}
where $\bs\odot \Delta\equiv (s_1\Delta_1,\ldots,s_p\Delta_p)$. Moreover, by the exchangeability in Step 2, 
\begin{equation*}
\Delta_{\text{swap}(\Acal_3)}\overset{d}{=}\Delta\left(\bSigma(I),\by(I)\right).
\end{equation*}
 Therefore, we have that, 
\begin{equation}
\label{eqn:detlaodot}
\Delta\left(\bSigma(I),\by(I)\right)\overset{d}{=} \bs\odot\Delta\left(\bSigma(I),\by(I)\right).
\end{equation}
This completes the proof of Theorem \ref{prop:coinflip}.
\end{proof}

%%%%%%%%%%%%%%%%%%%%%%%%%%%%%%%%%%%%%%%%%%%%%%%%%%%
\subsection{Proof of Theorem \ref{prop:mFDRcontrol}}

\begin{proof}
The proof follows that of Theorems 1 and 2 in \citet{BC15}.
\end{proof}

%%%%%%%%%%%%%%%%%%%%%%%%%%%%%%%%%%%%%%%%%%%%%%%%%%%
\subsection{Proof of Theorem \ref{thm:optimalrecovery}}
\label{sec:pfofthmrecovery}

We first present the main proof in Section \ref{sec:mainproof}, then give three auxiliary lemmas that are useful for the proof of this theorem in Section \ref{sec:auxiliarypower}.

%%%%%%%%%%%%%%%%%%%%%%%%%%%%%%%%%%%%%%%%%%%%%%%%%%%
\subsubsection{Main proof}
\label{sec:mainproof}

\begin{proof} 
We begin with some notation. For any given set $\cA\subset\{1,\ldots,p\}$ and a matrix $\bG\in\R^{2p\times 2p}$, let $\bG_{\cA,\cA}\in\R^{2|\cA|\times 2|\cA|}$ denote the submatrix formed by the columns and rows in the set $\{j:j\in\cA\text{ or }j-p\in\cA\}$, and $\bG_\cA$ the submatrix by the columns in the set $\cA$. Similarly, for a vector  $\ba\in\R^{2p}$, let $\ba_{\cA}\in\R^{2|\cA|}$ denote the subvector formed by components in the set $\{j:j\in\cA\text{ or }j-p\in\cA\}$. Define
\begin{equation*}
\widetilde{\bG}(I) = \frac{1}{|I|}\bSigma(I)\trans\bSigma(I)\in\R^{2pr\times 2pr}\quad \text{and}\quad \widetilde{\brho}(I) = \frac{1}{|I|}\bSigma(I)\trans\bz_I\in\R^{2p}
\end{equation*}
where $\bz_I \in \R^{|I|}$ has the $i$th element equal to $(y_i - \sum_{j\in I}y_j/|I|)$, and the matrix $\bSigma(I)\in\R^{|I|\times 2pr}$ is as defined in (\ref{eqn:defofSigmaI}). Set $\widehat{\bc}(I) = \left[ \widehat{c}_1(I)\trans,\ldots,\widehat{c}_{2p}(I)\trans \right]\trans$, where $\widehat{c}_j(I)=\mathbf{0}\in\R^r$ for $j\not\in\cS$, and $\widehat{\bc}_{\cS}(I)$ is the minimizer of the partial penalized likelihood,
\begin{equation*}
\widehat{\bc}_{\cS}(I) = \underset{\bc\in\R^{2r}}{\arg\min}\left\{\bc\trans\widetilde{\bG}_{\cS,\cS}(I)\bc-\widetilde{\brho}_\cS(I)\trans\bc+\tau\sum_{j\in\cS}\|c_j(I)\|_2\right\}.
\end{equation*}
By Lemma \ref{lem:keyinfo}  in Section \ref{sec:auxiliarypower},  $\widehat{\cS}(I) = \{j:\widehat{c}_j(I) \neq \mathbf{0}\}$. Recall the definition of $\widehat{\Pi}_j$ in (\ref{eqn:calofpi}) that $\widehat{\Pi}_j= \P\{j\in\widehat{\mathcal S}(I)\}$, $j=1,\ldots,2p$. Let $\Pi_j^* = \mathbf{1}\{j\in\cS\}$. Then  Lemma \ref{lem:keyinfo} implies that,
\begin{equation}
\label{eqn:bdonl1norm}
\sup_{\cS}\sum_{j=1}^{2p}|\widehat{\Pi}_j-\Pi_j^*| = o_p(1),
\end{equation}
Let $\Delta_j^{\cS}$ be the importance score based on the above $\widehat{\bc}(I)$.  Let $|\Delta_{(1)}^\cS|\geq\cdots\geq |\Delta_{(p)}^{\cS}|$ be the ordered importance scores by their magnitudes. Denote by $j^*$ the index such that $|\Delta_{(j^*)}^{\cS}|=T$, where $T$ is the knockoff threshold defined in (\ref{eqn:filterT}). By the definition of $T$, we have $-T< \Delta_{(j^*+1)}^{\cS}\leq 0$. We consider two scenarios, $-T<\Delta_{(j^*+1)}^{\cS}<0$, and $\Delta_{(j^*+1)}^{\cS}=0$, respectively.

\paragraph{Step 1.} For the scenario when $-T<\Delta_{(j^*+1)}^\cS<0$, we have, by the definition of $T$, 
\begin{equation}
\label{eqn:bd+2q}
\frac{|\{j:\Delta_j^{\cS}\leq -T\}+2|}{|\{j:\Delta_j^\cS\geq T\}|}>q,
\end{equation}
which is true, because otherwise, $|\Delta_{(j^*+1)}^\cS|$ would be the new smaller threshold for knockoffs. We next bound $T$. By (\ref{eqn:bd+2q}) and Lemma \ref{lem:keyinfo}, with probability approaching one, we have,
\begin{equation*}
|\{j:\Delta_j^\cS\leq -T\}|>q|\{j:\Delta_j^\cS\geq T\}|-2\geq qc_q|\cS|-2,
\end{equation*}
for some $c_q\geq 1$.
Moreover, when $\Delta_j^\cS\leq -T$, we have $\widehat{\Pi}_j-\widehat{\Pi}_{j+p}\leq -T$, and thus $\widehat{\Pi}_{j+p}\geq T$. By (\ref{eqn:bdonl1norm}), we have that, 
\begin{equation}
\label{eqn:bdingtdeltaj}
\begin{aligned}
o_p(1)  = \sum_{j=1}^{2p}|\widehat{\Pi}_j-\Pi_j^*|  \geq \sum_{j\in\{j:\Delta_j^\cS\leq -T\}}|\widehat{\Pi}_{j+p}| \geq T|\{j:\Delta_j^\cS\leq -T\}|.
\end{aligned}
\end{equation}
Henceforth, 
\begin{equation}
\label{eqn:boundonT}
T\leq \frac{\sum_{j=1}^{2p}|\widehat{\Pi}_j-\Pi_j^*|  }{qc_q|\cS|-2}=o_p(1),\quad \text{as}\quad n\to\infty.
\end{equation}

We next consider the type II error. By (\ref{eqn:bdonl1norm}), we have that, 
\begin{equation*}
\begin{aligned}
o_p(1)&= \sum_{j=1}^{2p}|\widehat{\Pi}_j-\Pi_j^*| =\sum_{j=1}^p\left[|\widehat{\Pi}_j-\Pi_{j}^*| + \widehat{\Pi}_{j+p}\right] \\
& \geq \sum_{j\in\cS\cap(\widehat{\cS})^c}\left[|\widehat{\Pi}_j-\Pi_{j}^*| + \widehat{\Pi}_{j+p}\right] 
\geq \sum_{j\in\cS\cap(\widehat{\cS})^c}\left[|\widehat{\Pi}_j-\Pi_{j}^*| + \widehat{\Pi}_{j}-T\right],
\end{aligned}
\end{equation*}
where the last inequality is due to the fact that $\widehat{\Pi}_{j+p}\geq \widehat{\Pi}_{j}-T$ when $j\in(\widehat{\cS})^c$. By the triangle inequality and that $\Pi_{j}^*=1$ for $j\in\cS$, we have that,
\begin{equation*}
o_p(1)\geq \sum_{j\in \cS\cap (\widehat{\cS})^c} (\Pi_{j}^*-T)= (1-T)\left|\cS \cap (\widehat{\cS})^c\right|.
\end{equation*}
By (\ref{eqn:boundonT}), we obtain that, 
\begin{equation*}
\begin{aligned}
\frac{|\widehat{\cS}\cap\cS|}{|\cS|\vee 1} &=  1-\frac{|(\widehat{\cS})^c\cap\cS|}{|\cS|\vee 1} \geq 1-\frac{1}{(1-T)|\cS|}o_p(1)= 1-o_p(1).
\end{aligned}
\end{equation*}

\paragraph{Step 2.} For the scenario when $\Delta_{(j^*+1)}^\cS=0$, we have that, 
\begin{equation*}
\widehat{\cS} = \{j:\Delta_j^\cS>0\},
\end{equation*}  
and  $\{j:\Delta_j^\cS<-T\} = \{j:\Delta_j^\cS<0\}$. 
Let $\widetilde{r}_n = \sum_{j=1}^{2p}|\widehat{\Pi}_j-\Pi_j^*| $, where both $\widehat{\Pi}_j$ and $\Pi_j^*$ depend on $n$. By Lemma \ref{lem:keyinfo}, $\widetilde{r}_n\to 0$ as $n\to\infty$. Define the sequence $\{\widetilde{q}_n\}$ with $\widetilde{q}_n = \sqrt{\widetilde{r}_n}$, $n\geq 1$. Then $\widetilde{q}_n\to 0$ as $n\to\infty$.

If $|\{j:\Delta_j^\cS<0\}|>\widetilde{q}_n|\cS|$, then similar to (\ref{eqn:bdingtdeltaj}), we have that, 
\begin{equation*}
\begin{aligned}
\sum_{j=1}^{2p}|\widehat{\Pi}_j-\Pi_j^*|  \geq \sum_{j\in\{j:\Delta_j^\cS< -T\}}|\widehat{\Pi}_{j+p}| \geq T|\{j:\Delta_j^\cS< -T\}|.
\end{aligned}
\end{equation*}
Henceforth, 
\begin{equation*}
T\leq \frac{\sum_{j=1}^{2p}|\widehat{\Pi}_j-\Pi_j^*|  }{|\{j:\Delta_j^\cS< -T\}|}= \frac{\sum_{j=1}^{2p}|\widehat{\Pi}_j-\Pi_j^*|  }{|\{j:\Delta_j^\cS< 0\}|}<  \frac{\widetilde{r}_n  }{\widetilde{q}_n|\cS|} = \frac{\sqrt{\widetilde{r}_n}  }{|\cS|} =o_p(1),\quad \text{as}\quad n\to\infty.
\end{equation*}
Therefore it reduces to Step 1, and the arguments therein follow. 

If $|\{j:\Delta_j^\cS<0\}|\leq \widetilde{q}_n|\cS|$, then, by noting $\widehat{\cS} = \text{supp}(\Delta^\cS)\backslash\{j:\Delta_j^\cS<0\}$, we have that, 
\begin{equation}
\label{eqn:intersecthatss0}
\begin{aligned}
\left|\widehat{\cS}\cap \cS\right|  = \left|\text{supp}(\Delta^\cS)\cap\cS\right| - \left|\{j:\Delta_j^\cS<0\}\cap\cS\right|
\geq  \left|\text{supp}(\Delta^\cS)\cap\cS\right|-\widetilde{q}_n|\cS|.
\end{aligned}
\end{equation}
Define $\cS_1=\{1\leq j\leq p:\widehat{c}_j(I)=\mathbf{0}\}$. Note that, 
\begin{equation}
\label{eqn:inclusions1}
\{1,\ldots,p\}\backslash \cS_1\subset \text{supp}(\Delta^\cS),
\end{equation}
By (\ref{eqn:bdonl1norm}), we have that, 
\begin{equation*}
\begin{aligned}
o_p(1) = \sum_{j=1}^{2p}|\widehat{\Pi}_j-\Pi_j^*| \geq \sum_{j\in\cS_1\cap\cS}|\widehat{\Pi}_j-\Pi_j^*| 
= \sum_{j\in\cS_1\cap\cS}\Pi^*_{j}\geq |\cS_1\cap\cS|\min_{j\in\cS}\Pi_{j}^*.
\end{aligned}
\end{equation*}
Note that $\Pi_j^*=1$ for $j\in\cS$. Then the inequality above implies that,
\begin{equation*}
\left|\cS_1\cap\cS\right| = o_p(1).
\end{equation*}
Henceforth, $\left|\left(\{1,\ldots,p\}\backslash\cS_1\right)\cap\cS\right|=|\cS|[1-o_p(1)]$. Combining this result with (\ref{eqn:inclusions1}) yields that, 
\begin{equation}
\label{eqn:suppWs}
\left|\text{supp}(\Delta^\cS)\cap\cS\right|\geq \left|\left(\{1,\ldots,p\}\backslash \cS_1\right)\cap\cS\right|=|\cS|[1-o_p(1)].
\end{equation}
By (\ref{eqn:intersecthatss0}) and (\ref{eqn:suppWs}), it holds that
\begin{equation*}
\frac{|\widehat{\cS}\cap\cS|}{|\cS|\vee 1}\geq 1-o_p(1).
\end{equation*}

\paragraph{Step 3.} Combining the two scenarios, we obtain that, with probability approaching one, 
\begin{equation*}
\frac{|\widehat{\cS}\cap\cS|}{|\cS|\vee 1}\geq 1-o_p(1),
\end{equation*}
which further implies that,
\begin{equation*}
\begin{aligned}
\text{Power}(\widehat{\cS}) = \E\left[\frac{|\widehat{\cS}\cap\cS|}{|\cS|\vee 1}\right]
\geq 1-o(1).
\end{aligned}
\end{equation*}
This completes the proof of Theorem \ref{thm:optimalrecovery}.
\end{proof}

%%%%%%%%%%%%%%%%%%%%%%%%%%%%%%%%%%%%%%%%%%%%%%%%%%%
\subsubsection{Auxiliary lemmas for Theorem \ref{thm:optimalrecovery}}
\label{sec:auxiliarypower}

%%%%%%%%%%%%%%%%%%%%%%%%%%%%%%%%%%%%%%%%%%%%%%%%%%%
\noindent 
For any $g\in\Hcal$, define the norm, $\|g(x(t))\|_n =\sqrt{(1/n) \sum_{i=1}^ng^2(x(t_i))}$. The next lemma gives a useful concentration bound for functional empirical processes.
\begin{lemma}
\label{lem:bdonradamacher1}
Suppose that $g=\sum_{j=1}^pg_j\in\Hcal_p$, and the errors $\{\epsilonbf_{i}\}_{i=1}^n$ are i.i.d.\ normal. Then there exists some constant $C>0$, such that, for any $c_1>0$ and $c_2>1$, with probability at least $1-2p^{-c_1}$,
\begin{equation*}
\begin{aligned}
 \frac{1}{n}\sum_{i=1}^n\epsilon_{i}g(x_i)
\leq \; & C \Bigg\{ c_2^{-\frac{4\beta}{2\beta-1}}\|g\|_{L_2}^2+ \left( c_2^{-\frac{4\beta}{2\beta-1}}+c_2^{\frac{4\beta}{4\beta+1}} \right)  n^{-\frac{2\beta}{2\beta+1}} \\
& \quad\quad + \left( c_2^{-\frac{4\beta}{2\beta-1}} + c_1 + 1 \right) \frac{\log p}{n} + \sqrt{(c_1+1) \frac{\log p}{n} } \|g\|_{L_2} + n^{-1/2}e^{-p} \Bigg\}.
\end{aligned}
\end{equation*}
\end{lemma}

\noindent
\begin{proof}
We first note that, the $\nu$th eigenvalue of the reproducing kernel of RKHS $\Hcal_1$ is of order $\nu^{-2\beta}$, for $\nu\geq 1$ \citep[see, e.g.,][]{bach2017equivalence}. Since $\{\epsilon_{i}\}_{i=1}^{n}$ are i.i.d.\ normal, by Lemma 2.2 of \citet{Yuan2016} and
Corollary 8.3 of \citet{Vandegeer2000}, we have that, for any $c_1>0$, with probability at least $1-p^{-c_1}$,
\begin{equation}
\label{eqn:empmean}
\begin{aligned}
 \frac{1}{n}\sum_{i=1}^n\epsilon_{i}g(x_i) 
\leq \; & 2C_1n^{-1/2}\sum_{j=1}^p \|g_{j}(x)\|_{n}^{1-\frac{1}{2\beta}}  \|g_{j}(x)\|_{K}^{\frac{1}{2\beta}} \\
& + 2C_1n^{-1/2}\sqrt{(c_1+1)\log p}\sum_{j=1}^p  \|g_{j}(x)\|_{n}  + 2C_1n^{-1/2}e^{-p}\sum_{j=1}^p\|g_{j}(x)\|_{K} \\ 
\equiv \; & 2C_1 (\Delta_1 + \Delta_2+ \Delta_3), 
\end{aligned}
\end{equation}
for some constant $C_1$. Next, we bound the three terms $\Delta_1, \Delta_2, \Delta_3$ in \eqref{eqn:empmean}, respectively. 

For $\Delta_1$, by the Young's inequality, for any $c_2>1$, we have, 
\begin{equation*}
\begin{aligned}
\Delta_1 \leq \; & c_2^{-\frac{4\beta}{2\beta-1}}\sum_{j=1}^p \|g_{j}(x)\|_{n}^{2} + c_2^{\frac{4\beta}{4\beta+1}}n^{-\frac{2\beta}{2\beta+1}}\sum_{j=1}^p  \|g_{j}\|_{K}^{\frac{2}{2\beta+1}}.
\end{aligned}
\end{equation*}
Note that 
\begin{equation*}
\sum_{j=1}^p  \|g_{j}\|_{K}^{\frac{2}{2\beta+1}} \leq C'_2 \sum_{j=1}^p \left( \|g_{j}\|_{K}\right)^{0}  \leq C_2, 
\end{equation*}
for some constants $C_2',C_2$, where the last inequality is due to Assumption \ref{assump:complexity} that the number of nonzero functional components of $\{g_1,\ldots,g_p\}$ is bounded. Henceforth,
\begin{equation}
\label{eqn:firststepbd}
\begin{aligned}
 n^{-1/2}\sum_{j=1}^p\|g_{j}(x)\|_{n}^{1-\frac{1}{2\beta}} \|g_{j}\|_{K}^{\frac{1}{2\beta}}
\leq \;  c_2^{-\frac{4\beta}{2\beta-1}}\sum_{j=1}^p\|g_{j}(x(t))\|_{n}^{2} + c_2^{\frac{4\beta}{4\beta+1}}n^{-\frac{2\beta}{2\beta2+1}}C_2.
\end{aligned}
\end{equation}
By Theorem 4 of \citet{koltchinskii2010sparsity}, there exists some constant $C_3>0$, such that, with probability at least $1-p^{-c_1}$,
\begin{equation*}
\begin{aligned}
\sum_{j=1}^p \|g_j(x)\|_n^2
\leq 2C_3^2\sum_{j=1}^p  \|g_{j}(x)\|_{L_2}^2 
+ 2C_3^2\left\{ n^{-\frac{2\beta}{2\beta+1}}+\frac{(c_1+1)\log p}{n} \right\} \sum_{j=1}^p  \|g_{j}\|_{K}^2.
\end{aligned}
\end{equation*}
Note that there exists some constant $c_3>1$, such that 
$\sum_{j=1}^p  \|g_{j}\|_{L_2}^2 \leq c_3  \|g\|_{L_2}^2$,
 and 
$\sum_{j=1}^p \|g_{j}\|_{K}^2 \leq \sum_{j=1}^p \|g_{j}\|_{K}^0\leq C_2$. Then, we have
\begin{equation*}
\begin{aligned}
\sum_{j=1}^p  \|g_{j}(x)\|_n^2 \leq 2C_3^2c_3 \|g(x)\|_{L_2}^2 + 2C_2C_3^2\left\{n^{-\frac{2\beta}{2\beta+1}} + \frac{(c_1+1)\log p}{n} \right\}.
\end{aligned}
\end{equation*}
Inserting into (\ref{eqn:firststepbd}) yields that
\begin{equation}
\label{eqn:firststepbdfinal}
\begin{aligned}
\Delta_1 \leq  2C_3^2c_3c_2^{-\frac{4\beta}{2\beta-1}}\|g\|_{L_2}^2  + 2C_2C_3^2c_2^{-\frac{4\beta}{2\beta-1}} \left\{ n^{-\frac{2\beta}{2\beta+1}}+\frac{(c_1+1)\log p}{n} \right\} +  C_2c_2^{\frac{4\beta}{4\beta+1}}n^{-\frac{2\beta}{2\beta2+1}}.
\end{aligned}
\end{equation}

For $\Delta_2$, by Theorem 4 of \citet{koltchinskii2010sparsity} again, there exists a constant $C_4>0$, such that
\begin{equation*}
\begin{aligned}
 \sum_{j=1}^p\|g_{j}(x)\|_{n} 
\leq \; & C_4\sum_{j=1}^p\|g_{j}(x)\|_{L_2}+C_4\left\{ n^{-\frac{\beta}{2\beta+1}}+\sqrt{\frac{(c_1+1)\log p}{n}} \right\}\sum_{j=1}^p\|g_{j}\|_{K}\\
\leq \; & C_4\sum_{j=1}^p\|g_{j}(x)\|_{L_2}+C_2C_4\left\{ n^{-\frac{\beta}{2\beta+1}}+\sqrt{\frac{(c_1+1)\log p}{n}}\right\}.
\end{aligned}
\end{equation*}
Define the set $\mathcal Q_1 \equiv \left\{ j=1,\ldots,p: \|g_{j}(x)\|_{L_2}>\sqrt{n^{-1}\log p} \right\}$. By the Cauchy-Schwartz inequality, we have, 
\begin{equation*}
\begin{aligned}
\sum_{j\in\mathcal Q_1}\|g_{j}(x)\|_{L_2} \leq | \mathcal Q_1 |^{1/2} \ \left(\sum_{j\in\mathcal Q_1}\|g_{j}(x(t))\|^2_{L_2}\right)^{1/2} 
\leq \sum_{j=1}^p\|g_{j}\|_\Hcal^0\cdot\left(\sum_{j=1}^p\|g_{j}\|^2_{L_2}\right)^{1/2}\leq C_2c_4\|g(x)\|_{L_2}, 
\end{aligned}
\end{equation*}
where $c_4>1$ satisfies that $\sum_{j=1}^p\|g_{j}(x)\|^2_{L_2}\leq c^2_4\|g(x)\|^2_{L_2}$. Next, define the set $\mathcal Q_2 \equiv \{j=1,\ldots,p: \|g_{j}(x)\|_{L_2}\leq\sqrt{n^{-1}\log p}\}$. By definition,
\begin{equation*}
\begin{aligned}
\sum_{j\in\mathcal Q_2}\|g_{j}(x(t))\|_{L_2}   \leq \sum_{j\in\mathcal Q_2}\|g_{j}(x)\|^0_{L_2}\sqrt{\frac{\log p}{n}} \leq C_2 \sqrt{\frac{\log p}{n}}.
\end{aligned}
\end{equation*}
Combining $\mathcal Q_1$ and $\mathcal Q_2$ gives,
\begin{equation*}
\begin{aligned}
\sum_{j=1}^p\|g_{j}(x)\|_{L_2} \leq \sum_{j\in\mathcal Q_1}\|g_{j}(x)\|_{L_2} + \sum_{j\in\mathcal Q_2}\|g_{j}(x)\|_{L_2}\leq C_2c_4\|g(x)\|_{L_2}+ C_2  \sqrt{\frac{\log p}{n}}.
\end{aligned}
\end{equation*}
Henceforth, we can bound $\Delta_2$ as,
\begin{equation*}
\begin{aligned}
\Delta_2 \leq C_2C_4c_4\sqrt{\frac{\log p}{n}}\|g(x)\|_{L_2} + C_2C_4n^{-\frac{\beta}{2\beta+1}}\sqrt{\frac{\log p}{n}} + 2C_2C_4\sqrt{(c_1+1)}\frac{\log p}{n}.
\end{aligned}
\end{equation*}

For $\Delta_3$, we have that, 
\begin{equation*}
\Delta_3 \leq n^{-1/2}e^{-p}\sum_{j=1}^p\|g_{j}\|^0_{K} \leq C_2n^{-1/2}e^{-p}.
\end{equation*}

Combining the bounds for $\Delta_1, \Delta_2$ and $\Delta_3$, and applying the Cauchy-Schwarz inequality completes the proof of Lemma \ref{lem:bdonradamacher1}. 
\end{proof}
\bigskip

%%%%%%%%%%%%%%%%%%%%%%%%%%%%%%%%%%%%%%%%%%%%%%%%%%%
The next lemma shows that the selection probabilities of variables by the kernel knockoffs procedure asymptotically recover the true selection probabilities.

\begin{lemma}
\label{lem:keyinfo}
Recall the definition of $\widehat{\Pi}_j$ in (\ref{eqn:calofpi}) that $\widehat{\Pi}_j= \P\{j\in\widehat{\mathcal S}(I)\}$ for $j=1,\ldots,2p$. Let $\Pi_j^* = \mathbf{1}\{j\in\cS\}$. Suppose Assumptions  \ref{eqn:perfect-relationship} to \ref{network:incoherence} hold. Then, for any true set $\cS$,
\begin{equation*}
\begin{aligned}
\sum_{j=1}^{2p}|\widehat{\Pi}_j-\Pi_j^*|=o_p(1), \quad \text{as}\quad n\to\infty.
\end{aligned}
\end{equation*}
\end{lemma}
\begin{proof}
We use the primal-dual witness method  to prove this lemma. Our result extends the previous techniques in \citet{ravikumar2010high} for the random feature basis and knockoff variables.

Recall that, by the representer theorem \citep{wahba1990spline} and the random feature mapping, the selection problem becomes \eqref{eq: kernel_reg}, i.e., 
\begin{equation}
\label{eqn:zjlasso}
\begin{aligned}
\widehat{\bc}_{2p}(I) = \underset{\underset{j=1,\ldots,2p}{c_{j} \in \mathbb R^{r}}}{\arg\min}~\left[\bz_I-\bSigma(I)\bc\right]\trans\left[\bz_I-\bSigma(I)\bc\right]+ |I|\tau \sum_{j=1}^{2p} \|c_{j}\|_2,
\end{aligned}
\end{equation}
where the ``response" is $\bz_I\in\R^{|I|}$ with the $i$th element equal to $(y_i - \sum_{j\in I}y_j/|I|)$, and the ``predictor" matrix is $\bSigma(I)\in\R^{|I|\times 2pr}$ as defined in (\ref{eqn:defofSigmaI}). Rewrite the matrix $\bSigma(I)$ as
\begin{equation*}
\bSigma(I) = 
\begin{bmatrix}
\Sigma_1(I) & \cdots &\Sigma_p(I) & \Sigma_{p+1}(I) & \cdots & \Sigma_{2p}(I) 
\end{bmatrix},
\end{equation*}
where $\Sigma_j(I)\in\R^{|I|\times r}$, for $j=1,\ldots,2p$. The vector $\bc(I)$ solves (\ref{eqn:zjlasso}) if it satisfies the Karush-Kuhn-Tucker (KKT) condition, 
\begin{equation}
\label{eqn:kkttheta}
\bSigma_j(I)\trans\left[\bSigma(I)\bc(I)-\bz_I\right] + \frac{1}{2}|I|\tau g_j(I) = 0, \quad j=1,\ldots,2p,
\end{equation}
where  for each $j=1,\ldots,2p$, $g_j(I)\in\R^r$ is defined by, 
\begin{equation}
\label{eqn:kkttg}
g_j(I)  = c_j(I)/\|c_j(I)\|_2 \;\; \text{if }c_j\neq 0, \quad\textrm{ and }\quad \|g_j(I)\|_2 \leq 1 \;\; \text{otherwise}.
\end{equation}

In order to apply the primal-dual witness method, we next construct an oracle primal-dual pair $(\widehat{\bc}(I),\widehat{\bg}(I))$ satisfying the KKT conditions (\ref{eqn:kkttheta}) and (\ref{eqn:kkttg}). Specifically,
\begin{itemize}
\item[(a)] We set $\widehat{\bc}(I) = \left[ \widehat{c}_1(I)\trans,\ldots,\widehat{c}_{2p}(I)\trans \right]\trans$ where  $\widehat{c}_{j}(I)=\mathbf{0}\in\R^r$ for $j\not\in \cS$.
\item[(b)] Let $\widehat{\bc}_{\cS}(I)$ be the minimizer of the partial penalized likelihood, 
\begin{equation}
\label{eqn:consttheta}
\left[\bz_I - \bSigma_{\cS}(I)\bc_{\cS}(I)\right]\trans \left[\bz_I - \bSigma_{\cS}(I)\bc_{\cS}(I)\right] + |I|\tau\sum_{j\in\cS}\|c_j(I)\|_2.
\end{equation}
\item[(c)] We obtain $\widehat{\bg}_{\cS^\perp}(I)$ from (\ref{eqn:kkttheta}) by substituting in the values of $\widehat{\bc}(I)$ and $\widehat{\bg}_{\cS}(I)$, where $\cS^\perp$ is the complement of $\cS$ in $\{1,\ldots,p\}$.
\end{itemize}

Next, we verify the support recovery consistency. Denote $\bc^*(I) = \left( {c^*_1}\trans,\ldots,{c_{2p}^*}\trans \right)\trans$, where $c_j^*(I)=\mathbf{0}\in\R^r$ for $j\not\in\cS$, and $\bc^*_\cS(I)$ is defined by
\begin{equation*}
\bc_\cS^*(I) \equiv \argmin_{\bc_\cS\in\R^{|\cS|r}} \left\|\E[\bz_I] - \bSigma_{\cS}(I)\bc_{\cS}\right\|^2_2.
\end{equation*}

Note that the subgradient condition for the partial penalized likelihood (\ref{eqn:consttheta}) is 
\begin{equation*}
\begin{aligned}
2\{\bSigma_{\cS}(I)\}\trans[\bSigma_{\cS}(I)\widehat{\bc}_{\cS}(I)-\bz_I]+\tau|I|\widehat{\bg}_{\cS}(I) = \mathbf{0},
\end{aligned}
\end{equation*}  
which implies that
\begin{equation*}
\begin{aligned}
2\{\bSigma_{\cS}(I)\}\trans[\bSigma_{\cS}(I)\widehat{\bc}_{\cS}(I)-\bSigma_{\cS}(I)\bc^*_{\cS}(I)]+2\{\bSigma_{\cS}(I)\}\trans[\bSigma_{\cS}(I)\bc^*_{\cS}(I)-\bz_I]+\tau|I|\widehat{\bg}_{\cS}(I) = \mathbf{0}.
\end{aligned}
\end{equation*}
Define $\bR_{\cS}(I) \equiv 2|I|^{-1}\{\bSigma_{\cS}(I)\}\trans[\bSigma_{\cS}(I)\bc^*_{\cS}(I)-\bz_I]$.  Then,
\begin{equation}
\label{eqn:diftheta}
\widehat{\bc}_{\cS}(I)-\bc^*_{\cS}(I) = -|I|\left[2\{\bSigma_{\cS}(I)\}\trans\bSigma_{\cS}(I)\right]^{-1}(\bR_{\cS}(I)+\tau\widehat{\bg}_{\cS}(I) ).
\end{equation}
For each $j\in\cS$, denote the corresponding submatrix of $\bSigma_{\cS}(I)$ by $\bSigma_{j}(I)$. Then for $j \in \cS$,
\begin{equation}
\label{eqn:defofrk}
R_{j}(I) \equiv 2|I|^{-1}\{\bSigma_{j}(I)\}\trans [\bSigma_{\cS}(I)\bc^*_{\cS}(I) - \bz_I].
\end{equation}
By Lemma \ref{lem:conds3}, we have $\|R_{j}(I)\|_{2}\leq\eta_{R}$. Then,
\begin{equation}
\label{eqn:rs0}
\|\bR_{\cS}(I)\|_{2}\leq \eta_{R}\sqrt{|\cS|}.
\end{equation}
By Assumption \ref{network:dependency} and Claim 1 of \citet{RR07}, we have $\Lambda_{\min}\left(|I|^{-1}\bSigma_{\cS}(I)\trans \bSigma_{\cS}(I)\right)\ge C_{\min}/2$ for a large $n$. Henceforth,
\begin{equation*}
 \Lambda_{\max}\left\{|I|\left(2\bSigma_{\cS}(I)\trans \bSigma_{\cS}(I)\right)^{-1}\right\}\leq \frac{1}{C_{\min}}.
\end{equation*}
Note that for any $j\in \cS$, $\|\widehat{g}_{j}(I)\|_{2}\leq 1$, which implies that, 
\begin{equation}
\label{eqn:gs0hat}
\|\widehat{\bg}_{\cS}(I)\|_{2} \leq \sqrt{|\cS|}.
\end{equation}
Let $f^*_{\min} = \min_{j\in \cS}\|f_j(X_j)\|_{L_2(X_j)}$. We have that, 
\begin{equation*}
\begin{aligned}
\max_{j\in \cS}\|\Sigma_j(I)\widehat{c}_{j}(I)-\Sigma_j(I)c^*_{j}(I)\|_{L_2(X_j)} & \leq \|\Sigma_{\cS}(I)\widehat{\bc}_{\cS}(I)-\Sigma_{\cS}(I)\bc^*_{\cS}(I)\|_{L_2(X_j)} \\
& \leq c\left(\frac{\eta_R\sqrt{|\cS|}}{C_{\min}}+\frac{\tau\sqrt{|\cS|}}{C_{\min}}\right)\leq \frac{2}{3}f^*_{\min}.
\end{aligned}
\end{equation*}
where the second inequality is due to the fact that the random feature $\| \bPsi(x)\|_2$ is bounded, and the last inequality is due to Assumption \ref{network:minregeffect}, and $\tau=C_\tau\{[(\log p)/n]^{1/2}+n^{-\beta/(2\beta+1)}\} = O(\eta_R)$. 

Combining this result with Proposition \ref{prop:nullequivalence} implies that the oracle estimator $\widehat{\bc}(I)$ recovers the support $\cS$ exactly.  

Next, we verify the strict dual feasibility, i.e., 
\begin{equation*}
\max_{j\not\in \cS}|\widehat{g}_{j}(I)|<1,
\end{equation*}
which in turn implies that the oracle estimator $\widehat{c}_j$ satisfies the KKT conditions (\ref{eqn:kkttheta}) and (\ref{eqn:kkttg}).

For any $j\not\in \cS$, by (\ref{eqn:kkttheta}), we have, 
\begin{equation*}
\{\bSigma_{j}(I)\}\trans[\bSigma_{\cS}(I)\widehat{\bc}_{\cS}(I)-\bz_I]+\frac{1}{2}|I|\tau\widehat{g}_{j}(I) = 0,
\end{equation*}
which further implies that, 
\begin{equation*}
\begin{aligned}
2\{\bSigma_{j}(I)\}\trans[\bSigma_{\cS}(I)\widehat{\bc}_{\cS}(I)-\bSigma_{\cS}(I)\bc_{\cS}^*(I)]+2\{\bSigma_{j}(I)\}\trans[\bSigma_{\cS}(I)\bc^*_{\cS}(I)-\bz_I]+|I|\tau\widehat{g}_{j}(I) = 0.
\end{aligned}
\end{equation*}
By (\ref{eqn:diftheta}) and (\ref{eqn:defofrk}), we have, 
\begin{equation*}
\tau\widehat{g}_{j}(I) = \{\bSigma_{j}(I)\}\trans \bSigma_{\cS}(I)[\bSigma_{\cS}(I)\trans \bSigma_{\cS}(I)]^{-1}[\bR_{\cS}(I)+\tau\widehat{\bg}_{\cS}(I)]-R_{j}(I).
\end{equation*}
By Assumption \ref{network:incoherence}, we have that, 
\begin{equation*}
\max_{j\not\in \cS}\left\|\{\bSigma_{j}(I)\}\trans \bSigma_{\cS}(I)[\bSigma_{\cS}(I)\trans \bSigma_{\cS}(I)]^{-1}\right\|_{2}\leq \xi_{\bSigma}.
\end{equation*}
Then by (\ref{eqn:rs0}) and (\ref{eqn:gs0hat}), we have that
\begin{equation*}
|\widehat{g}_{j}(I)|\leq \frac{\xi_{\bSigma}\sqrt{|\cS|}+1}{\tau}\eta_{R}+\xi_{\bSigma}\sqrt{|\cS|},\quad j\not\in \cS.
\end{equation*}
By Assumption \ref{network:minregeffect} that
\begin{equation*}
 \frac{\xi_{\bSigma}\sqrt{|\cS|}+1}{\tau}\eta_{R}+\xi_{\bSigma}\sqrt{|\cS|}<1,
\end{equation*}
we obtain that, 
\begin{equation*}
|\widehat{g}_{j}(I)|<1, \quad\text{for any } j\not\in\cS.
\end{equation*} 

Finally, combining the above results, we have that, if $j\in\cS$, then $\widehat{\Pi}_j = 1$, and if  $j\not\in\cS$, then $\widehat{\Pi}_j = 0$. This completes the proof of Lemma \ref{lem:keyinfo}.
\end{proof}
\bigskip

%%%%%%%%%%%%%%%%%%%%%%%%%%%%%%%%%%%%%%%%%%%%%%%%%%%
The next lemma gives a bound similar to the deviation condition studied in \citet{loh2012, dai2021kernel}. The difference is that, the matrix $\bSigma(I)$ in (\ref{eqn:defofSigmaI}) involves the random feature mapping.

\begin{lemma}
\label{lem:conds3}
For $j=1,\ldots,p$, and subsampling $I\subset\{1,\ldots,n\}$ with $|I| =\lfloor n/2 \rfloor$, for $R_j(I)$ as defined in (\ref{eqn:defofrk}), we have that, 
\begin{equation*}
\|R_{j}(I)\|_{2}\leq \eta_{R}, \;\; \textrm{ where } \; 
\eta_{R} = O_p\left(n^{-\frac{\beta}{2\beta+1}} + \left(\frac{\log p}{n}\right)^{1/2} + r^{-1/2}\right).
\end{equation*}
\end{lemma}
\noindent
\begin{proof}
We first note that, 
\begin{equation}
\label{eqn:Gdecomp}
\begin{aligned}
&|I|^{-1}\left\||\{\bSigma_{j}(I)\}\trans \bz_I- \{\bSigma_{j}(I)\}\trans \bSigma_{\cS}(I)\bc^*_{\cS}(I) \right\|_{2} \\
\leq \; & |I|^{-1}\left\|\{\bSigma_{j}(I)\}\trans \E[\bz_I] -   \{\bSigma_{j}(I)\}\trans \bSigma_{\cS}(I)\bc^*_{\cS}(I)\right\|_{2} + |I|^{-1}\left\|\{\bSigma_{j}(I)\}\trans (\bz_I-\E[\bz_I])\right\|_{2} \\\equiv \; & \Delta_1 + \Delta_2.
\end{aligned}
\end{equation}
We next bound the two terms $\Delta_1, \Delta_2$ in (\refeq{eqn:Gdecomp}), respectively.

For $\Delta_1$, by the Cauchy-Schwarz inequality, we have that, for some constant $C_1$, 
\begin{equation*}
\begin{aligned}
\Delta^2_1 \leq \frac{\left\| \bSigma_{j}(I) \right\|_{2}^2}{|I|} \cdot\frac{\left\| \E[\bz_I]- \bSigma_{\cS}(I)\bc^*_{\cS}(I)\right\|_{2}^2}{|I|}  \leq C_1\frac{\left\| \E[\bz_I]-\bSigma_{\cS}(I)\bc^*_{\cS}(I)\right\|_{2}^2}{|I|} 
= O_p\left(r^{-1}\right),
\end{aligned}
\end{equation*}
due to the boundedness of the random features (\ref{eqn:fouriermc}), and Theorem 5 of \citet{rudi2017generalization}.

For $\Delta_2$, by Lemma \ref{lem:bdonradamacher1}, we have that, 
\begin{equation*}
\begin{aligned}
\Delta^2_2 = O_p\left( n^{-\frac{2\beta}{2\beta+1}} + \frac{\log p}{n} \right).
\end{aligned}
\end{equation*}

Combining the bounds for $\Delta_1$ and $\Delta_2$, we obtain that, 
\begin{equation*}
\left\||\{\bSigma_{j}(I)\}\trans \bz_I- \{\bSigma_{j}(I)\}\trans \bSigma_{\cS}(I)\bc_{\cS} \right\|_{2} = O_p\left( n^{-\frac{\beta}{2\beta+1}} + \left(\frac{\log p}{n}\right)^{1/2} + r^{-1/2} \right),
\end{equation*}
which completes the proof of Lemma \ref{lem:conds3}.
\end{proof}

%%%%%%%%%%%%%%%%%%%%%%%%%%%%%%%%%%%%%%%%%%%%%%%%%%%
\section{Additional numerical results}
\label{sec:add-numerical}

We report some additional numerical results, including a sensitivity analysis to the choice of the kernel function and the tuning parameters, the performance under the varying correlation level and the sparsity level, and a parallelization experiment.

%%%%%%%%%%%%%%%%%%%%%%%%%%%%%%%%%%%%%%%%%%%%%%%%%%%
\subsection{Sensitivity analysis to the choice of kernel}
\label{sec:sensivity-kernel}

We first carry out a sensitivity analysis to investigate the robustness of the kernel knockoffs to the choice of kernel function. We consider three commonly used kernel functions: the Laplacian kernel, the Gaussian kernel, and the Cauchy kernel \citep{scholkopf2002learning, RR07, buazuavan2012fourier}, 
\begin{eqnarray*}
\textrm{Laplacian kernel}: & & K(X, X') = c_1 e^{-|X-X'|/b_1}, \\
\textrm{Gaussian kernel}: & & K(X, X') = c_2 e^{-b_2^2|X-X'|^2/2}, \\
\textrm{Cauchy kernel}:    & & K(X, X') = c_3 (1+ b_3^2|X-X'|^2)^{-1},
\end{eqnarray*}
where $c_1,c_2, c_3$ are the normalization constants, and $b_1,b_2, b_3$ are the scaling parameters. 

We consider the simulation example in Section \ref{sec:varysamplesim}, with two forms of the component function, the trigonometric polynomial function \eqref{eq: sinpoly}, and the sin-ratio function \eqref{eq: sinratio}. We set the sample size $n=900$, the dimension $p = 50$, the sparsity level $|\mathcal S| = 10$, the signal strength $\theta = 100$, and simulate the predictors from a multivariate normal distribution with mean zero and covariance $\Sigma_{ij}=0.3^{|i-j|}$. We tune $r$ and $\tau$ following the criteria in Section \ref{sec:tuning}, and set $L = 300$ and $n_s = \lfloor n/2 \rfloor$. We set the target FDR level at $q=0.2$. Figure \ref{fig: kernel} reports the FDR and power based on 200 data replications. It is seen that the performance of the kernel knockoffs method is relatively robust to different kernel choices. We choose the Laplacian kernel for all other simulations. 

\begin{figure}[t!]
\centering
\includegraphics[width=3.0in,height=\myfigheight]{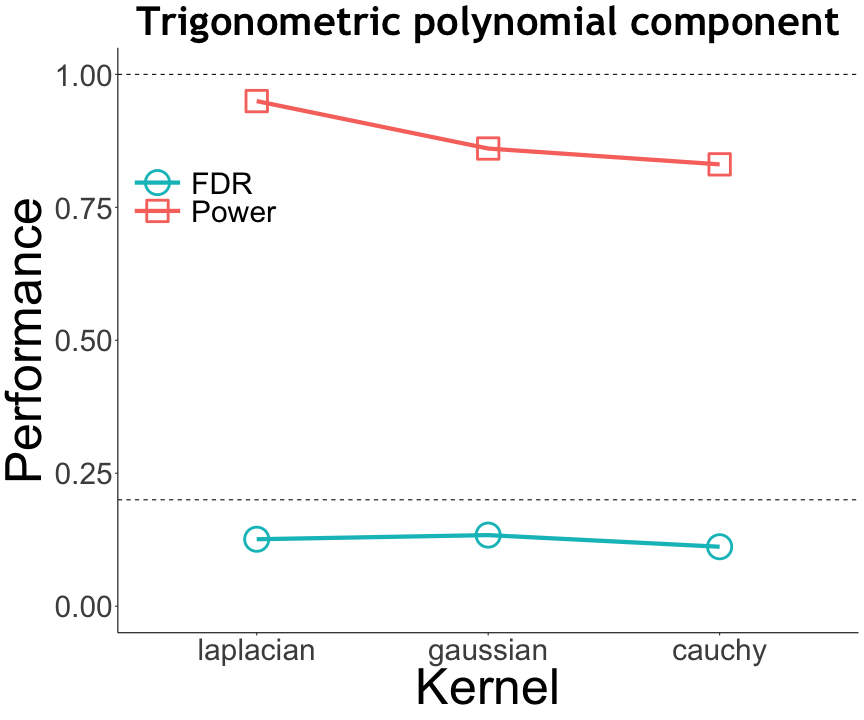}
\includegraphics[width=3.0in,height=\myfigheight]{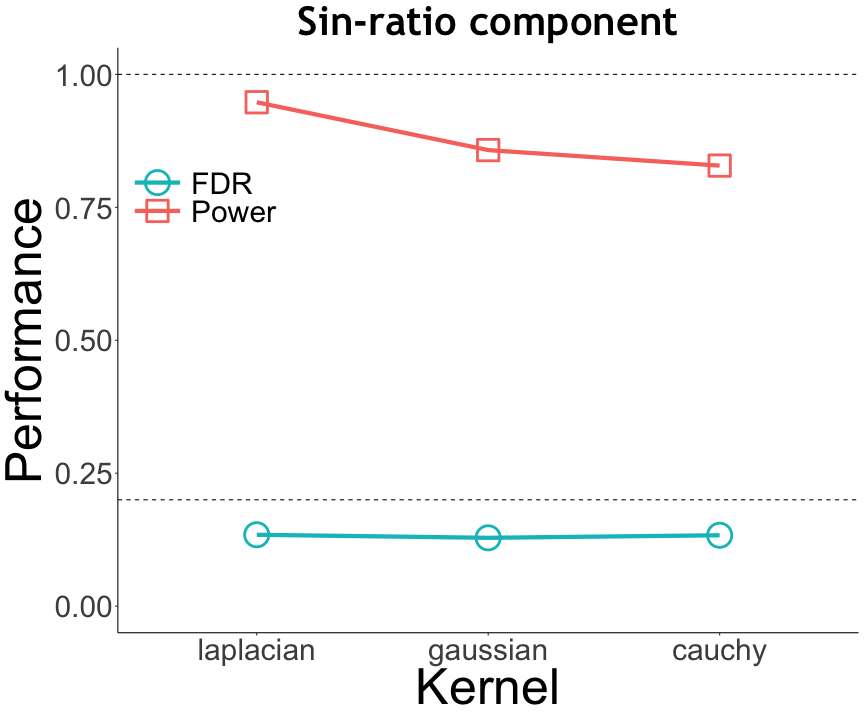}
\caption{Sensitivity analysis to the choice of kernel.}
\label{fig: kernel}
\end{figure}

%%%%%%%%%%%%%%%%%%%%%%%%%%%%%%%%%%%%%%%%%%%%%%%%%%%
\subsection{Sensitivity analysis to the tuning parameters}
\label{sec:sensivity-tuning}

We next carry out a sensitivity analysis of the kernel knockoffs method with respect to the tuning parameters. Specifically, we tune the number of random features $r$ and the regularization parameter $\tau$ following the criteria in Section \ref{sec:tuning}, and we study here the sensitivity to the choice of the number of subsampling replications $L$, and the subsampling sample size $n_s$. 
  
\begin{figure}[t!]
\centering
\begin{tabular}{cc}
{\small $L$} & 
{\small $n_s$} \\
\includegraphics[width=3.0in,height=\myfigheight]{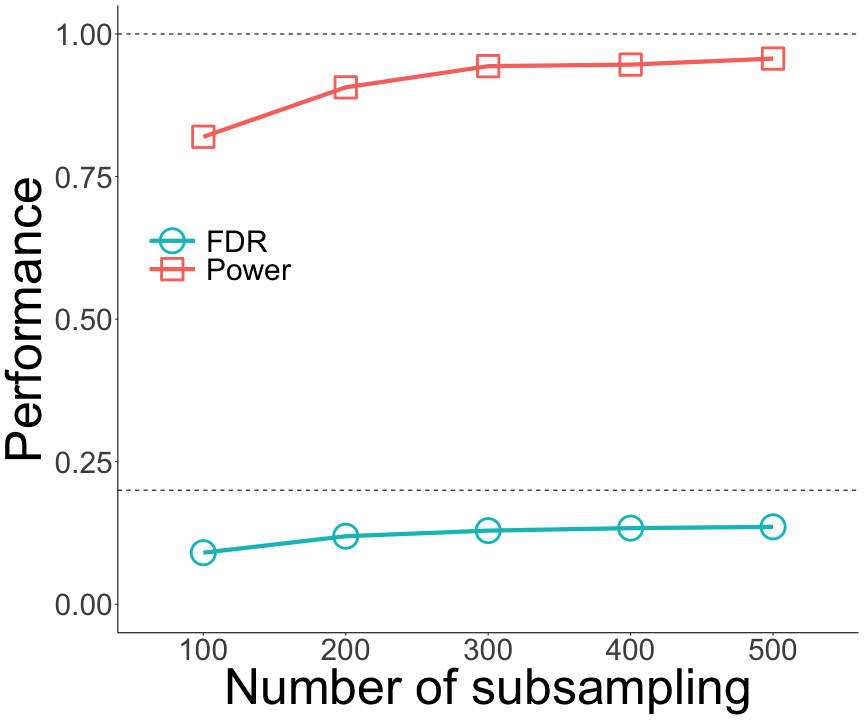} &
\includegraphics[width=3.0in,height=\myfigheight]{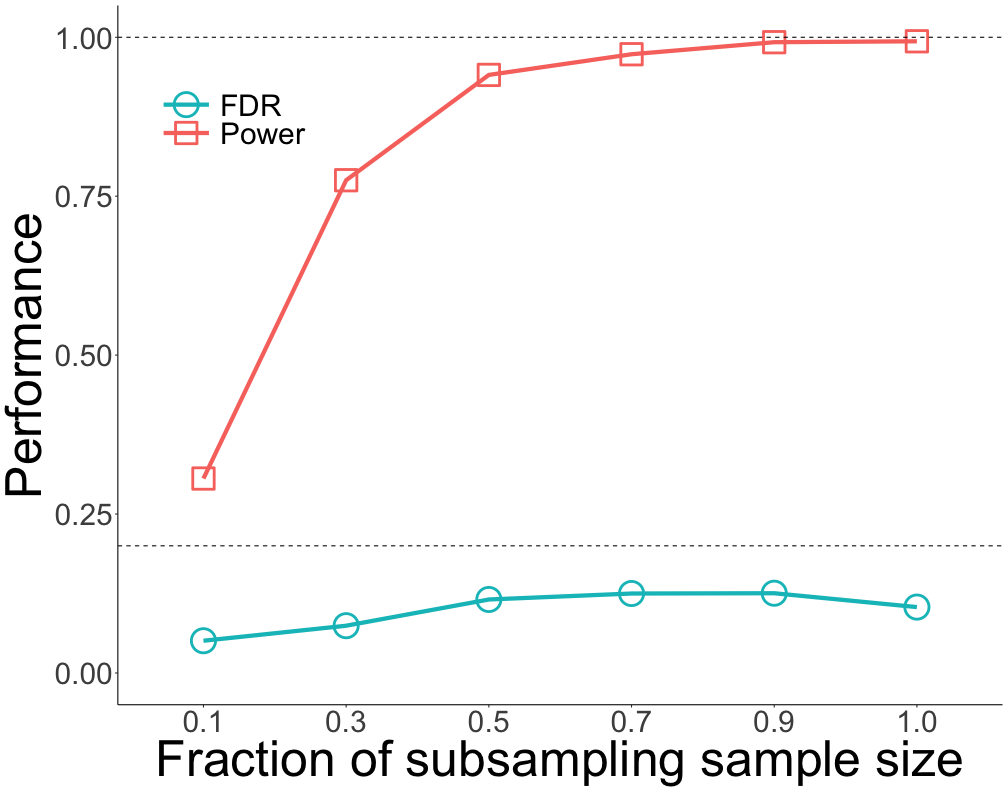} 
\end{tabular}
\caption{Sensitivity analysis to the choice of the tuning parameters.}
\label{fig: robust_tune}
\end{figure} 

For $L$, we adopt the same simulation setup as in Section \ref{sec:sensivity-kernel} with the trigonometric polynomial function \eqref{eq: sinpoly}. We vary $L = \{100,200,300,400,500\}$, and set $n_s = \lfloor n/2 \rfloor$. Figure \ref{fig: robust_tune} (left panel) reports the FDR and power with the varying $L$ based on 200 data replications. It is seen that the method performs slightly better with a larger $L$, but the performance is relatively robust. Since a larger $L$ implies a heavier computation, we set $L = 300$ for all other simulations. 

For $n_s$, we adopt the same simulation setup as in Section \ref{sec:sensivity-kernel} with the trigonometric polynomial function \eqref{eq: sinpoly}. We vary the fraction $n_s / n = \{0.1, 0.3, 0.5, 0.7, 0.9, 1\}$ and set $L = 300$. Figure \ref{fig: robust_tune} (right panel) reports the FDR and power with the varying fraction based on 200 data replications. It is seen that the FDR is consistently below the nominal level, and the power becomes stable once $n_s$ is no smaller than $\lfloor n/2 \rfloor$. We set $n_s = \lfloor n/2 \rfloor$ for all other simulations.

%%%%%%%%%%%%%%%%%%%%%%%%%%%%%%%%%%%%%%%%%%%%%%%%%%%
\subsection{Predictor correlation}
\label{sec:correlation}

In the simulations so far, we have considered to simulate the predictors from a multivariate normal distribution with mean zero and covariance $\Sigma_{ij}=\rho^{|i-j|}$, where we fix $\rho=0.3$. Next, we consider more correlation values for $\rho$. Specifically, we adopt the same simulation setup as in Section \ref{sec:sensivity-kernel} with the trigonometric polynomial function \eqref{eq: sinpoly}. We vary $\rho=\{0.3,0.4,0.5,0.6,0.7,0.8\}$. Figure \ref{fig: dependence} reports the FDR and power based on 200 data replications. It is seen that our method maintains the FDR control and a good power across different values of the correlation $\rho$. 

\begin{figure}[b!]
\centering
\includegraphics[width=3.0in,height=\myfigheight]{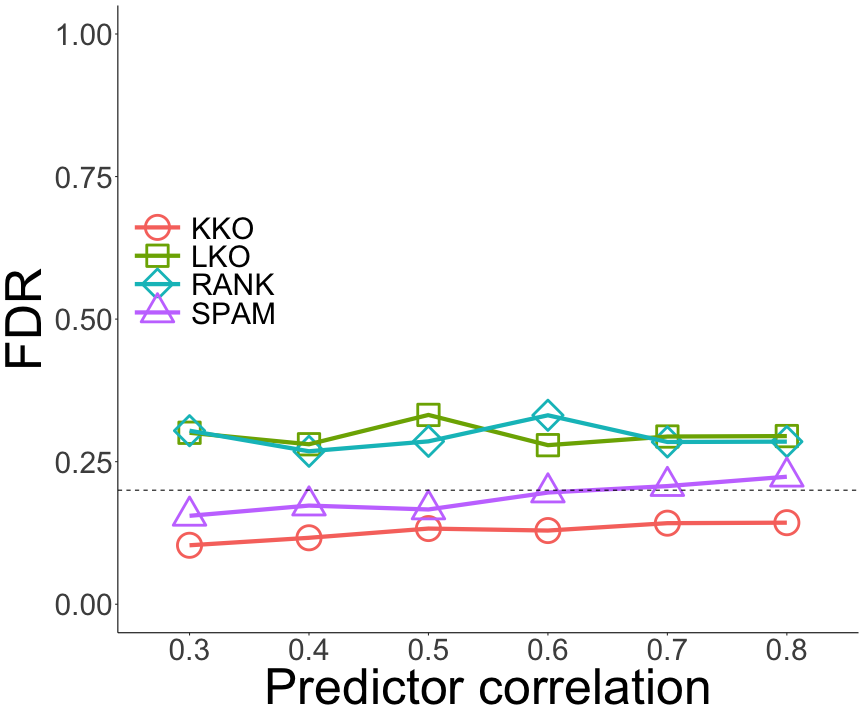}
\includegraphics[width=3.0in,height=\myfigheight]{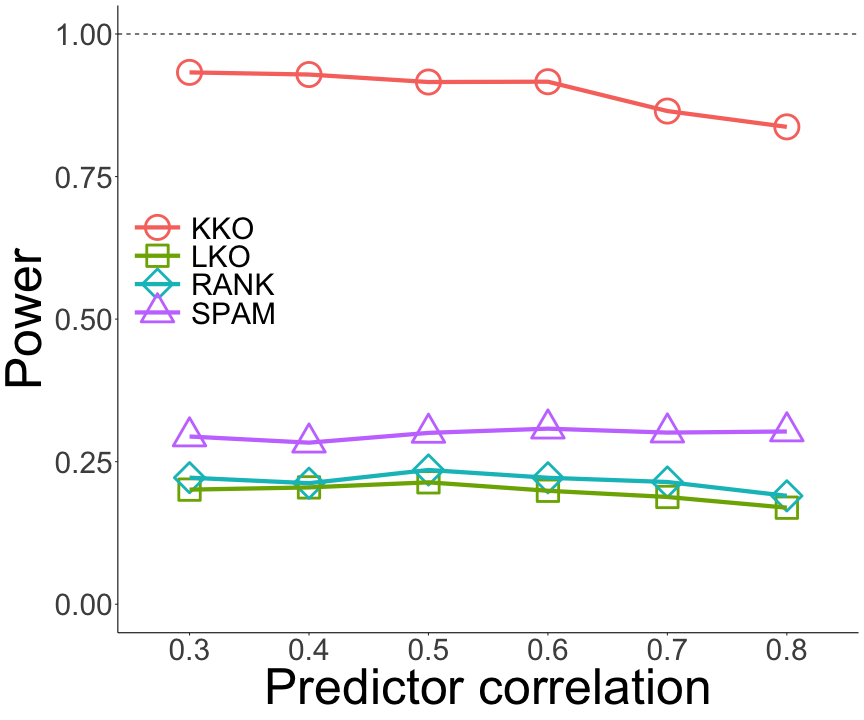}
\caption{Empirical performance and comparison in terms of FDR and power with the varying predictor correlation $\rho$. The same four methods as in Figure \ref{fig: robust_distribution_SNR} are compared.}
\label{fig: dependence}
\end{figure}

%%%%%%%%%%%%%%%%%%%%%%%%%%%%%%%%%%%%%%%%%%%%%%%%%%%
\subsection{Sparsity level}
\label{sec:sparisty-level}

In the power analysis in Section \ref{sec:poweranalysis}, we require that the number of nonzero component functions $|\cS|$ to be bounded in Assumption \ref{assump:complexity}. Next, we carry out a numerical experiment and show empirically that our method still works reasonably well when the number of nonzero components $| \cS |$ increases along with the sample size $n$.  

Specifically, we adopt the same simulation setup as in Section \ref{sec:sensivity-kernel} with the trigonometric polynomial function \eqref{eq: sinpoly}. We vary $(|\cS|, n) = \{ (10, 900), (20, 1800), (30, 2700), (40, 3600) \}$. Figure \ref{fig: robust_sparsity} reports the FDR and power based on 200 data replications. It is seen that our method maintains a reasonably good performance, and still outperforms all the alternative solutions when $|\cS|$ increases along with $n$. 

\begin{figure}[t!]
\centering
\includegraphics[width=3.0in,height=\myfigheight]{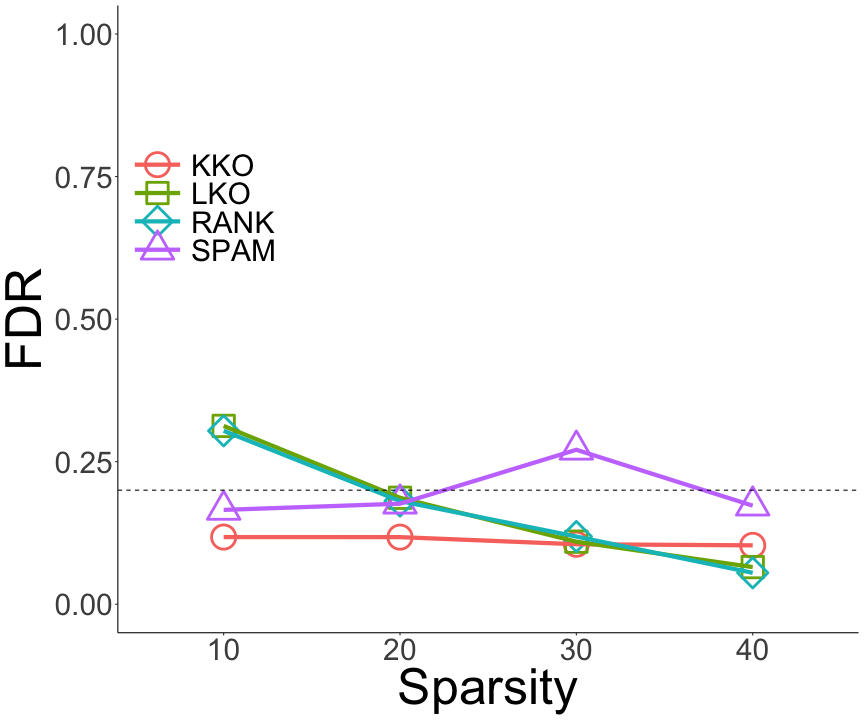}
\includegraphics[width=3.0in,height=\myfigheight]{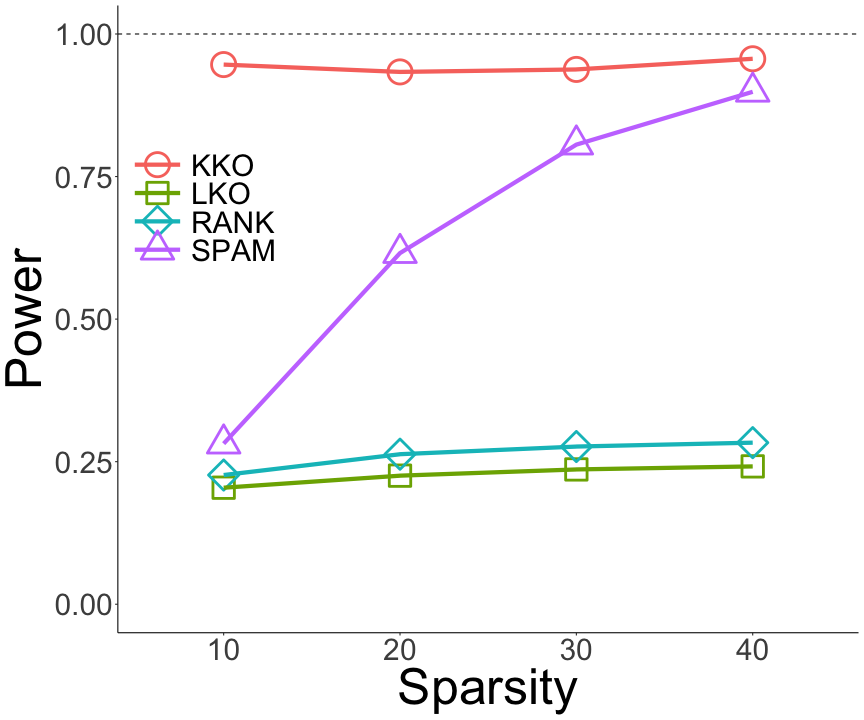}
\caption{Empirical performance and comparison in terms of FDR and power with the varying number of nonzero component functions $|\cS|$. The same four methods as in Figure \ref{fig: robust_distribution_SNR} are compared.}
\label{fig: robust_sparsity}
\end{figure}

%%%%%%%%%%%%%%%%%%%%%%%%%%%%%%%%%%%%%%%%%%%%%%%%%%%
\subsection{Computation parallelization}
\label{sec:parallelization}

We note that the computation of our method can be easily parallelized. This is because it requires no information sharing across different subsamples, and thus Steps 2 to 4 in Algorithm \ref{alg: kernel_KO} can be parallelized. We next carry out a numerical experiment to further investigate the computational acceleration by parallelization.

\begin{figure}[t!]
\centering
\includegraphics[width=3.0in,height=\myfigheight]{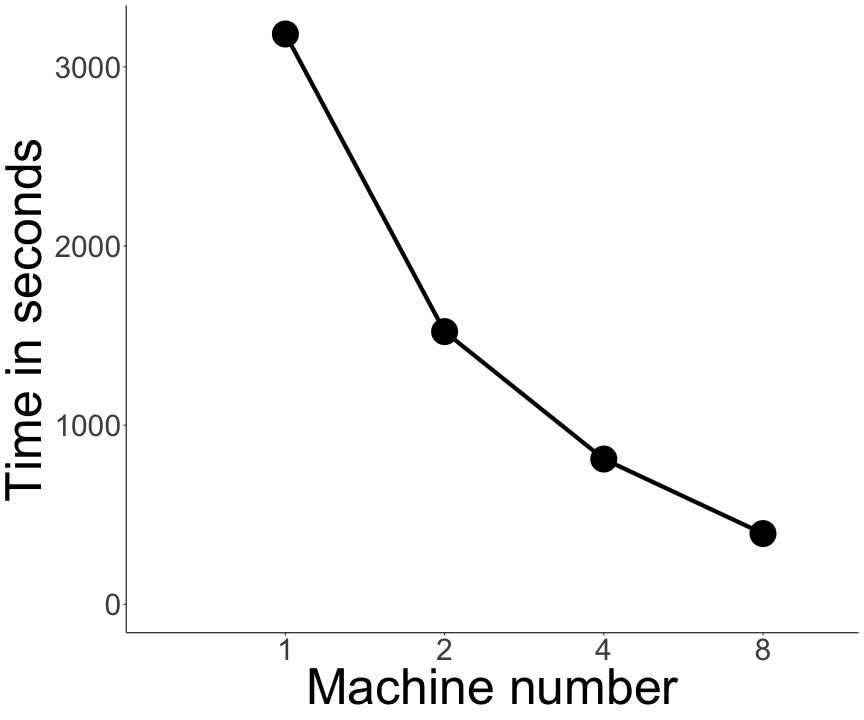}
\caption{Average computation time, in seconds, of Algorithm \ref{alg: kernel_KO} with the varying number of parallel machines. The x-axis is rescaled by base-2 logarithm.}
\label{fig: parallel_effect}
\end{figure} 

Specifically, we adopt the same simulation setup as in Section \ref{sec:sensivity-kernel} with the trigonometric polynomial function \eqref{eq: sinpoly}. We parallelize the $L$ subsamples to a total of $M$ machines using the R package \texttt{foreach}. Each machine independently runs Steps 2 to 4 of Algorithm  \ref{alg: kernel_KO} on the allocated subsamples, then the results are aggregated from all machines to perform Steps 5 to 6 of Algorithm \ref{alg: kernel_KO}. We use the \texttt{c4.4xlarge} instance in Amazon Web Services (AWS). We vary the number of machines $M = \{1,2,4,8\}$. Figure \ref{fig: parallel_effect} reports the average computation time, in seconds, of running the entire Algorithm \ref{alg: kernel_KO} for one data replication, while the results are averaged over 200 data replications. It is seen that  the computation time reduces significantly as the number of machines increases.

%%%%%%%%%%%%%%%%%%%%%%%%%%%%%%%%%%%%%%%%%%%%%%%%%%%
\subsection{Alternative importance scores}
\label{sec:importance-stat}

Finally, we investigate some alternative importance scores studied in \citet{candes2018panning}, and compare them with our importance score in the kernel knockoffs method that is built on the selection probability of the variables and their knockoffs across multiple subsamples.  

Specifically, let $\delta_j$ and $\tilde{\delta}_j$ denote some importance measure of variable $X_j$ and the knockoff variable $\tilde{X}_j$, respectively. \citet[Section 4.1]{candes2018panning} considered three importance scores: 
\begin{eqnarray*}
\textrm{Coefficient difference (CD)}: & & |\delta_j| - |\tilde{\delta}_j |, \\
\textrm{Log-coefficient difference (log-CD)}:    & & \log(|\delta_j|) -\log(|\tilde{\delta}_j|), \\
\textrm{Signed max (SM)}: & & \textrm{sign} (|\delta_j| - |\tilde{\delta}_j | ) \max \left\{ |\delta_j|, |\tilde{\delta}_j | \right\},
\end{eqnarray*}
We adopt the same simulation example as in Section \ref{sec:varysamplesim} with the trigonometric polynomial function \eqref{eq: sinpoly} and the varying signal strength $\theta$. Following \citet{candes2018panning}, we choose the Euclidean norm of the estimated component function as $\delta_j$ for the coefficient difference score, whereas we choose the penalty value at which the variable enters the model as $\delta_j$ for the log-coefficient difference score and the signed max score. Moreover, since \citet{candes2018panning} did not use subsampling for stability, all these scores are computed based on a single random feature expansion. Figure \ref{fig: modelx} reports the FDR and power of the selection results based on these three importance scores as well as our kernel knockoff score for $200$ data replications. It is clearly seen that our importance score performs much better than those alternative scores. In particular, the coefficient difference score and the log-coefficient difference score can not control the FDR below the nominal level. The signed max score can control the FDR, but its power is much lower than the power of our kernel knockoffs method. 

\begin{figure}[t!]
\centering
\includegraphics[width=3.0in,height=\myfigheight]{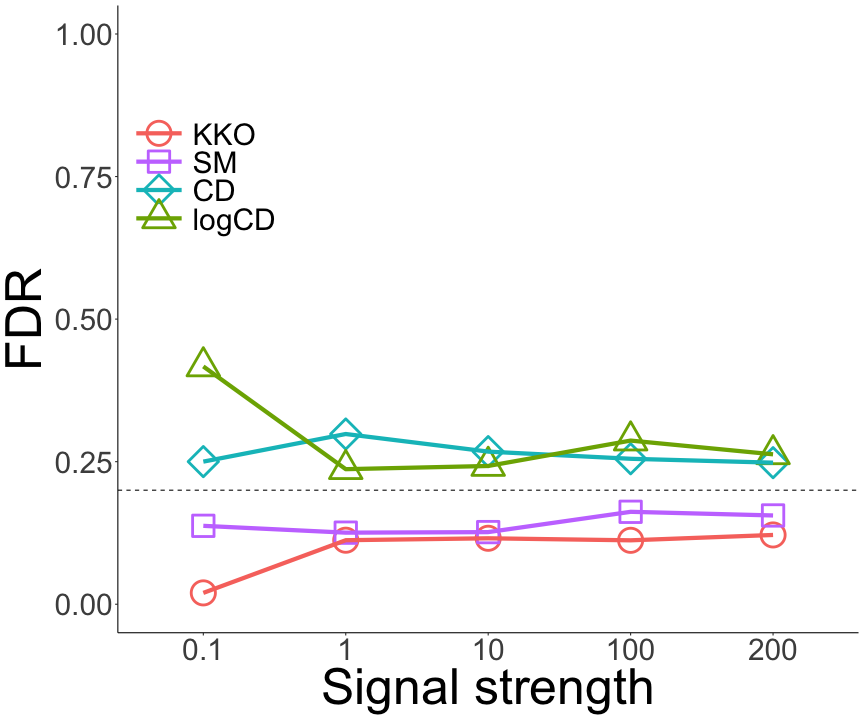}
\includegraphics[width=3.0in,height=\myfigheight]{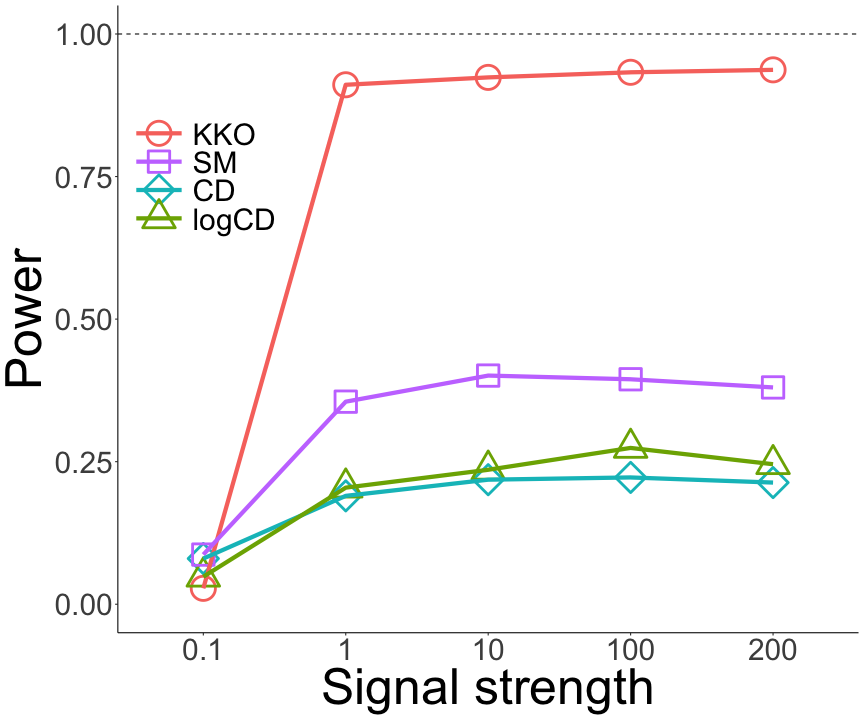}
\caption{Empirical performance and comparison in terms of FDR and power for different importance score functions: our kernel knockoff score (KKO), the coefficient difference score (CD), the log-coefficient difference score (log-CD), and the signed max score (SM).}
\label{fig: modelx}
\end{figure}

\end{document}